\documentclass[]{jfm}

\usepackage{amsmath}
\usepackage{amssymb}
\usepackage{amscd}
\usepackage{graphicx}
\usepackage{newtxtext}
\usepackage{newtxmath}
\usepackage{natbib}
\usepackage{hyperref}
\usepackage{import}
\usepackage{algorithm}
\usepackage{algpseudocode}
\usepackage{array}
\usepackage{multicol}
\usepackage{multirow}
\usepackage{tabularx}
\usepackage{caption}
\usepackage{siunitx}

\hypersetup{
    colorlinks = true,
    urlcolor   = blue,
    citecolor  = blue,
}
\newtheorem{lemma}{Lemma}

\newtheorem{theorem}{Theorem}
\newcommand{\RomanNumeralCaps}[1]
\linenumbers
\newcommand{\refsec}[1]{section \ref{#1}}

\newcommand{\bsym}[1]{\boldsymbol{#1}}
\newcommand{\bu}{\bsym{u}}
\newcommand{\bU}{\bsym{U}}
\newcommand{\be}{\bsym{e}}
\newcommand{\bx}{\bsym{x}}
\newcommand{\Ubulk}{U_{\text{bulk}}}
\newcommand{\utot}{\bsym{u_{\text{tot}}}}
\newcommand{\ptot}{p_{\text{tot}}}
\newcommand{\hbu}{\hat{\bsym{u}}}
\newcommand{\id}{1}
\newcommand{\lx}{a}
\newcommand{\lz}{b}
\newcommand{\tx}{\tau_x}
\newcommand{\tz}{\tau_z}
\newcommand{\txz}{\tau_{xz}}

\newcommand{\sxy}{\sigma_{xy}}
\newcommand{\sx}{\sigma_x}
\newcommand{\sy}{\sigma_y}
\newcommand{\sz}{\sigma_z}
\newcommand{\syz}{\sigma_{yz}}
\newcommand{\sxyz}{\sigma_{xyz}}
\newcommand{\Gppfh}{G_{\text{PPF-}16}}
\newcommand{\Gpcfh}{G_{\text{PCF-}16}}
\newcommand{\Gppf}{G_{\text{PPF}}}
\newcommand{\Gpcf}{G_{\text{PCF}}}
\newcommand{\TW}{\text{TW}}
\newcommand{\rms}{\text{rms}}

\newcommand{\bbU}{\mathbb{U}}
\newcommand{\bbR}{\mathbb{R}}
\newcommand{\bbZ}{\mathbb{Z}}

\newcommand{\bbV}{\mathbb{V}}

\newcommand{\btau}{\overline{\tau}}
\newcommand{\dtr}{\delta_{\text{TR}}}

\title{Symmetry groups and invariant solutions of plane Poiseuille flow}

\author{Pratik P. Aghor\aff{1}
\corresp{\email{paghor3@gatech.edu}}
  \and John F. Gibson\aff{2}
  }

\affiliation{
  \aff{1}School of Earth and Atmospheric Sciences, Georgia Institute of Technology, USA
  \aff{2}Integrated Applied Mathematics Program, Department of Mathematics \& Statistics, University of New Hampshire, USA
}

\begin{document}
\maketitle
%----------------------------------%
\begin{abstract}
Equilibrium, traveling-wave, and periodic-orbit solutions of the Navier-Stokes equations provide a
promising avenue for investigating the structure, dynamics, and statistics of transitional flows.
Many such invariant solutions have been computed for wall-bounded shear flows, including
plane Couette, plane Poiseuille, and pipe flow. However, the organization of invariant solutions 
is not well understood. 
In this paper we focus on the role of symmetries in the organization and computation of invariant
solutions of plane Poiseuille flow. We show that enforcing symmetries while computing invariant 
solutions increases the efficiency of the numerical methods, and that redundancies between search spaces
can be eliminated by consideration of equivalence relations between symmetry subgroups. We determine
all symmetry subgroups of plane Poiseuille flow in a doubly-periodic domain up to translations by half
the periodic lengths and classify the subgroups into equivalence classes, each of which represents
a physically distinct set of symmetries and an associated set of physically distinct invariant solutions. 
We calculate fifteen new traveling waves of plane Poiseuille flow in seven distinct symmetry groups
and discuss their relevance to the dynamics of transitional turbulence. We present a few examples
of subgroups with fractional shifts other than half the periodic lengths and one traveling wave 
solution whose symmetry involves shifts by one-third of the periodic lengths. We conclude with 
a discussion and some open questions about the role of symmetry in the behavior of shear flows.
\end{abstract}
%----------------------------------%

%----------------------------------%
%----------------------------------%
{\bf MSC Codes }  {\it(Optional)} Please enter your MSC Codes here
%----------------------------------%
% \import{sections/}{section-1.tex}
%----------------------------------%
%----------------------------------%
\section{Introduction}
\label{sec:intro}
%----------------------------------%

%\subsection{Motivation}
In the past three decades, computations of invariant solutions 
of the Navier-Stokes equations  have considerably advanced the understanding of subcritical
transition to turbulence in wall-bounded shear flows.
Invariant solutions,  or ``exact coherent structures,'' (\cite{graham2021exact}) have been computed primarily in pipe flow,
plane Couette flow, and plane Poiseuille flow (\cite{graham2021exact} and references therein). 
\cite{nagata1990three} calculated the first known pair of finite amplitude, nonlinear equilibrium
solutions for plane Couette flow using a numerical continuation technique, continuing wavy vortex
flow from the Taylor-Couette to plane Couette conditions. The same solution was independently found by \cite{clever1992three}
and \cite{waleffe2001exact}.
%\cite{waleffe1997self} explained the equilibrium of this fluid state as
%a self-sustaining process in which rolls sustain streaks whose primary instability feeds back to support
%the rolls.
Numerous equilibria and traveling waves have since been found for the shear flows
mentioned above. For example, see \cite{clever1997tertiary}, \cite{waleffe1998three}, \cite{faisst2003traveling},
\cite{wedin2004exact},  \cite{gibson2009equilibrium}. Periodic and relative periodic orbits have also
been computed for plane Couette flow (see \cite{kawahara2001periodic, viswanath2007recurrent, cvitanovic2010geometry}),
plane Poiseuille flow (\cite{toh2003periodic}), and pipe flow (\cite{duguet2008transition}). 
Relatively fewer solutions
have been computed for plane Poiseuille flow (see \cite{ehrenstein1991three}, \cite{waleffe1998three},
\cite{toh2003periodic}, \cite{waleffe2003homotopy}, \cite{gibson2014spanwise}, \cite{zammert2015crisis},
and \cite{park2015exact}).

Another line of work has been to study transitional turbulence in large periodic domains
(see \cite{tuckerman2020patterns} and references therein). In large domains, localized puffs
in the case of pipe flow ( \cite{avila2011onset}) and bands in the cases of plane Couette 
(\cite{tuckerman2011patterns}) and plane Poiseuille flow (\cite{tuckerman2014turbulent})
are known to be important structures that facilitate the transition to turbulence. In plane
Couette and Poiseuille flows, the alternating laminar and turbulent bands are often found to
be arranged at $24^{\circ}$ to the streamwise direction. These tilted transitional bands have
been studied in small doubly-periodic domains with axes tilted with respect to the streamwise
(\cite{tuckerman2011patterns}, \cite{tuckerman2014turbulent}). The structures seen in such
simulations, e.g. puffs in the case of pipe flow or bands in planar flows, might not be universal,
as was recently shown in \cite{gome2024phase}. Some phenomenological models have been proposed
for the transition to turbulence in pipe flows, based on excitable media (\cite{barkley2016theoretical})
or stochastic ecological models (\cite{shih2016ecological},\cite{wang2022stochastic}). Though
the present study is primarily focused on periodic domains aligned with the streamwise forcing, 
extension to tilted domains is possible and is briefly discussed in section \ref{sec:extension_to_tilted_domains}.

The importance of invariant solutions stems from their role in the dynamical-systems perspective on 
transitional turbulence.
This view considers the dynamics of turbulence in an infinite-dimensional state space, where each
point represents an entire fluid velocity field, and where state-space dynamics are given by the
Navier-Stokes equations. In this view, an equilibrium solution is an equilibrium point in state space,
and a periodic orbit is a closed loop. Equilibria and periodic orbits are thus zero- and
one-dimensional invariant solutions of the governing equations; i.e. each is a subset of state
space that remains unchanged under the evolution of dynamics. Since individual states in an
invariant solution repeat themselves in finite time under the dynamics, the linear stability of
solutions can be computed by numerical integration of small perturbations \citep{viswanath2007recurrent}.
Invariant solutions for transitional flows typically have very few unstable eigenmodes, 
a moderate number of weakly stable modes and a very large number of strongly stable
modes, so that the dynamics in the neighborhood of a solution is inherently low-dimensional. 
Further, we expect a generic dynamical system to have finitely many equilibria (up to equivalence
under coordinate transformations), and a generic chaotic system to have a countable and dense 
set of periodic orbits in its chaotic attractor or repeller. The combination of linearization 
and countability leads to periodic orbit theory, in which averages over chaotic sets are 
computed from the linearizations of finite sets of periodic orbits (\cite{cvitanovic2005chaos}). 
The hope, then, for the study of invariant solutions of wall-bounded shear flows is to develop 
quantitatively accurate, dynamical, spatio-temporal model of transitional turbulence 
based on computations of invariant solutions and their linearizations. Short of that, the
equilibrium, traveling-wave, and periodic-orbit solutions of a flow serve as precise
representative examples of the flow's complex, far-from-laminar spatiotemporal dynamics.
\citep{gibson2008visualizing, budanur2017relative}. 

Invariant solutions often satisfy a set of symmetries inherent to the equations of motion.
The set of symmetries determines which types of solutions are possible 
and which solution parameters should be treated as free variables in a search. 
For example, in plane Couette flow, a parametrically robust equilibrium must have reflection or
shift-refect symmetries in both the spanwise and streamwise directions in order to fix the
spatial phase of the velocity field. In the absence of such reflection symmetries, we expect to find
traveling waves rather than equilibria, and the wavespeed must be included as an unknown
in the search space. Further, symmetries can be leveraged to increase the efficiency of the
computational search for invariant solutions, if they are specified prior to the search 
and imposed as constraints on the search space. Imposing a second-order symmetry reduces
the dimensionality of the search space by a factor of two, leading to a corresponding speedup
in the iterative search algorithms used to compute solutions. For both these reasons, it is
advantageous to specify the set of symmetries of a solution as a kind of control parameter
for its computation. In practice, to find an invariant solution, we first specify a symmetry
group, next determine which type of invariant solution the symmetry group allows, and
then search in the constrained subspace for solutions of the appropriate form. 

Knowledge of the subgroups of the general symmetry group of a flow is thus crucial for computing
the invariant solutions of that flow. The subgroups of plane Couette flow were partially
enumerated and classified in \cite{gibson2009equilibrium}, with an emphasis on subgroups that
support equilibrium solutions. In this paper, we partially classify the subgroups of plane Poiseuille
flow in doubly-periodic rectangular domains. Our main results are (1) determination of all symmetry
subgroups of rectangular, periodic plane Poiseuille flow up to phase shifts of half the spanwise and
steamwise periodic lengths, (2) classification of these subgroups into equivalence classes of subgroups
that are the same up to coordinate transformations, and (3) computation of new traveling waves in several
previously explored and some hitherto unexplored symmetry subgroups. We find that doubly-periodic 
plane Poiseuille flow has thirty-nine nontrivial subgroups up to half-box shifts, and that these 
thirty-nine subgroups fall into twenty-four equivalence classes with distinct dynamical behavior. 
We also expand on the analysis of plane Couette subgroups in \cite{gibson2009equilibrium} by computing
and classifying subgroups that support traveling waves and subgroups with diagonal phase shifts.

%%%%%%%%%%%%%%%%%%%%%%%%%%%%%%%%%%%%%%%%%%%%%%%%%%%%%
% governing equations
\subsection{The governing equations and invariant solutions} \label{sec:governing_eqns}

We focus on plane Poiseuille flow driven by a constant bulk velocity $\Ubulk$ in the streamwise direction and zero bulk 
velocity in the spanwise direction. We define the Reynolds number as $Re = U_c h/\nu$, where $h$ is the
channel half-height, $U_c = 3 \Ubulk/2$ is the centerline velocity of the parabolic laminar flow with
the same streamwise bulk velocity, and $\nu$ is the kinematic viscosity of the fluid. After 
nondimensionalization the equations of motion are
\begin{align}\label{eqn:gov_eqns_tot}
 \begin{split}
  \frac{\partial \utot}{\partial t} + \utot \cdot \nabla \utot & =  - \nabla \ptot + \frac{1}{Re} \nabla^{2} \utot , \\
  \nabla \cdot \utot &= 0.
 \end{split}
\end{align}
We denote the spatial domain of (\ref{eqn:gov_eqns}) by $\Omega$ and the spatial coordinates 
by $\bx = (x,y,z) \in \Omega$, where $x$ is the streamwise direction, $y$ is wall-normal, and $z$ is spanwise.
%consistent with \cite{waleffe2001exact}, \cite{toh2003periodic}, and \cite{park2015exact}.
%and differs from \cite{nagata2013mirror}. 
We take $\Omega$ to be doubly periodic with periodic lengths $L_x$ and $L_z$, so that
$\Omega = [0, L_x) \times [-1, 1] \times [0, L_z)$, where $[0,L)$ denotes the periodic interval
of length $L$. 

The total velocity and pressure fields are decomposed into a sum of a base flow and a fluctuation,
$\utot =\bU + \bu, \ptot = P + p$, where the base velocity $\bU$ is the nondimensionalized laminar flow
solution $\bU(y) = (1-y^{2}) \be_x$, for which $U_c = 1$ and $\Ubulk = 2/3$. The boundary conditions
on $\bu$ are periodic in $x$ and $z$ and $\bu = 0$ at the walls $y = \pm 1$. 
The base pressure takes the form $P(t) = x P_x(t) + z P_z (t)$ giving the total pressure gradient the form
$\nabla \ptot(\bx, t) = P_x(t) \, \be_x + P_z(t) \, \be_z + \nabla p(\bx,t)$, where 
$\be_x, \be_z$ are unit vectors in the $x$ and $z$ directions. We constrain $\nabla p(\bx,t)$ to have
zero spatial mean so that $P_x(t)$ and $P_z(t)$ are uniquely determined and adjust dynamically to balance
the stream- and spanwise bulk velocity constraints. In terms of the spatial average operator
\begin{equation}
 \langle \cdot \rangle_{\Omega} = \frac{1}{\text{vol}(\Omega)}  \int_{\Omega} \cdot \, d \Omega,
\end{equation}
the constraint for the pressure decomposition is
$\langle \nabla p \rangle_{\Omega} = 0$. The base flow $\bU(y)$ carries the bulk velocity,
$\langle \utot \rangle_{\Omega} = \langle \bU \rangle_{\Omega} = \Ubulk \, \be_x$, so that
$\langle \bu \rangle_{\Omega} = 0$.
With this decomposition \ref{eqn:gov_eqns_tot} becomes
\begin{align}\label{eqn:gov_eqns}
 \begin{split}
  \frac{\partial \bu}{\partial t} + \bu \cdot \nabla \bu 
     &= -\nabla p + \frac{1}{Re} \nabla^2 \bu 
       + \left(\frac{2}{Re} - P_x \right) \be_x - P_z \be_z - 2 y v \be_y - (1-y^2) \frac{\partial \bu}{\partial x}, \\
  \nabla \cdot \bu &= 0.
 \end{split}
\end{align}
Here $v$ is the wall-normal component of $\bu$; i.e. $\bu(\bx,t) = [u,v,w](x,y,z,t)$.
Henceforth we refer to the fluctuations $\bu$ and $p$ as ``velocity'' and ``pressure.''
Since the total pressure terms $- p - P_x \be_x - P_z \be_z$ act as Lagrange multipliers,
determined instantaneously to meet the incompressibility and bulk velocity constraints,
we can view the system (\ref{eqn:gov_eqns}) as defining $\partial \bu/\partial t$ as a function
of $\bu$ alone. We represent this infinite-dimensional dynamical system abstractly as
\begin{align}\label{eqn:dynamics}
  \begin{split}
  \frac{\partial \bu}{\partial t} &= f(\bu), \\
  \bu(t) &= \phi^t(\bu(0)).\\
 \end{split}
\end{align}
for $t \geq 0$. These equations represent the infinitesimal and finite-time dynamics of (\ref{eqn:gov_eqns})
with the given boundary conditions and bulk constraint. We take the space of solutions of (\ref{eqn:gov_eqns}) 
to be the set of square-integrable, divergence-free velocity fields on $\Omega$ with Dirichlet boundary 
conditions at the walls, 
\begin{align}\label{eqn:uspace}
  \bbU = \{\bu \in L^2(\Omega) : {\bf \nabla} \cdot \bu = 0, \; \bu|_{y=\pm 1} = 0\},
\end{align}
where the $L^2$ norm $\|\cdot\|$ is defined by $\|\bu\|^2 = \langle \bu \cdot \bu \rangle_{\Omega}$. 
Conveniently, $\bbU$ is a vector space. 

%%%%%%%%%%%%%%%%%%%%%%%%%%%%%%%%%%%%%%%%%%%%%%%%%%%%%
% invariant solutions
%\subsection{Invariant solutions} \label{sec:invariant_solutions}

An invariant solution of plane Poiseuille flow is a velocity field $\bu \in \bbU$ satisfying
\begin{align}\label{eqn:invariance}
  \sigma \phi^t(\bu) - \bu = 0
\end{align}
for some time $t>0$ and some symmetry $\sigma$ of the equations of motion (see \refsec{sec:symmetries}). 
Equilibria satisfy (\ref{eqn:invariance}) for all $t > 0$ and $\sigma = \id$. Traveling waves satisfy
(\ref{eqn:invariance}) for all $t > 0$ and some $\sigma(t) = \tau(c_x t, c_z t)$, a phase shift that
increases linearly in time by some fixed wavespeeds $c_x$ and $c_z$, one or both nonzero (see (\ref{def:symmetries_PPF})).
Periodic orbits satisfy (\ref{eqn:invariance}) with $\sigma=\id$ and $t = T$ for a fixed $T  > 0$, with
$\phi^t(\bu) \neq \bu$ for $0 < t < T$. Relative periodic orbits satisfy equation
(\ref{eqn:invariance}) with the same temporal conditions as periodic orbits but with
$\sigma \neq \id$. Invariant solutions are computed by solving
(\ref{eqn:invariance}) in discretized form from initial guesses $\hat{\bu}$ and 
potentially $\hat{\sigma}, \hat{T}$ that approximately satisfy the discretized equations.
For simple geometries such as doubly-periodic rectangular boxes and low to moderate Reynolds
numbers, discretizing $\bu$ with a spectral representation and $\phi^t$ with finite-difference
time-stepping results in a system of $10^4$ to $10^6$ equations on the same number of unknown
spectral coefficients. Depending on the type of invariant solutions, the system can include
additional unknowns parameters such as $c_x, c_z$ or $T$ and additional constraint equations
related to these variables. The resulting system of nonlinear discretized equations can
be solved efficiently using trust-region Newton-Krylov methods \citep{viswanath2007recurrent}.
%%%%%%%%%%%%%%%%%%%%%%%%%%%%%%%%%%%%%%%%%%%%%%%%%%%%%
% symmetries
\subsection{Symmetries, equivariance, and invariant subspaces} \label{sec:symmetries}

The symmetries $\sigma$ allowed in (\ref{eqn:invariance}) are given by the plane Poiseuille
symmetry group $\Gppf$. This group is generated by two discrete symmetries corresponding to
reflections about the wall-normal and spanwise axes $y$ and $z$, and two continuous
symmetries corresponding to translations in the streamwise and spanwise directions $x$ and $z$.
We represent these symmetries by operators $\sy, \sz,$ and $\tau(\lx,\lz)$, whose actions on
velocity fields $\bu(\bx) = [u,v,w](x,y,z)$ are 
\begin{align}\label{def:symmetries_PPF}
 \begin{split}
  \sy[u, v, w](x, y, z) &= [u, -v, w](x, -y, z),\\
  \sz[u, v, w](x, y, z) &= [u, v, -w](x, y, -z), \\
  \tau (\lx, \lz) [u, v, w](x, y, z) &= [u, v, w](x + \lx, y, z + \lz),
 \end{split}
\end{align}
where $\lx,\lz \in \bbR$ represent continuous phase shifts. (These definitions are in accordance with
\cite{gibson2014spanwise} and differ from the definitions in \cite{gibson2009equilibrium} which were
specialized to plane Couette flow.) The action of a symmetry can be interpreted as either as mapping one velocity
field to another in a fixed coordinate system, or a change in the representation of a given velocity
field due to a change of coordinates. 

The plane Poiseuille symmetry group $\Gppf$ is then defined as the set of all possible products
of the generators $\sy$, $\sz,$ and $\tau(a,b)$,
\begin{align} \label{def:Gppf}
  \Gppf = \langle \sy,\, \sz,\, \{\tau(\lx, \lz) : a,b \in \bbR \} \rangle,
\end{align}
with group multiplication determined by (\ref{def:symmetries_PPF}) and substitution. For example,
the product $\sy \tau(a,b)$ is given by its action $(\sy \tau(a,b)) [u,v,w](x,y,z)) = [u, -v, w](x+a, -y, z+b)$.
From (\ref{def:symmetries_PPF}) we can deduce the multiplication rules 
\begin{align} \label{eqn:multiplication_rules}
  \begin{split}
    \sy^2 &= \sz^2 = \id, \\
    \sy \sz &= \sz \sy, \\
  \tau(\lx, \lz) \, \sy &= \sy \, \tau(\lx, \lz), \\
  \tau(\lx, \lz) \, \sz &= \sz \, \tau(\lx, \, -\lz), \\
  \tau(a_1, b_1) \, \tau(a_2, b_2) &= \tau(a_1+a_2, \, b_1+ b_2).
  \end{split}
\end{align}
The inverses of $\sy, \sz,$ and $\tau(a,b)$ are thus $\sigma_y^{-1} = \sigma_y$,
$\sigma_z^{-1} = \sigma_z$, and $\tau^{-1}(a,b) = \tau(-a, -b)$. We use concatenation
of subscripts to indicate multiplication of group elements, e.g. $\syz = \sy\sz$.
The multiplication rules and concatenation notation allow any element of $\Gppf$
to be written as the product of a single reflection symmetry (one of $1, \sy, \sz,$ or $\syz$)
and a single translation $\tau(a,b)$. We will express such products in either order
depending on context and typographical considerations.

%%%%%%%%%%%%%%%%%%%%%%%%%%%%%%%%%%%%%%%%%%%%%%%%%%%%%
% equivariance and subgroups
% \subsection{Equivariance} \label{sec:equivariance}

Plane Poiseuille flow is {\em equivariant} in $\Gppf$, meaning that all symmetries 
$\gamma \in \Gppf$ commute with the dynamics $f$ and $\phi^t$,
\begin{align} \label{eqn:equivariance}
  \begin{split}
    f(\gamma \bu) &= \gamma f(\bu), \\
    \phi^t(\gamma \bu) &= \gamma \phi^t(\bu),
  \end{split}
\end{align}
for all $\bu \in \bbU$ and all $t\geq0$. In terms of physics, equivariance means
that the equations of motion are preserved under any coordinate transformation described
by the group. Given a solution $\bu(t)$ of (\ref{eqn:dynamics}) and some coordinate transformation
$\gamma \in \Gppf$, let $\bu'(t) = \gamma \bu(t)$. Then
\begin{align} \label{eqn:solution_preservation}
\bu'(t) = \gamma \bu(t) = \gamma \phi^t(\bu(0)) = \phi^t(\gamma \bu(0)) = \phi^t(\bu'(0)),
\end{align}
and thus $\bu'(t)$ is also a solution of (\ref{eqn:dynamics}). The dynamics $\phi^t$
are thus preserved under the coordinate transformation $\gamma : \bu(t) \rightarrow \bu'(t)$. 
The relations in (\ref{eqn:solution_preservation}) are conveyed by the commutation diagram
\begin{align} \label{eqn:commutation_diagram}
  \begin{CD}
\bu(0) @> \phi^t >> \bu(t) \\
@VV \gamma V @VV \gamma V\\
\bu'(0) @> \phi^t >> \bu'(t), \\
\end{CD}
\end{align}
which shows how the dynamics $\phi^t$ maps initial conditions $\bu(0)$ and $\bu'(0)$
to future states $\bu(t)$ and $\bu'(t)$, and how a symmetry $\gamma$ maps a given 
trajectory $\bu(t)$ into a distinct but dynamically equivalent trajectory 
$\bu'(t) = \gamma \bu(t)$. It is worth emphasizing that the actions of symmetries
$\gamma: \bu(t) \rightarrow \bu'(t)$ are vastly simpler than the action of dynamics
$\phi^t : \bu(0) \rightarrow \bu(t)$: the former are simple, invertible coordinate 
transformations, whereas the latter is time integration of the Navier-Stokes equations.
In this light, the commutation diagram shows that the complex time evolution of the
trajectory $\bu'(t)$ is just a mirror image of the complex time evolution of $\bu(t)$,
under some symmetric ``mirror'' $\gamma$.

A second important property of equivariance is the preservation of the symmetries of 
states under dynamics. If $\bu(0)$ is a $\gamma$-symmetric initial condition, i.e.\ an
initial fluid state that satisfies $\bu(0) = \gamma \bu(0)$ for some $\gamma \in \Gppf$, then 
\begin{align} \label{eqn:symmetry_preservation}
\bu(t) = \phi^t(\bu(0)) = \phi^t(\gamma \bu(0)) = \gamma \phi^t(\bu(0)) = \gamma \bu(t),
\end{align}
so that the entire trajectory $\bu(t)$ for $t \geq 0$ is $\gamma$-symmetric as well.
The preservation of symmetry under dynamics also holds for subgroups of $\Gppf$.
Let $H$ be a subgroup of $\Gppf$, and suppose $\bu(0)$ is $h$-symmetric
for all $h \in H$. Then $\bu(t)$ is $h$-symmetric for all $h \in H$
and all $t\geq0$. A subgroup $H$ of $\Gppf$ thus defines a dynamically invariant
subspace of velocity fields,
\begin{align} \label{eqn:invariant_subspace}
\bbU_H = \{\bu \in \bbU : \bu = h \bu \; \forall h \in H\},
\end{align}
with $\bu(0) \in \bbU_H$ implying $\bu(t) \in \bbU_H$ for all $t \geq 0$. 
% [Instead of speaking of $\bu \in L^2(\Omega)$, should we use the space of zero-div
% fields that satisfy the BCs? Probably. Ask Marianna how to handle \& notate this.]
Since the symmetries in $\Gppf$ act linearly on $\bu$, each symmetric subspace $\bbU_H$
is a linear subspace of $\bbU$, in that any linear combination of states in $\bbU_H$
remains $\bbU_H$.

%%%%%%%%%%%%%%%%%%%%%%%%%%%%%%%%%%%%%%%%%%%%%%%%%%%%%
% subgroups and invariant solutions
\subsection{Leveraging symmetry in computation of invariant solutions}
\label{sec:preliminary} % old tag {sec:leveraging}

The relations between the equations of motion, symmetries, and invariant solutions
have important practical consequences for the computation of invariant solutions.
Firstly, the symmetries of a given subgroup $H$ determine which kinds of invariant solutions
are allowed within the associated symmetric subspace $\bbU_H$. For example, subgroups
and subspaces with $\sz$ symmetry do not support traveling waves that travel in the
spanwise direction since the restriction $[u,v,w](x,y,z) = [u,v,-w](x,y,-z)$ requires
all $z$-varying velocity fields in $\bbU_H$ to have $u,v$ even and $w$ odd about the fixed
$z=0$ plane. 

Secondly, the computational cost of finding an invariant solution can be significantly
reduced by enforcing known symmetry constraints on the search space. The dominant cost of
computing an invariant solution comes from the many fixed-length time integrations of small
velocity perturbations in the iterative Krylov-subspace solution to the Newton-step equation.
Enforcing known symmetries on all velocity fields reduces the effective dimension
of the space that the Krylov algorithm explores, so that fewer iterations and fewer 
time integrations of perturbations are needed to converge on a solution to a given accuracy.
For example, velocity fields with $\sy$ symmetry satisfy $[u,v,w](x,y,z) = [u,-v,w](x,-y,z)$,
and these even/odd symmetries in $y$ imply that half the coefficients of the Chebyshev expansions
of $[u,v,w]$ in $y$ are zero. Given the variety of possible symmetry constraints and the
general complexity of the algebraic relations they impose on the coefficients, it is impractical
to develop specialized time-integration codes that operate on the reduced set of spectral coefficients for
a given set of symmetries. However, arbitrary symmetries can be imposed on the full-coefficient
search space by projection operations $\bu \rightarrow (\bu + \gamma \bu)/2$ for second-order
$\gamma$ (or  $\bu \rightarrow (\bu + \gamma \bu + \gamma^2 \bu \ldots + \gamma^{m-1} \bu)/m$ for $m$th-order
$\gamma$). Such projection-based symmetry enforcement is commonly used in the computation of 
invariant solutions \citep{willis2013revealing,gibson2009equilibrium}, and available as functionality in
widely-used codes such as Openpipeflow \citep{willis2017openpipeflow} and Channelflow 
\citep{gibson2024channelflow}.

Each independent second-order symmetry restriction halves the effective dimensionality
of the search space and roughly doubles the convergence rate of the Krylov-subspace solution
of the Newton-step equation, so that enforcing $n$ independent symmetry constraints for a
subgroup $H$ with $n$ second-order generators decreases the overall computational cost of
the search by a factor of $2^n$. 
This halving of effective dimensionality has group-theoretic origin. Each symmetry generator
$\gamma$ in this case is order two, making the group $\{\id, \gamma\}$ isomorphic to the group
$\bbZ_2$. There are two one-dimensional irreducible representations of $\bbZ_2$, the trivial
representation defined by $R_{\textrm{trivial}}(g) = 1$ and the so-called 'sign' representation
defined by $R_{\textrm{sign}}(g) = -1$, for $g \in \langle \gamma \rangle$.  The isotypic
decomposition of the vector space on which these representations act can be written as
$\bbV = \bbV_{\textrm{trivial}} \oplus \bbV_{\textrm{sign}}$, with 
$\bbV_{\textrm{trivial}} = \{f: \gamma f(\bsym{x}) = f(\bsym{x})\}$ and $\bbV_{\textrm{sign}} = \{f: \gamma f(\bsym{x}) = -f(\bsym{x})\}$.
In a symmetry-adapted basis, the Jacobian of the Newton-Krylov iteration takes a
block-diagonal form with two block diagonals that decouple the even/odd modes and effectively
halve the dimension of the search space. 
This procedure can be generalized for other symmetries of order $3$ and higher.

\begin{figure}  
 \begin{center}
   \includegraphics[width=0.8\textwidth]{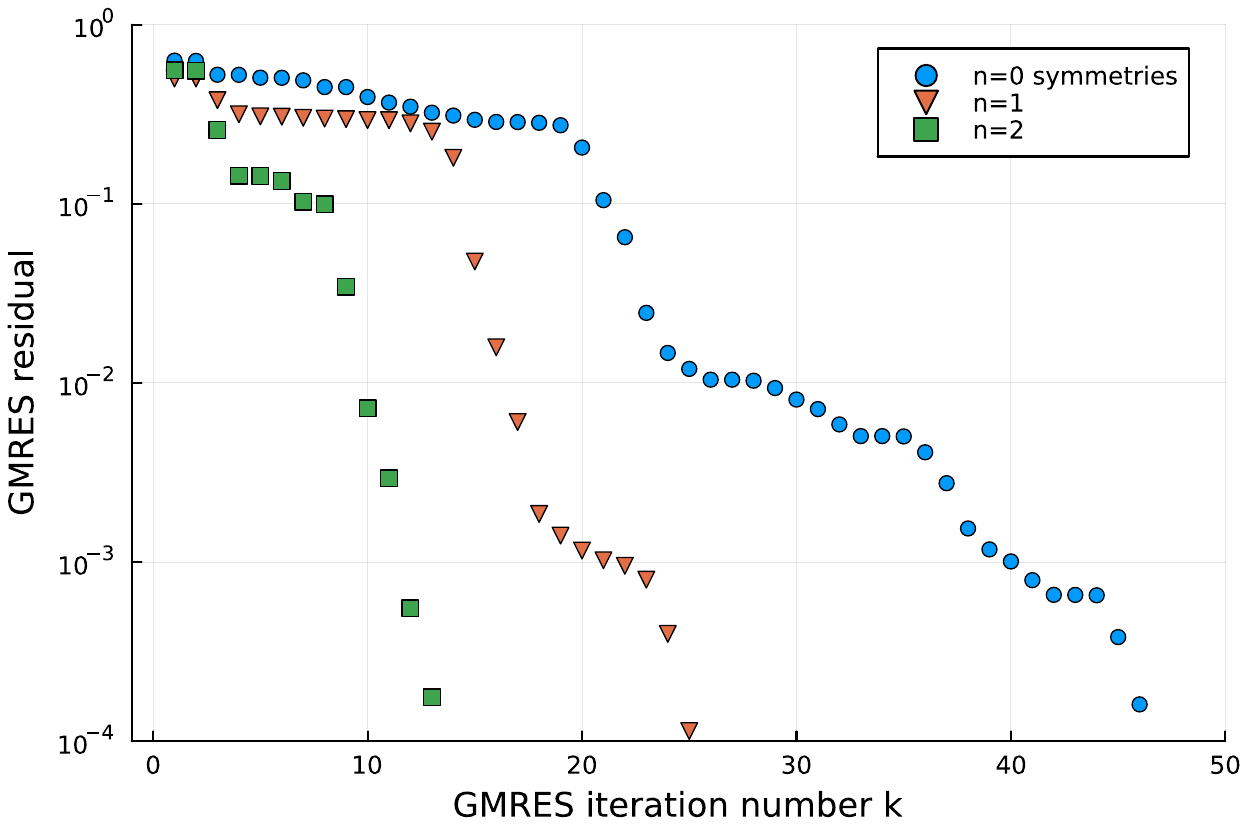}
 \end{center}
 \caption{{\bf Effect of symmetry restrictions on the convergence of Krylov-subspace methods.}
The figure shows the normalized residual $\|Dg|_{\hbu}\; \Delta \hbu^{(k)} + g(\hbu)\|/\|g(\hbu)\|$
of the Newton-step equation versus GMRES iteration number $k$ while enforcing $n=0$, 1, and 2
independent second-order symmetry constraints during a computation of the Nagata lower-branch
equilibrium of plane Couette flow. The solution was computed for $\Rey=400$ and $L_x, L_z = 2\pi, \pi$ 
on a $48 \times 49 \times 48$ grid with 2/3-style dealiasing, and 3rd-order semi-implicit backwards
differentiation time stepping.}
 \label{fig:gmres_convergence}
\end{figure}

Figure \ref{fig:gmres_convergence} illustrates this $2^n$ speed-up in the computation of
the upper-branch Nagata equilibrium of plane Couette flow. The plot shows the convergence
of the GMRES algorithm with $n=0,1,$ and 2 symmetries imposed.
%%%%%%%%%%%%
The invariance equation for this solution is $g(\bu) = \phi^T(\bu) - \bu = 0$
(eqn. \ref{eqn:invariance} with $\sigma = 1$,  $t$ set to a fixed value $T$,
and plane Couette conditions rather than plane Poiseuille). Given an approximate
solution $\hbu$ of $g(\bu) = 0$, the Newton step is the solution $\Delta \hbu$
of the linear system $Dg|_{\hbu}\; \Delta \hbu = - g(\hbu)$. For the given discretization
with $32 \times 49 \times 32 \times 2$ independent spectral coefficients, the discrete representations
of $\hbu$ and $\Delta \hbu$ are vectors of about $10^5$ free variables, and $Dg|_{\hbu}$ is a
$10^5 \times 10^5$ matrix, the derivative of each component of the discretized equation 
$g(\bu)$ with respect to each spectral coefficient in $\bu$ evaluated at $\bu = \hbu$.
An explicit representation of $Dg$ is not required for matrix-free Krylov methods;
instead it is sufficient to calculate the action of $Dg|_{\hbu}$ on test vectors $\Delta \hbu$
using finite differencing,
$Dg|_{\hbu} \; \Delta \hbu = (g(\hbu + \epsilon \Delta \hbu) - g(\hbu)) /\epsilon + O(\epsilon)$.
%%%%%%%%%%%%
The initial guess for the search was constructed from a smooth, divergence-free, no-slip
perturbation of the previously-computed Nagata equilibrium, with magnitude 0.02 relative
to the unperturbed solution in $L_2$ norm. This perturbed velocity field was integrated forward
in time for ten outer time units in order to eliminate any artificiality in its spectral
characteristics and then projected into the appropriate symmetric subspace to form the
initial guess $\hbu$. 
%%%%%%%%%%%%
The symmetry group of the Nagata solution in the chosen spatial phase is
$H = \{1, \sxy \tz, \sz \tx, \sxyz \txz \}$ (see eqn.\ \ref{eqn:halfbox_shifts}).
From these elements we select two generators $h_1 = \sxy \tz$ and $h_2 = \sz \tx$.
The figure shows GMRES convergence of the same Newton step with no imposed symmetry
($n=0$), with $h_1$ imposed ($n=1$), and with $h_1, h_2$ imposed ($n=2)$. Symmetries
were imposed by projection $\bu \rightarrow (\bu + h_i \bu)/2$ at intervals $\Delta T=1$.
%%%%%%%%%%%%

The convergence of GMRES is measured by the normalized residual of the Newton-step equation,
namely $\|Dg|_{\hbu} \, \Delta \hbu + g(\hbu)\|/\|g(\hbu)\|$.
%%%%%%%%%%%%
The figure shows that for any given level of tolerance in the residual, GMRES
converges about twice as quickly with one symmetry enforced, and four times as quickly 
with two symmetries enforced, confirming the expected $2^n$ improvement.

%%%%%%%%%%%%%%%%%%%%%%%%%%%%%%%%%%%%%%%%%%%%%%%%%%%%%
% equivalence classes
\subsection{Equivalent subgroups and equivalence classes}
\label{sec:equivalence}

The efficiency gains outlined in section \ref{sec:preliminary} require that the generators 
of the symmetry subgroup are specified beforehand as control parameters for the numerical search. 
We are thus lead to the questions
\begin{itemize}
\item  For a given flow and general symmetry group $G$, what are the subgroups $H$ of $G$?
\item  How are those subgroups specified in terms of generators? 
\item  Which subgroups of $G$ are dynamically distinct, and which are dynamically equivalent?
\end{itemize}
The main aim of this paper is to answer these questions partially in the context of plane Poiseuille flow. 
In this section, we use the language of equivalence relations and equivalence classes to formalize  
the idea of dynamical equivalence and extend it to symmetry subgroups and their associated symmetric
subspaces. 

For fluid states, we define two velocity fields $\bu$ and $\bu'$ to be {\em equivalent} or 
if there is a symmetry $\gamma \in \Gppf$ such that $\bu' = \gamma \bu$.
Due to the properties of symmetries as a group, this definition of equivalence is reflexive, symmetric,
and transitive, and so forms an {\em equivalence relation} on velocity fields. We follow common practice
and indicate equivalence with the symbol $\sim$, i.e. $\bu \sim \bu'$ means $\bu$ and $\bu'$ are equivalent
by the above definition. The {\em equivalence class} $\overline{\bu}$ of a given velocity field $\bu$ is
then defined as the set of all velocity fields equivalent to $\bu$ under the given equivalence relation,
i.e.
\begin{align}
  \overline{\bu} &= \{\bu' \in \bbU : \bu' \sim \bu\} 
\intertext{or equivalently,}
  \overline{\bu} &= \{\gamma \bu : \gamma \in \Gppf \}.  
\end{align}
For example, the equivalence class of a given traveling-wave solution is the same traveling
wave transformed by all possible combinations of phase shifts and $y$ and $z$ reflections. 
Similarly, we define two trajectories $\bu(t)$ and $\bu'(t)$ to be equivalent if there
is a $\gamma \in \Gppf$ such that $\bu'(t) = \gamma \bu(t)$. This defines an equivalence
relation on trajectories whose equivalence classes are sets of equivalent trajectories.
By (\ref{eqn:solution_preservation}) and (\ref{eqn:commutation_diagram}), if two initial
conditions are equivalent, their ensuing trajectories are equivalent. We will often 
refer to such states and trajectories as {\em dynamically equivalent} to emphasize 
that $\bu'(t)$ and $\bu(t)$ are the same trajectory up to a coordinate transformation.
If two states or trajectories are not equivalent, we call them {\em dynamically distinct}.

Equivalence between symmetry groups is somewhat more complicated. Suppose that $\bu$ and $\bu'$
are dynamically equivalent states, related by $\bu' = \gamma \bu$ for some $\gamma \in \Gppf$,
and that $\bu$ has symmetry $h \in \Gppf$. Then $\bu'$ has symmetry $h' = \gamma h \gamma^{-1}$,
the {\em conjugate} of $h$ under $\gamma$. The conjugacy relation can be derived by substituting
$\bu = \gamma^{-1} \bu'$ into $\bu = h \bu$ and then applying $\gamma$ from the left. The conjugacy
relation extends to symmetry groups; namely, if the symmetry group of $\bu$ is $H$ and
$\bu' = \gamma \bu$, then the symmetry group of $\bu'$ is 
\begin{align} \label{eqn:group_conjugacy}
H' = \gamma H \gamma^{-1} = \{ \gamma h \gamma^{-1} : h \in H\},
\end{align}
the conjugate of $H$ under $\gamma$. If we interpret $\bu' = \gamma \bu$ as a coordinate
transformation on velocity fields, then $h = \gamma h \gamma^{-1}$ and $H' = \gamma H \gamma^{-1}$
are the corresponding coordinate transformation on symmetries and symmetry groups.

Two subgroups $H$ and $H'$ of $\Gppf$ are then defined as equivalent, $H' \sim H$, if there is a
$\gamma \in \Gppf$ such that $H' = \gamma H \gamma^{-1}$. This relation is reflexive, symmetric,
and transitive and so defines an equivalence relation on the set of subgroups of $\Gppf$. The
equivalence class $\overline{H}$ of subgroup $H$ is then the set of all subgroups equivalent to $H$,
\begin{align} \label{eqn:group_equiv_class}
\overline{H} = \{H' \leq \Gppf : H' \sim H \} = \{ \gamma H \gamma^{-1} :  \gamma \in \Gppf \}
\end{align}
An equivalence class of subgroups represents a given set of symmetries viewed from 
all coordinate transformations that preserve the dynamics. By general theorems of equivalence
relations and classes, the equivalence classes of subgroups of $\Gppf$ are nonempty,
disjoint, and complete. Every subgroup $H$ of $\Gppf$ is in one and only one equivalence class,
and all subgroups in a given equivalence class are equivalent. Lastly, two invariant symmetric
subspaces $\bbU_H$ and $\bbU_{H'}$ are defined to be equivalent if their symmetry subgroups $H$
and $H'$ are equivalent. 

The import of these abstract ideas for the problem at hand is this: In general a given
invariant solution $\bu$ with symmetry group $H$ and in subspace $\bbU_H$ exists in an
infinite number of equivalent forms $\bu' = \gamma \bu$ with equivalent symmetry groups 
$H' = \gamma H \gamma^{-1}$ and equivalent subspaces $\bbU_H'$. We do not need to search
for invariant solutions in all possible subgroups and subspaces, but only for those that 
are dynamically distinct (not equivalent). To accomplish this, we first 
determine the subgroups $H$ of $\Gppf$ and classify them into equivalence classes
$\overline{H}$. We then choose one representative subgroup $H$ from each equivalence
class and search for invariant solutions within $\bbU_H$. This approach will save us
from searching redundantly in multiple different but dynamically equivalent subspaces;
i.e. from finding the same solutions repeatedly with different phase shifts and orientations. 

\subsection{Organization} \label{sec:organization}

This paper is organized as follows.
In section \ref{sec:finite_subgroups} we partially enumerate the finite subgroups of $\Gppf$
and their equivalence classes. 
Section \ref{sec:nl_tws} details $15$ newly-found traveling wave solutions in several
of the subgroups given in section \ref{sec:finite_subgroups}. Various properties of the 
newly-found traveling waves are analyzed, such as their physical structure, their continuation curves
in Reynolds number, and their bifurcations. %including some symmetry-breaking bifurcations 
Section \ref{sec:discussion} focuses on the dynamical importance of the newly-found traveling waves. 
Conclusions are presented in section \ref{sec:conclusion}.

%----------------------------------%

%----------------------------------%
% \import{sections/}{section-2.tex}
%----------------------------------%
%---------------------------------------
\section{Finite subgroups of plane Poiseuille flow}\label{sec:finite_subgroups}
%---------------------------------------

We seek to determine the finite subgroups of the plane Poiseuille symmetry group $\Gppf$,
to classify them into equivalence classes of equivalent subgroups, to count the
number of equivalence classes (the number of dynamically distinct symmetric subspaces),
and to select one representative subgroup from each equivalence class. We focus on
discrete symmetries and finite subgroups, since the infinite subgroups with continuous
translation symmetries correspond to spanwise- and streamwise-constant velocity fields,
which are less relevant to the dynamics of turbulence. 

We assume a doubly-periodic domain $\Omega = [0, L_x) \times [-1, 1] \times [0, L_z)$. For phase shifts
$\tau(a,b)$, this allows us to conduct arithmetic on $a$ and $b$ modulo $L_x$ and $L_z$ respectively.
%for example taking $\tau(L_x,0) = \tau(0,L_z) = \tau(0,0) = 1$. 
A velocity field on $\Omega$ with symmetry $\tau(a,0)$ for $0<a<L_x$ or $\tau(0,b)$ for $0<b<L_z$
can be recast on a smaller periodic domain with a simplified symmetry group. 
We will distinguish between {\em minimal-domain} subgroups, which have no such elements, and
{\em nonminimal-domain} subgroups, which do, and exclude nonminimal-domain subgroups
from consideration. 

Our strategy is as follows: 
In section \ref{sec:centering}, we begin with an arbitrary, finite, subgroup $\hat{H}$ 
of $\Gppf$ on a fixed, doubly-periodic domain $\Omega$. If $\hat{H}$ contains a $z$-reflection
symmetry of a certain form, we perform a phase shift that maps $\hat{H}$ to an equivalent subgroup
$H$ in which the $z$ reflection is centered at $z=0$. Otherwise we let $H = \hat{H}$. In either
case, $H$ is equivalent to $\hat{H}$. 
In section \ref{sec:second_order_elements}, we restrict attention to minmal-domain $\hat{H}$ and
$H$ whose elements are all second order, except for the identity. We show that all elements of $H$
with these restrictions can be written as products of the four second-order commuting symmetries 
$\sy, \sz, \tau(L_x/2, 0),$ and $\tau(0,L_z/2)$. All subgroups $H$ with the second-order
restriction are therefore subgroups of the 16th-order abelian group generated by these four
symmetries. In section \ref{sec:halfbox_groups}, we analyze the subgroups of the 16th-order
abelian group computationally. We generate all of its subgroups, eliminate those that imply
nonminimal domains, and determine the equivalence classes of the minimal-domain subgroups.
We find that there are 40 minimal-domain $z$-centered subgroups in 24 equivalence classes.
These are listed in Table \ref{tbl:subgroups}. 
It follows that any finite, minimal-domain subgroup $\hat{H}$ of $\Gppf$ with at most second-order
elements belongs to one of these 24 equivalence classes, and that these 24 equivalence classes
represent all possible finite, minimal-domain subgroups of $\Gppf$ with at most second-order elements.

In section \ref{sec:nonhalfbox_groups} we discuss subgroups $H$ with elements of order 3 or more,
present some examples, and state without proof some general principles of such groups. 

\subsection{Centering subgroups about $z=0$} \label{sec:centering}

\begin{theorem} \label{thm:zcenter}
Let $\hat{H}$ be a subgroup of $\Gppf$ for the doubly-periodic domain 
$\Omega = [0, L_x) \times [-1, 1] \times [0, L_z)$. Suppose $\hat{h} = \sy^k \sz \tau(a,b)$
is an element of $\hat{H}$, where $k=0$ or 1, $a \in [0, L_x)$,  and $b \in [0, L_z)$,
and let $\gamma = \tau(0,b/2)$. Then the subgroup $H = \gamma \hat{H} \gamma^{-1}$ of $\Gppf$
is equivalent to $\hat{H}$ and contains the element $h = \sy^k \sz \tau(a, 0)$.
\end{theorem}
\begin{proof}
Assume the suppositions of the enunciation and let $H = \gamma \hat{H} \gamma^{-1}$ for 
$\gamma = \tau(0,b/2) \in \Gppf$. By construction $H$ is equivalent to $\hat{H}$. Since
$\hat{h} = \sy^k \sz \tau(a,b)$ is in $\hat{H}$, $h = \gamma \hat{h} \gamma^{-1}$ is in $H$.
Substituting in the values of $\hat{h}$ and $\gamma$ gives
\begin{align*}
 h &= \tau(0, b/2)\, \sy^k \sz \tau(a,b)\, \tau(0, -b/2), \\
   &= \sy^k \sz \tau(0, -b/2)\, \tau(a,b)\, \tau(0, -b/2), \\
   & = \sy^k \sz \tau(a, 0).
\end{align*}
Thus $h = \sy^k \sz \tau(a, 0)$ is in $H$. 
\end{proof}

\begin{theorem} \label{thm:implication}
Let $\hat{H}$ be a subgroup of $\Gppf$ for
$\Omega = [0, L_x) \times [-1, 1] \times [0, L_z)$. If $\hat{H}$ has any elements
of the form $\hat{h}_i = \sy^{k_i} \sz \tau(a_i, b_i)$, $a_i \in [0, L_x),$ 
$b_i \in [0, L_z)$, $k_i=0$ or 1, chose one such element as $\hat{h}_1$ and let 
$H = \gamma \hat{H} \gamma^{-1}$ for $\gamma = \tau(0,b_1/2)$. Otherwise
let $H = \hat{H}$. In either case, the implication
\begin{align} \label{eqn:implication}
\exists \; h_i = \sy^{k_i} \sz \tau(a_i, b_i) \in H \; \; \Rightarrow \; \; h_1 = \sy^{k_1} \sz \tau(a_1, 0) \in H
\end{align}
is true. 
\end{theorem}
\begin{proof} Let $\hat{H}$ be a subgroup of $\Gppf$ for the given domain.
If $\hat{H}$ has one or more elements of the form $\hat{h}_i = \sy^{k_i} \sz \tau(a_i, b_i)$,
for $a_i \in [0, L_x),$ $b_i \in [0, L_z)$, $k_i=0$ or 1, we choose one such element as 
$\hat{h}_1 = \sy^{k_1} \sz \tau(a_1, b_1)$ and let $\gamma = \tau(0,b_1/2)$, and let
$H = \gamma \hat{H} \gamma^{-1}$. By theorem \ref{thm:zcenter}, $H$ is equivalent to $\hat{H}$
and contains the element $h_1 = \sy^{k_1} \sz \tau(a_1, 0)$. The consequence of implication
\ref{eqn:implication} is true, so the implication is true trivially. 

If instead $\hat{H}$ does not have an element of the form $\hat{h}_i = \sy^{k_i} \sz \tau(a_i, b_i)$,
then we let $H = \hat{H}$, or equivalently, $H = \gamma \hat{H} \gamma^{-1}$ for $\gamma = 1$. Then 
implication \ref{eqn:implication} is true vacuously, since its premise is false. 
\end{proof}

Geometrically, a symmetry of form $\hat{h}_i = \sy^{k_i} \sz \tau(a_i, b_i)$ represents a
$z$-reflection symmetry centered about $z=b_i/2$, possibly combined with a phase shift in $x$ and/or 
a reflection in $y$.  
The coordinate transformation $H = \gamma \hat{H} \gamma^{-1}$ for $\gamma = \tau(0, b_1/2)$ or 1
specified in Theorem \ref{thm:implication} shifts the origin in $z$ so that, if $H$ has elements
with $z$-reflections, at least one element $h_1 = \sy^{k_1} \sz \tau(a_1, 0)$ is centered on $z=0$.  
We call symmetry groups $H$ that satisfy implication \ref{eqn:implication} ``$z$-centered.''

\subsection{Restriction to subgroups with second-order elements} \label{sec:second_order_elements}

At this point we restrict attention to minimal-domain subgroups $\hat{H}$ of $\Gppf$ that have elements
of order 2 or less; i.e. $\hat{h}^2 = 1$ for all $\hat{h} \in \hat{H}$. Let $H$ be the $z$-centered
equivalent of $\hat{H}$ by the construction of Theorem \ref{thm:implication}, $H = \gamma \hat{H} \gamma^{-1}$
for $\gamma = \tau(0, b_1/2)$ or $\gamma = 1$. It follows that $H$ is also minimal domain, since
$\tau(a,0)$ and $\tau(0,b)$ are invariant under conjugacy by either form of $\gamma$, and that 
all elements of $H$ are second-order or less, since for any $h \in H$, 
$h^2 = (\gamma \hat{h} \gamma^{-1}) (\gamma \hat{h} \gamma^{-1}) = \gamma \hat{h}^2 \gamma^{-1} = \gamma \gamma^{-1} = 1$. 

We now show that any element $h$ in $H$ can be written as a product of $\sx, \sy, \sz, \tau(L_x/2, 0),$
and $\tau(0, L_z/2)$. For the latter two elements and their product we introduce the more compact notation 
\begin{align} \label{eqn:halfbox_shifts}
\tx &= \tau(L_x/2, 0), \nonumber\\
\tz &= \tau(0, L_z/2), \\
\txz &= \tau(L_x/2, L_z/2). \nonumber
\end{align}

\begin{theorem} 
Let $H$ be a $z$-centered, minimal-domain subgroup of $\Gppf$ for the doubly-periodic domain
$\Omega = [0, L_x) \times [-1, 1] \times [0, L_z)$. Assume all elements of $H$ are order 2 or
less. Then every element of $H$ can be expressed as a product of factors chosen from $\sy, \sz, \tx,$
and $\tz$.
\end{theorem}
\begin{proof} Let $h$ be an element of the $z$-centered, minimal-domain subgroup $H$ of $\Gppf$ on 
$\Omega$, and let $h^2=1$ for all $h \in H$. As discussed in section \ref{sec:intro}, since $h$ is
an element of $\Gppf$, it can be expressed in one of the four forms $\tau(a,b)$, $\sy \tau(a,b),$
$\sz \tau(a,b),$ or $\syz \tau(a,b)$, where $a \in [0,L_x)$ and $b \in [0, L_z)$. These four forms
can be reduced to two, $h = \sy^k \tau(a,b)$ or $h = \sy^k \sz \tau(a,b)$, using the auxiliary 
variable $k=0$ or 1. 

{\em Case 1:} $h = \sy^k \tau(a,b)$. Since $h$ is second-order or less, 
$h^2 = \sy^k \tau(a, b) \, \sy^k \tau(a, b) = \sy^{2k} \tau(a, b) \, \tau(a, b) = \tau(2a,2b) = 1$.
Since $\tau(0,0) = 1$, we have $\tau(2a,2b) = \tau(0,0)$. Since 
$\tau(0,0) = \tau(L_x,0) = \tau(0, L_z) = \tau(L_x,L_z)$ by the periodicity of the domain, the
complete set of solutions to $\tau(2a,2b) = \tau(0,0)$ for $a$ in $[0, L_x)$ and $b$ in $[0, L_z)$ are 
$(a,b) = (0,0),$ $(L_x/2,0),$ $(0, L_z/2)$, and $(L_x/2, L_z/2)$. Thus $\tau(a,b)$ is either 
$1, \tx, \tz,$ or $\txz = \tx\tz$, and $h = \sy^k \tau(a,b)$ can be written as a product of factors
chosen from $\sy, \tx$ and $\tz$. 

{\em Case 2:} $h = \sy^k \sz \tau(a,b)$. Since $h$ is second-order or less,
$h^2 = \sy^k \sz \tau(a,b) \, \sy^k \sz \tau(a,b) = \sy^{2k} \sz^2 \tau(a,-b) \tau(a,b) = \tau(2a,0) = 1$.
The solutions of this equation for $a$ in $[0, L_x)$ are $a = 0$ and $a=L_x/2$.

Since $H$ is $z$-centered and has an element of form $\sy^k \sz \tau(a,b)$, it follows
from lemma that $H$ has an element of form $h_1 = \sy^{k_1} \sz \tau(a_1, 0)$ for some 
$a_1 \in [0,L_x)$ and $k_1=0$ or 1. The element $h_1^2 = \sy^{2k_1} \sz^2 \tau(2a_1,0) = \tau(2a_1,0)$
is then also in $H$. Since $L_x>0$, $2a_1 = jL_x + r$ for some integer $j$ and some real remainder
$r$ satisfying $0 \leq r < L_x$. Then $\tau(2a_1,0) = \tau(jL_x+r,0) = (\tau(L_x, 0))^{\,j}\tau(r,0) = \tau(r,0)$
is in $H$. However, since $\Omega$ is the minimal periodic domain for $H$, $\tau(r,0)$ is not
in $H$ for $0<r<L_x$. Therefore $r=0$ and $2a_1 = jL_x$ for some integer $j$. Since
$0 \leq 2a_1 = jL_x < 2L_x$, we have $j = 0$ or 1, and thus $a_1 = 0$ or $L_x/2$. 

Since $h$ and $h_1$ are in $H$, it follows that $(h_1 h)^2 = (\sy^{k_1} \sz \tau(a_1, 0) \, \sy^k \sz \tau(a,b))^2 = (\sy^{k_1+k} \tau(a_1 + a, b))^2 = \tau(2a_1+2a, 2b)$
is in $H$. Since $a$ and $a_1$ are each either 0 or $L_x/2$, $2a_1 + 2a$ is an integer multiple
of $L_x$, and thus $\tau(2a_1+2a, 2b) = \tau(0,2b)$ is in $H$. By an argument similar to that 
above for $a_1$, the minimal periodic length of $L_z$ in $z$ implies $b=0$ or $L_z/2$. Since $a$ is
either $0$ or $L_x/2$ and $b$ is either 0 or $L_z/2$, $\tau(a,b)$ is either $1, \tx, \tz,$ or 
$\txz = \tx\tz$. Thus $h = \sy^k \sz \tau(a,b)$ can be written as a product of factors chosen from
$\sy, \sz, \tx,$ and $\tz$. 
\end{proof}

\begin{lemma}
The elements $\sy, \sz, \tx$ and $\tz$ of $\Gppf$ commute with each other, and each is second-order.
\end{lemma}
\begin{proof} Inspection. \end{proof}

\begin{lemma} 
Any product of factors $\sy, \sz, \tx$ and $\tz$ is at most second order.
\end{lemma}
\begin{proof}
Let $h$ be the product of any number of factors $\sy, \sz, \tx$ and $\tz$ in any order.
Since these factors commute, $h$ can be written as $h = \sy^j \sz^k \tx^m \tz^n$
where $j,k,m,n$ are non-negative integers. 
Then $h^2 = (\sy^j \sz^k \tx^m \tz^n)^2 = \sy^{2j} \sz^{2k} \tx^{2m} \tz^{2n} = 1$. 
\end{proof}
\subsection{Enumeration and classification of all minimal half-box subgroups} \label{sec:halfbox_groups}

In this section we enumerate the $z$-centered subgroups $H$ of $\Gppf$ with second-order
elements or less, eliminate the nonminiml-domain subgroups, and determine the equivalence
classes of the remaining minimal-domain subgroups.
It follows from the results of section \ref{sec:second_order_elements} that any second-order
element of any $z$-centered, minimal-domain subgroup of $\Gppf$ can be written as a product of factors
$\sy, \sz, \tx,$ and $\tz$. Thus any $z$-centered subgroup with at most second-order elements
is a subgroup of the group generated by these factors. We denote this group as 
\begin{align} \label{def:Gppfh}
  \Gppfh &= \langle \sy, \sz, \tx, \tz \rangle \\
  &=
  \{1, \sy, \sz, \tx, \tz, \syz, \txz, \sy\tx, \sy\tz, \sz\tx, \sz\tz, \syz\tx, \syz\tz, \sy\txz, \sz\txz, \syz\txz\} \nonumber
\end{align}
Note that $\Gppfh$ is abelian subgroup of $\Gppf$, since its generators are commuting elements
of $\Gppf$. By Lagrange's theorem, the order of any subgroup of $\Gppfh$ must be a divisor of 
$|\Gppfh| = 16$, namely 1, 2, 4, 8, or 16. Since $\Gppfh$ is abelian, the converse of
Lagrange's theorem holds as well, guaranteeing the existence of nontrivial subgroups of orders
2, 4, and 8, in addition to the trivial order-1 subgroup $\{1\}$ and the trivial order-16 subgroup
$\Gppfh$. Since $\Gppfh$ is abelian, any two subgroups $H_1, H_2$ of $\Gppfh$ commute,
and the product $H_1 H_2 = \{h_1 h_2 \,|\, h_1 \in H_1, \, h_2 \in H_2\}$ is itself a subgroup
of $\Gppfh$. Since all elements of $\Gppfh$ except the identity are second order, we can
construct all nontrivial subgroups by constructing all order-2 subgroups 
$\langle h \rangle = \{1, h\}$ for all non-identity elements $h \in \Gppfh$, 
then taking the products of all distinct order-2 subgroups to generate all order-4 subgroups,
then similarly forming the order-8 subgroups from products of order-2 subgroups and order-4 subgroups,
at each step eliminating redundantly generated subgroups for efficiency.

After generating all subgroups of $\Gppfh$, we eliminate from consideration all those that contain
$\tx$ or $\tz$ as elements, since these elements in isolation imply periodicity on a smaller periodic
domain. From the remaining set of subgroups, we determine the equivalence classes by checking 
which subgroups are equivalent under conjugacy. That is, given two subgroups $H_1$ and $H_2$ of
$\Gppfh$, we must determine if there is a $\gamma \in \Gppf$ such that $H_1 = \gamma^{-1} H_2 \gamma$.
However, we need not check for conjugacy under all $\gamma \in \Gppf$; a theorem in Appendix
\ref{sec:appendixA} shows that if $H_1$ and $H_2$ are unequal but conjugate subgroups of $\Gppfh$,
then they are conjugate under $\gamma = \tau(0,L_z/4)$. Thus we check for conjugacy only under
$\gamma = \tau(0,L_z/4)$. After arranging all the subgroups into equivalence classes, we then select
a single representative subgroup from each equivalence class. The set of the representatives then
represents the set of physically distinct minimal-domain subgroups of plane Poiseuille flow, up to
second-order elements. 

\newcommand{\bin}{\bsym{\text{\bf~in~}}}
\newcommand{\HinEj}{H\_in\_Ej}
\algnewcommand\algorithmicnot{\textbf{not}}
\algdef{SE}[IF]{IfNot}{EndIf}[1]{\algorithmicif\ \algorithmicnot\ #1\ \algorithmicthen}{\algorithmicend\ \algorithmicif}%

\begin{algorithm} 
\caption{Generate all minimal-domain subgroups of $\Gppfh$ and arrange into equivalence
classes.}\label{alg:algo}
\begin{algorithmic}[1]
\Statex \hspace{-5mm} {\bf Notation:}
\Statex $H_i, H_j, H:$ subgroups of $\Gppfh$
\Statex $S^{(n)} = \{H_1, H_2, \ldots\},$ set of all order $2^{n-1}$ subgroups of $\Gppfh$ containing neither $\tx$ nor $\tz$
%\Statex $S^{(n)}_i :$ the $i$th subgroup in the set $S^{(n)}$
\Statex $E^{(n)}_j = \{H_{j_1}, H_{j_2}, \ldots\}, $ the $j$th equivalence class, a set of equivalent order $2^{n-1}$ subgroups
\Statex $E^{(n)} = \{E^{(n)}_1, E^{(n)}_2, \ldots\},$ the set of equivalence classes for order $2^{n-1}$ subgroups
\Statex $k :$ a counter for the number of equivalence classes in $E^{(n)}$ as it is built
\Statex \HinEj : a Boolean indicating subgroup $H$ was found in some $E_j^{(n)}$, $1 \leq j \leq k$
\Statex
%\State $\Gppfh \gets \{1, \sy, \sz, \tx, \tz, \syz, \txz, \sy\tx, \sy\tz, \sz\tx, \sz\tz, \syz\tx, \syz\tz, \sy\txz, \sz\txz, \syz\txz\}$
\State $\Gppfh \gets \langle \sy,\, \sz,\, \tx,\, \tz \rangle$
\State $S^{(1)} \gets \{ \langle 1 \rangle \}$
\State $S^{(2)} \gets \{ \langle h \rangle \; : \; h \in \Gppfh, \, h \not\in \{1, \tx, \tz\} \}$
\Statex
\For{$n = 3,4,5$}  \Comment{Construct set $S^{(n)}$ inductively from $S^{(n-1)}$}
   \State $S^{(n)} = \emptyset$ \Comment{Initiate $S^{(n)}$ to the empty set}
   \For{$H_i \bin S^{(n-1)}$}
      \For{$H_j \bin  S^{(n-1)}$, ~$j < i$}
        \State $H \gets H_i \, H_j$ \Comment{Construct subgroup $H$ as candidate for inclusion in $S^{(n)}$}
        \If{ $|H| = 2^{n-1}~\text{\bf and } \tx \not\in H ~\text{\bf and } \tz \not\in H$ }
            \State $S^{(n)} \gets S^{(n)} \cup \{ H\}$ \Comment{If $H$ is order $2^{n-1}$ and minimal, add it to set $S^{(n)}$}
        \EndIf
      \EndFor
   \EndFor
\EndFor
\Statex
\State $\gamma \gets \tau(0, L_z/4)$
\For{$n = 1,2,3,4,5$} \Comment {Partition $S^{(n)}$ into set of equiv. classes $E^{(n)} = \{E^{(n)}_1, E^{(n)}_2, \ldots\}$}
   \State $E^{(n)} \gets \emptyset$\Comment{Initiate $E^{(n)}$ to empty set}
   \State $k \gets 0$\Comment{Set counter for $|E^{(n)}|$ to zero}
   \For{$H \bin S^{(n)}$} \Comment{Put each subgroup $H$ in $S^{(n)}$ into some equiv. class $E^{(n)}_j$}
       \State \HinEj $\, \gets$ false
       \For{$j=1,\ldots,k$}
       \If{ $\gamma H \gamma^{-1} \in E^{(n)}_j$} \Comment{$H$ belongs to equivalence class $E^{(n)}_j$}
            \State $E^{(n)}_j \gets E^{(n)}_j \cup \{ H \}$ \Comment {Add subgroup $H$ to equivalence class $E^{(n)}_j$}
            \State \HinEj $\, \gets$ true \Comment{Mark that $H$ belongs to an existing equivalence class}
            \State {\bf break}
          \EndIf
       \EndFor
       \IfNot{\HinEj} \Comment{$H$ does not belong to an existing equivalence class}
         \State{$E^{(n)}_{k+1} \gets \{H\}$} \Comment{Form new equivalence class $E^{(n)}_{k+1}$ with one element $H$}
         \State{$k \gets k+1$}               \Comment{Increment counter for $|E^{(n)}|$}
       \EndIf
   \EndFor
\EndFor

\end{algorithmic}

\end{algorithm}

Since the the number of subgroups and pairwise conjugacies to check is large,
we implemented the algorithm described above in symbolic code and generated the
subgroups and equivalence classes of $\Gppfh$ computationally. Our symbolic codes
were written in the Julia programming language \citep{bezanson2017julia}; pseudocode
for the algorithm is given in Algorithm \ref{alg:algo}.
The action of an arbitrary symmetry $\gamma \in \Gppf$ can be represented by
\begin{align}
  \gamma[u,v,w](x,y,z) = [s_x u, s_y v, s_z w](s_x x + a_x L_x, s_y y, s_z z + a_z L_z) \label{eqn:gamma}
\end{align}
where $a_x, a_z$ are real-valued and $s_x, s_y, s_z$ take values $\pm 1$. The symbolic code
represents a given symmetry $\gamma$ by its 5-tuple parametrization $(s_x, s_y, s_z, a_x, a_z)$
and computes inverses and group multiplication in terms of this representation.
Arithmetic on $a_x$ and $a_z$ is performed modulo 1 so that each group element has a unique 5-tuple
representation and equality of group elements can be tested by equality of their 5-tuples. The symbolic
code uses Julia's built-in rational number type for $a_x, a_z$ rather than floating-point numbers so
that group operations and equality tests are exact. The symbolic code represents a group by an array
of its group elements. We defined a strict ordering on group elements in terms of the representation,
so that a given group has a unique representation and equality of two groups is established by
equality of their ordered elements. Conjugation of subgroups $\gamma^{-1} H \gamma$ is performed 
by conjugation $\gamma^{-1} h \gamma$ of the elements $h$ of $H$. Equivalence of subgroups $H_1$ and
$H_2$ under conjugacy by $\gamma$ is tested by computing $\gamma^{-1} H_1 \gamma$ and then testing
for equality with $H_2$. Group products $H_1 H_2$ are determined by pairwise multiplication
of all elements of $H_1$ against all elements of $H_2$ followed by a reduction to a set
of sorted, unique elements. Once the parametrization, group arithmetic, and ordering
were defined, much of the rest of the functionality of the symbolic code resulted from
application of built-in Julia libraries for tuples, arrays, sorting, rational
numbers, and tests for equality and set membership. To check for correctness, we wrote two
symbolic codes independently and confirmed that they produced identical sets of subgroups 
and conjugacy classes. We also performed manual verification of conjugacy between the 
subgroups in each equivalence class.

\begin{table}
  %\begin{tabular}{p{0.08\textwidth} p{0.4\textwidth} p{0.2\textwidth} p{0.4\textwidth} }
  \begin{tabular}{llp{0.1\textwidth}ll}
     order $1$: & 1 subgroup in 1 equivalence class\\
     \\
     & $\bsym{\langle 1 \rangle}$\\  % E01
     \\
     order $2$: & 13 subgroups in 9 equivalence classes \\
     \\
     &       $\langle \sy \rangle$        & &       $\langle \sz \rangle  \sim \langle \sz \tz \rangle $  \\       % E02, E03
     & $\bsym{\langle \txz \rangle}$      & &       $\langle \sz \tx \rangle \sim \langle \sz \txz \rangle$  \\    % E05, E08
     & $\bsym{\langle \sy \tx \rangle}$   & & $\bsym{\langle \syz \tx \rangle \sim \langle \syz \txz \rangle }$ \\ % E06, E09
     & $\bsym{\langle \sy \tz \rangle}$   & & $\bsym{\langle \syz \rangle \sim \langle \syz \tz \rangle }$ \\      % E07, E04
     & $\bsym{\langle \sy \txz \rangle}$ \\                                                                        % E10
     \\ order $4$: & 22 subgroups in 12 equivalence classes \\ 
     \\
     & $\bsym{\langle \sy,\, \txz \rangle}$ & & $\bsym{\langle \sy\tx,\, \txz \rangle}$  \\ % E12, E21
     & $\langle \sy,\, \sz \rangle \sim \langle \sy,\, \sz\tz \rangle$  &                   % E11
     & $\bsym{\langle \syz,\, \txz \rangle \sim \langle \syz\tx,\, \txz \rangle}$ \\       % E18
     & $\langle \sz,\, \txz, \rangle \sim \langle \sz\tx,\, \txz \rangle $  &               % E14
     & $\bsym{\langle \sz,\, \sy\tz \rangle \sim \langle \syz,\, \sy\tz \rangle}$\\         % E16
     & $\langle \sy,\, \sz\tx \rangle \sim \langle \sy,\, \sz\txz \rangle $  &              % E13
     & $\bsym{\langle \sz,\, \sy\txz \rangle \sim \langle \sz\tz,\, \syz\tx \rangle}$ \\    % E17
     & $\bsym{\langle \sz,\, \sy\tx \rangle \sim \langle \sy\tx,\, \sz\tz \rangle}$  &      % E15
     & $\bsym{\langle \sy\tz,\, \sz\tx \rangle \sim \langle \sz\txz,\, \sy\tz \rangle}$ \\  % E22
     & $\bsym{\langle \syz,\, \sy\txz \rangle \sim \langle \sz\tx,\, \syz\tz \rangle}$ &    % E20
     & $\bsym{\langle \sy\tz,\, \sz\tx, \rangle \sim \langle \sy\tz,\, \sz\txz \rangle}$ \\ % E19 (E19.2 was listed as <sztxz, sytxz> == <sy, sztxz> == E13.2, typo was \sy\txz instead of \sy\tz for second generator
     \\ order $8$: & 4 subgroups in 2 equivalence classes\\
     \\
     & $\langle \sy,\, \sz,\, \txz \rangle \sim \langle \sy, \sz \tx,\, \txz \rangle$  &         % E23
     & $\bsym{\langle \sz,\, \sy\tx,\, \txz \rangle \sim \langle \syz,\, \sy \tx,\, \txz \rangle}$ \\  % E24
   \end{tabular}
   \caption{\label{tbl:subgroups}
     {\bf Symmetry subgroups and their equivalence classes for plane Poiseuille flow in minimal doubly-periodic domains.}
     The symbol $\sim$ denotes equivalence between two subgroups under conjugation. 
     A set of subgroups related by $\sim$ symbols forms an equivalence class.
     A subgroup listed in isolation forms an equivalence class with that subgroup as its sole element. 
     The first subgroup listed in each equivalence class is chosen as the representative for that equivalent class. 
     Subgroups and equivalence classes in bold are new to the literature.
     The subgroups listed are the complete set of $z$-centered, minimal-domain subgroups with elements of order 2 or less
     for plane Poiseuille flow in doubly-periodic domains.
}
 \end{table}

Table \ref{tbl:subgroups} provides a complete list of the $z$-centered minimal subgroups of $\Gppf$ 
with elements of order 2 or less and their equivalence classes. The subgroups are listed 
in terms of their generators, e.g. $\langle \sy, \txz \rangle$ is the 4th-order subgroup
$\{1, \sy, \txz, \sy \txz\}$. The equivalence classes are indicated by equivalence 
relations between groups, i.e. $H_1 \sim H_2$ indicates that groups $H_1$ and $H_2$ are 
related by a conjugacy and thus equivalent under a coordinate transformation. A group listed by itself 
indicates that that group is the sole member of its equivalence class. We find there are 
1 order-1 subgroups in 1 equivalence class, 13 order-2 subgroups in 9 equivalence classes, 
22 order-4 subgroups in 12 equivalence classes, and 4 order-8 subgroups in 2 equivalence classes. 
Thus there are a total of 24 equivalence classes representing 24 physically distinct
symmetry subgroups with 2nd-order elements or less. 

Prior studies have found invariant solutions for plane Poiseuille flow in 7 of these
24 distinct subgroups.
\cite{waleffe2001exact} found invariant solutions with
$\langle \sz\tx \rangle$                  % E08.1
symmetry.
\cite{nagata2013mirror} found traveling waves with 
$\langle \sz,\, \txz \rangle$,            % E14.1, originally listed with noncanonical generators \sz\txz, \sz
$\langle \sy,\, \sz\txz \rangle$, and     % E13.2
$\langle \sy, \, \sz, \, \txz \rangle$    % E23.1
symmetries.
\cite{park2015exact} reported invariant solutions in 
$\langle \sy \rangle$,                              % E02.1
$\langle \sy, \, \sz \rangle$,                      % E11.1
$\langle \sy, \, \sz, \, \txz \rangle$,             % E23.1 
and the equivalent subgroup 
$\langle \sy, \, \txz,\, \sz \tx \rangle$.          % E23.2 (redundant)
\cite{gibson2014spanwise} reported traveling waves with
$\langle \sz \tx \rangle$,                          % E08.1
$\langle \sz \rangle$,                              % E03.1
$\langle \sy, \, \sz \rangle$,                      % E11.1
$\langle \sy, \, \sz \tx \rangle$,                  % E13.1
$\langle \sz, \, \txz \rangle$, and                 % E14.1
$\langle \sy, \, \sz, \, \txz \rangle$  % E23.1
symmetries.             
\cite{zammert2015crisis} found a traveling wave with 
$\langle \sy, \, \sz, \, \txz \rangle$              % E23.1
symmetry.
Thus the 7 previously-explored subgroups are 
$\langle \sy \rangle,$                              % E02.1
$\langle \sz \rangle,$                              % E03.1
$\langle \sz \tx \rangle,$                          % E08.1
$\langle \sy, \, \sz \rangle,$                      % E11.1
$\langle \sy, \, \sz\tx \rangle \sim \langle \sy, \, \sz\txz \rangle,$     % E13.1, 13.2 
$\langle \sz, \, \txz \rangle,$ and                 % E14.1
$\langle \sy, \sz, \txz \rangle \sim \langle \sy, \, \sz \tx, \, \txz \rangle$.   

Note that our coordinate system is consistent with \cite{waleffe2001exact}, \cite{toh2003periodic}, and
\cite{park2015exact}. 
%The shift-reflect symmetry imposed by \cite{waleffe2001exact} corresponds to our $\sz \tx$ symmetry. 
The conventions of \cite{nagata2013mirror} differ; they have $x$ as the streamwise coordinate, $y$
as spanwise, and $z$ as wall-normal. Their $\Omega$, $\mathcal{S},$ $\mathcal{Z}_{y},$ and $\mathcal{Z}_{z}$
symmetries correspond to $\sxy \txz,$ $\sz \txz,$ $\sz,$ and $\sy$, respectively, in our notation.
% The new solutions we computed have these symmetries
% 
% <sy>              E02
% <sy, sz>          E11.1
% <sy, txz>         E12
% <sy, sztx>        E10
% <sz, txz>         E14.1
% <sy, sz, txz>     E23
% <tau(Lx/3, Lz/3)> 

%-----------------------------------%
\subsection{Subgroups with higher-order elements} \label{sec:nonhalfbox_groups}
%-----------------------------------%
\newcommand{\Hh}[1]{H_{#1}}
\begin{figure} 
 \begin{center}
 (a) \includegraphics[width=0.28\textwidth]{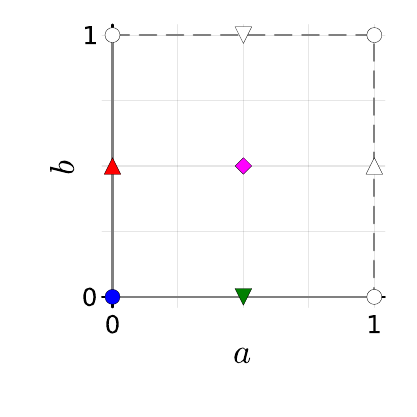}
 (b) \includegraphics[width=0.28\textwidth]{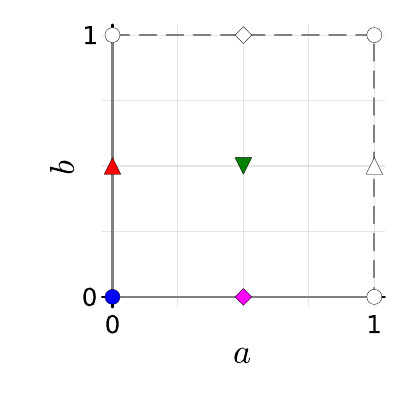}  
 ~~~\, \includegraphics[width=0.28\textwidth]{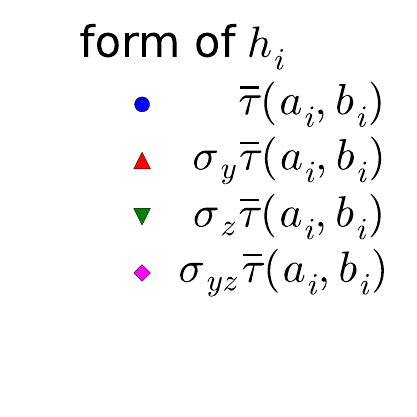} \\
 (c) \includegraphics[width=0.28\textwidth]{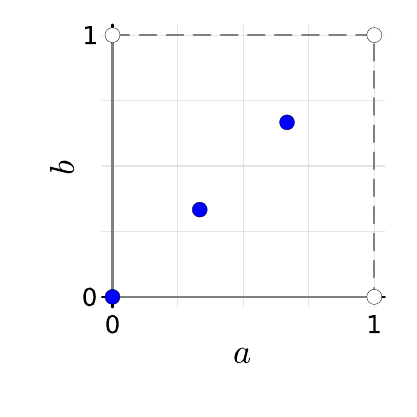} 
 (d) \includegraphics[width=0.28\textwidth]{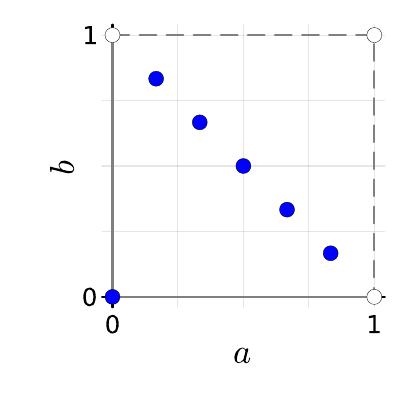}
 (e) \includegraphics[width=0.28\textwidth]{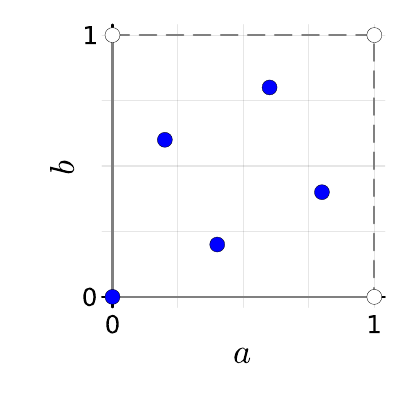} \\
 (f) \includegraphics[width=0.28\textwidth]{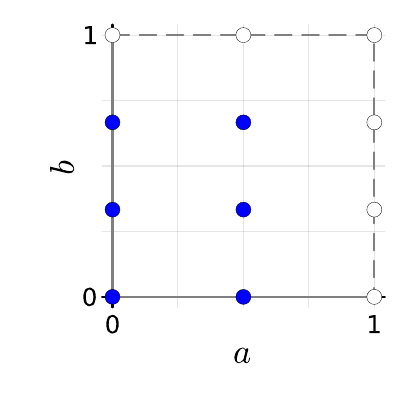} 
 (g) \includegraphics[width=0.28\textwidth]{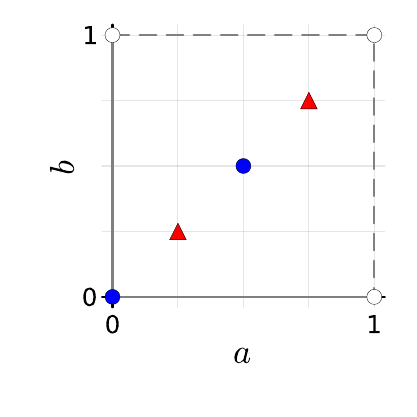}
 (h) \includegraphics[width=0.28\textwidth]{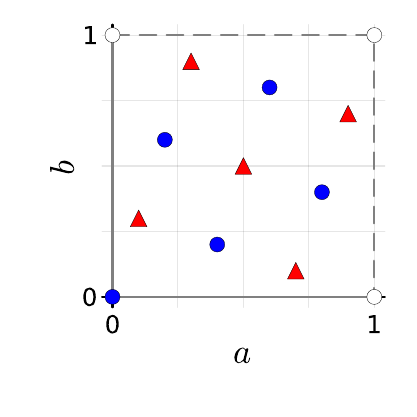} \\
 (i) \includegraphics[width=0.28\textwidth]{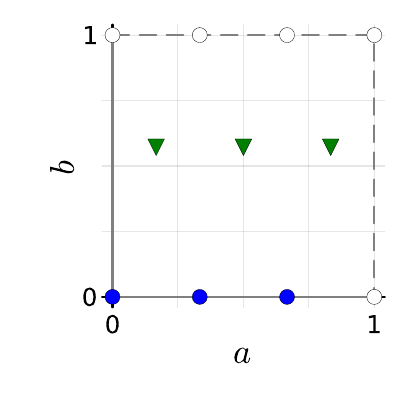} 
 (j) \includegraphics[width=0.28\textwidth]{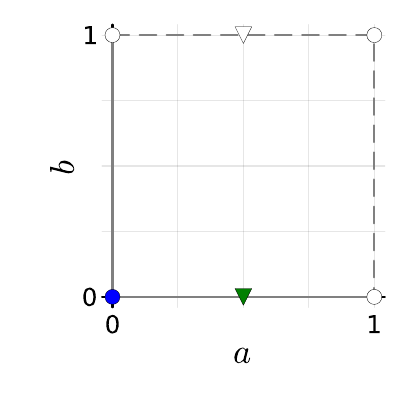} 
 (k) \includegraphics[width=0.28\textwidth]{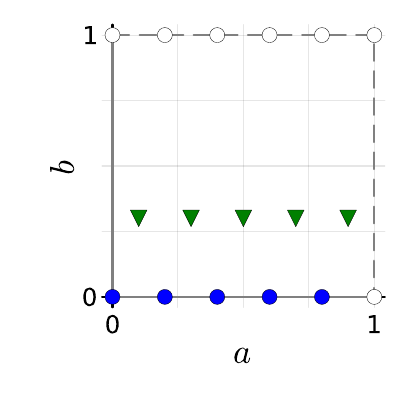} \\
  \end{center}
\caption{{\bf Visualizations of symmetry subgroups of plane Poiseuille flow.} 
Each subplot shows a subgroup $H$ of the plane Poiseuille symmetry group $\Gppf$.
The subgroup is specified in terms of its generators, $H = \langle h_1 \rangle$ or $H = \langle h_1, h_2 \rangle$.
The plot is formed by expanding the group to a set of elements $H = \{h_1, h_2, h_3, \ldots\}$
and then plotting each $h_i$ in $H$ on the $a,b$ plane using the marker codings
indicated in the legend. %: $\textcolor{blue} \bullet$ for elements of form $h_i = \btau(a_i,b_i)$, etc.
Group elements on the $a$ or $b$ axes are repeated without shading at $a=1$ or $b=1$ to indicate
the periodicity of the domain. \\
%$\textcolor{red} \blacktriangle$ for $h_i = \sy \btau(a_i,b_i)$, etc. \\ 
(a) $\Hh{a} = \langle \sz \tx,  \sy \tz \rangle = \langle \sz \btau(1/2,0),\, \sy \btau(0,1/2) \rangle$,\\
(b) $\Hh{b} = \langle \sz \txz, \sy \tz \rangle = \langle \sz \btau(1/2,1/2),\, \sy \btau(0,1/2) \rangle$,\\
(c) $\Hh{c} = \langle \btau(1/3, 1/3) \rangle$, ~~~~~~~(d) $\Hh{d} = \langle \btau(1/6, 5/6) \rangle$, ~~~~~~~~~~~(e) $\Hh{e} = \langle \btau(1/5, 3/5) \rangle$, \\
(f) $\Hh{f} = \langle \btau(1/2, 2/3) \rangle$, ~~~~~~~(g) $\Hh{g} = \langle \sy \btau(1/4, 1/4) \rangle,$ ~~~~~\, (h) $\Hh{h} \, = \langle \sy \btau(1/10, 3/10) \rangle$, \\
(i) $\Hh{i} = \langle \sz \btau(1/6, 4/7) \rangle$, ~~~~(j) $\Hh{j} = \langle \btau(1/2, 0) \rangle$, ~~~~~~~~~~~~~~~~(k) $\Hh{k} = \langle \sz \btau(1/10, 3/10) \rangle$.\\
 }
\label{fig:groupviz}
\end{figure}

We briefly expand on finite subgroups of $\Gppf$ involving phase shifts by lengths other than $L_x/2$
and $L_z/2$, or equivalently, subgroups that contain elements of order three or more.  As in section
\ref{sec:halfbox_groups} we assume that velocity fields and subgroups are $z$-centered. In this section,
we use domain-normalized notation
\begin{align} \label{eqn:btau}
  \btau(a,b) = \tau(aL_x, bL_z)
\end{align}
to simplify notation in the important case where the phase shifts are rational multiples of the periodic lengths.  
In principle it is possible to construct and classify all finite subgroups with elements up 
to some maximal order using the methods presented in sections \ref{sec:centering}-\ref{sec:halfbox_groups}.
This approach, however, quickly becomes impractical due to combinatoric explosion as $n$ increases. 
Instead we present a number of examples of such groups and infer some principles governing them.

Figure \ref{fig:groupviz} provides a visualization of 11 subgroups of $\Gppf$.
The subgroups are labeled $\Hh{a}$ through $\Hh{k}$ corresponding to the labels of the subplots. 
The first two subgroups are equivalent minimal-domain half-box subgroups 
$\Hh{a} =  \langle \sz \tx, \sy \tz \rangle$ and 
$\Hh{b} =  \langle \sz \txz,\, \sy \tz \rangle$.
In domain-normalized notation, $\Hh{a} =  \langle \sz \btau(1/2, 0),\, \sy \btau(0, 1/2) \rangle$,
and when expanded to a set of elements, 
\begin{align}
  \Hh{a} &=  \{1,\, \sz \btau(1/2, 0),\, \sy \btau(0, 1/2),\, \syz \btau(1/2, 1/2)\}.
\end{align}
The four elements of $\Hh{a}$ are plotted in the $(a,b)$ plane in fig.\ \ref{fig:groupviz}(a)
using the marker code shown in the legend. For example, $h_1 = 1 = \btau(0,0)$ is shown as a
blue circle at $(a,b) = (0,0)$, and $h_2 = \sz \btau(1/2, 0)$ is shown with a green 
downward-pointing triangle at $(1/2,0)$. Similarly, in domain-normalized form $\Hh{b}$ is
\begin{align}
  \Hh{b} &=  \{1,\, \sz \btau(1/2, 1/2),\, \sy \btau(0, 1/2),\, \syz \btau(1/2, 0)\},
\end{align}
and its four elements are illustrated in fig.\ \ref{fig:groupviz}(b). 

The remaining subfigures (c) through (k) illustrate subgroups with elements of order 3 or more. 
A key question is whether a given set of generators generates a minimal- or nonminimal-domain 
subgroup; i.e. whether an element of domain-normalized form $\btau(\alpha,0)$ for some $0<\alpha<1$ or 
$\btau(0,\beta)$ for some $0<\beta<1$ results from some product of the generators. For a single 
generator of form $\btau(a,b)$ with rational $a,b$, this question reduces to number theory.
Specifically, if $a,b$ are reduced rationals $a=j/m$ and $b=k/n$, where $j,k$ are integers and
$m,n$ natural numbers, with $j$ and $k$ relatively prime to $m$ and $n$ respectively, is there
an natural number $\ell$ for which $\ell j \equiv 0 \!\! \mod m$ and $\ell k \not\equiv 0 \!\!\mod n$,
or vice versa? If so, then $\btau^{\ell}(a, b) = \btau(\ell j/m, \ell k/n)$ has form 
$\btau(0, \beta)$ for $0<\beta<1$ or $\btau(\alpha, 0)$ for $0<\alpha<1$ and thus the generated
group is nonminimal.

Figure \ref{fig:groupviz}(c) illustrates the minimal-domain subgroup generated by the rational phase shift 
$\btau(1/3, 1/3)$, namely,
$\Hh{c} = \langle \btau(1/3, 1/3) \rangle  =  \{1,\, \btau(1/3, 1/3),\, \btau(2/3, 2/3) \}$.
In general, a phase shift of form $\btau(1/n, 1/n)$ for natural number $n$ generates the
$n$th-order minimal-domain subgroup
\begin{align}
  \langle \btau(1/n, 1/n) \rangle  = \{1,\, \btau(1/n, 1/n),\, \btau(2/n, 2/n),\, \ldots, \btau((n-1)/n, (n-1)/n)\}.
\end{align}
Subgroups of this form are minimal-domain for all natural numbers $n$, since they contain no elements
of form $\btau(\alpha,0)$ or $\btau(0,\beta)$ for $0 < \alpha, \beta < 1$. Subgroups of this form
induce a diagonal structure on velocity fields, since the points of the velocity field that are
equated by these symmetries lie along diagonals of slope $L_z/L_x$ in the $x,z$ plane, namely 
$\bu(x,y,z) = \bu(x + (j/n) L_x, y, z + (j/n) L_z)$ for $j = 0, 1, \ldots, n-1$. 

Figure \ref{fig:groupviz}(d) illustrates the minimal-domain subgroup generated by the rational phase shift
$\btau(1/6, 5/6)$, namely, 
$\Hh{d} = \{1,\, \btau(1/6, 5/6),\, \btau(2/6, 4/6),\, \ldots\, \btau(5/6, 1/6)\}.$
In general, $\btau(1/n, (n-1)/n)$ for natural number $n$ generates the $n$th-order 
minimal-domain subgroup
\begin{align} \label{eqn:reversed_diagonal}
  \langle\btau(1/n, (n-1)/n) \rangle  = \{1,\, \btau(1/n, (n-1)/n),\, \btau(2/n, (n-2)/n),\, \ldots, \btau((n-1)/n, 1/n)\}.
\end{align}
The $h_j = \btau(j/n, (n-j)/n)$ structure of the subgroup elements derives from the fact that
$n-j \equiv -j \mod n$, by which $\btau(1/n, (n-1)/n) = \btau(1/n, -1/n)$, and thus 
$\btau^j(1/n, (n-1)/n) = \btau^j(1/n, -1/n) = \btau(j/n, -j/n) = \btau(j/n, (n-j)/n)$.
Subgroups of this form are minimal for all natural numbers $n$, and they induce diagonal
structure along the diagonal of slope $-L_z/L_x$ in the $x,z$ plane. 

Figure \ref{fig:groupviz}(e) illustrates the minimal-domain subgroup generated by the 
rational phase shift of $\btau(1/5, 3/5)$, namely 
$\Hh{e} = \{1,\, \btau(1/5, 3/5),\, \btau(2/5, 1/5),\, \ldots\, \btau(4/5, 2/5)\}$.
The general form of this type is $\btau(j/p, k/p)$ for prime $p$ and integers $j,k$, 
satisfying $1<j,k<p$.
This generates a minimal-domain diagonal subgroup similar to $\langle \tau(1/n, 1/n) \rangle$
or $\langle \tau(1/n, (n-1)/n) \rangle$, since each of $j,k$ are relatively prime to $p$, 
so that $\ell j \!\!\! \mod p \equiv 0$ and $\ell k \!\!\! \mod p \equiv 0$ only for the
same values of $\ell$, namely multiples of $p$. Note the difference from the subgroups
generated by $\btau(j/n, k/n)$ for nonprime $n$, which are minimal-domain only if $j$
and $k$ are 1 or $n-1$ (or an integer congruent to these mod $n$). This difference
is related to the fact that the integers mod $p$ form a Galois field only when $p$
is prime. 

Figure \ref{fig:groupviz}(f) illustrates a nonminimal-domain subgroup generated by
$\btau(1/2, 2/3)$, namely 
$\Hh{f} = \langle \btau(1/2, 2/3) \rangle = \{1,\, \btau(1/2, 2/3),\, \btau(0, 1/3),\, \btau(1/2, 0), \ldots, \btau(1/2, 1/3)\}$. 
This subgroup is non-minimal due to the presence of elements 
$\btau^2(1/2, 2/3) = \btau(0, 1/3)$, 
$\btau^3(1/2, 2/3) = \btau(1/2, 0)$, and
$\btau^4(1/2, 2/3) = \btau(0, 2/3)$,
each of which appears in the figure as a blue dot on the $a$ or $b$ axis away from the origin. 
The general form of this type is $\langle \btau(j/m, k/n) \rangle$ for unequal integers $m,n>1$
and nonzero integers $j,k$ relatively prime to $m,n$ respectively. Such subgroups are non-minimal because
the integers modulo a composite number $mn$ do not form a Galois group. 

Figure \ref{fig:groupviz}(g) illustrates a minimal-domain subgroup generated by an element 
of form $\sy \btau(1/n,1/n)$, namely $\Hh{g} = \langle \sy \btau(1/4,1/4) \rangle = \{1,\, \sy \btau(1/4, 1/4),\, \btau(1/2, 1/2),\, \sy \btau(3/4, 3/4)\}$. 
The structure of this group is similar to that of $\Hh{c} = \langle \btau(1/3, 1/3) \rangle$, 
except the diagonal elements alternate between pure phase shifts and phase shifts with $y$
reflection. 

Figure \ref{fig:groupviz}(h) illustrates a minimal-domain subgroup generated by an element 
of form $\sy \btau(j/n, k/n)$ for $j,k$ relatively prime to natural number $n$. The
example shown is $\Hh{h} = \langle \sy \btau(1/10, 3/10) \rangle$. Since $(\sy \btau(1/10, 3/10))^2 = \btau(1/5, 3/5)$,
the generator of example (e), it follows that $\Hh{e}$ is a subgroup of $\Hh{h}$. 
This is apparent from subplot (h) containing all the markers of subplot (e). 

Figures \ref{fig:groupviz}(i,j) illustrate the simplification of a non-$z$-centered, nonminimal-domain
subgroup to $z$-centered, minimal-domain form. Figure \ref{fig:groupviz}(i) shows the 
subgroup 
$\Hh{i} = \langle \sz \btau(1/6, 4/7) \rangle = \{1,\, \sz \btau(1/6, 4/7),\, \btau(1/3, 0),\,  \sz \btau(1/2, 4/7),\, \btau(2/3, 0),\, \sz \btau(5/6, 4/7)\}$. This subgroup is non-$z$-centered since it has elements of the
form $\sz \btau(a,b)$ with $b\neq0$, but no element of form $\sz \btau(a,0)$. It can
be transformed to an equivalent subgroup $\Hh{i}' = \gamma \Hh{i} \gamma^{-1}$ for 
$\gamma = \btau(0, 2//7)$. 
This produces the $z$-centered, nonminimal subgroup 
$\Hh{i}' = \langle \sz \btau(1/6, 0) \rangle = \{1, \sz \btau(1/6, 0), \btau(1/3, 0),  \sz \btau(1/2, 0), \btau(2/3, 0), \sz \btau(5/6, 0)\}$ (not shown). 
The presence of $\btau(1/3, 0)$ in both $\Hh{i}$ and $\Hh{i}'$ implies a periodicity of $L_x/3$ in $x$, so 
we further simplify this subgroup by recasting on a domain with this smaller periodicity.
This is accomplished by multiplying the $x$ component of each normalized phase shift by 3
and simplifying, $\sz \btau(1/6, 0) \rightarrow \sz \btau(1/2, 0)$
and $\btau(1/3, 0) \rightarrow \btau(3/3, 0) = 1$, etc., resulting in the 
$z$-centered, minimal-domain subgroup $\Hh{j} = \langle \sz \btau(1/2, 0) \rangle = \{1, \sz \btau(1/2, 0)\}$
shown in fig.\ \ref{fig:groupviz}(j). 

Figures \ref{fig:groupviz}(k) in comparison to (h) illustrates the substantial difference
in the behavior of $z$- and $y$-reflection symmetries $\sz$ and $\sy$ in combination
with phase shifts $\btau(a,b)$. Figure \ref{fig:groupviz}(k) shows 
$\Hh{k} = \langle \sz \btau(1/10,3/10) \rangle$, whose generator differs from that
of $\Hh{h} = \langle \sy \btau(1/10,3/10) \rangle$ only by the replacement of $\sz$ for $\sy$.
In contrast to $\Hh{k}$, $\Hh{h}$ is both non-$z$-centered and nonminimal-domain. It can be
simplified to $\Hh{j} = \langle \sz \btau(1/2, 0) \rangle$ by a process similar to the reduction of
$\Hh{i}$.

From these examples we infer and present without proof a few general principles of subgroups
with higher-order elements. 
\begin{enumerate}
\item Any element of $\Gppf$ with rational normalized phase shifts has finite order.
\item The subgroup $H = \langle \btau(j/p, k/p) \rangle$, where $p$ is prime and $j,k$
  are integers $1 \leq j,k < p$, is finite and has $p$ elements ($|H| = p$). 
\item If a subgroup contains a phase shift $\btau(j/m, k/n)$ where $j/m$ and $k/n$ are 
  reduced rationals with natural numbers $m,n$, $m\neq n$, and integers $j,k$ relatively prime
  to $m$ and $n$ respectively, then the subgroup is nonminimal.
\item A subgroup containing an element $\sz \btau(a,b)$ is minimal only for $a=0$ or $1/2$.
\end{enumerate}
These examples and principles provide some structure for understanding the features of the
symmetry subgroups of plane Poiseuille symmetry involving phase shifts other than half the
periodic lengths. They also illustrate a rich variety of possible symmetry groups for velocity
fields and invariant solutions in a doubly-periodic computational domain. 
%----------------------------------%
% \import{sections/}{section-3.tex}
%----------------------------------%
%-----------------------------------%
\section{Nonlinear traveling waves}\label{sec:nl_tws}
%-----------------------------------%
In this section, we start by outlining the numerical
approach and discuss about exploiting connections between imposed symmetries and dynamics. 
We also report the fifteen new numerically calculated traveling wave solutions of plane Poiseuille flow in seven different symmetric subspaces.  
A summary of the properties is given in table \ref{tbl:tws}. 
Lastly, we present streamwise-averaged cross sections and Reynolds-number continuations
of the computed traveling waves and discuss their structure and bifurcations.

%-----------------------------------%
\subsection{Numerical approach}\label{subsec:num_approach}
%-----------------------------------%

We used Channelflow $2.0$ to compute traveling-wave solutions of plane Poiseuille flow on doubly-periodic
domains in a selection of symmetric subspaces. 
Channelflow uses the Newton-Krylov-hookstep algorithm of 
\cite{viswanath2007recurrent} to find solutions of (\ref{eqn:invariance}) while enforcing a given 
symmetry subgroup by projection with its generators, as described in section \ref{sec:preliminary}.
For traveling waves, the invariance equation (\ref{eqn:invariance}) is 
\begin{align} 
\tau(c_x t, c_z t) \, \phi^t(\bu) - \bu = 0 \label{eqn:tw}
\end{align}
for all $t \geq 0$, unknown wavespeeds $c_x, c_z$, and unknown $\bu$. Channelflow solves the related 
equation 
\begin{align} 
\btau(a_x, a_z) \, \phi^T(\bu) - \bu = 0 \label{eqn:twchflow}
\end{align}
where $T$ is fixed, $\btau$ is the domain-normalized form of a phase shift (\ref{eqn:btau}), and the free
variables are  the spectral coefficients of $\bu$ and the domain-normalized phase shifts $a_x, a_z$ on the
periodic unit interval.
The computed phase shifts are nominally related to the wavespeeds by 
$c_x = a_x L_x/T, c_z = a_z L_z/T$. After convergence of a solution we check that the computed values
of $a_x, a_z$ are on the correct branch of the periodic unit interval and produce the correct wavespeeds
by verifying that $\tau(c_x t, c_z t) \phi^t(\bu) - \bu = 0$ for $0 < t \leq T$. 

The control parameters for a computation of a traveling wave are the periodic domain lengths $L_x, L_z$,
the Reynolds number $\Rey$, the fixed integration time $T$, and the symmetry  group $H \leq \Gppf$.
We chose $\Rey = 2000$, $L_x, L_z = 2\pi, \pi$ or $\pi, \pi/2$, and $T=5$. For $H$ we selected several 
of the subgroups derived in section \ref{sec:halfbox_groups}, namely $H = \langle \sy \rangle$,
$\langle \sz,\, \txz \rangle$, $\langle \sy,\, \txz \rangle$,  $\langle \sy,\, \sz,\, \txz \rangle$,
$\langle \sy,\, \sz \rangle$, $\langle \sy,\, \sz\tx \rangle$, and $\langle \tau(L_x/3, L_z/3)\rangle$.
All numerical computations were performed at the discretization of $N_x \times N_y \times N_z = 48 \times 81 \times 48$,
with dealiasing in periodic directions performed using the 3/2 dealiasing rule. Time marching was done using a third-order
semi-implicit backward differentiation scheme. We confirmed the adequate numerical resolution by 
recomputing each solution at higher resolution of $96 \times 97 \times 96$.%$96 \times 97 \times 96$.

To generate initial guesses, we generated random divergence-free initial conditions $\bu(0)$ with
exponentially decaying spectral coefficients to mimic the smoothness of transitionally turbulent
velocity fields. We projected the random field onto the given symmetric subspace and then rescaled it
to have sufficient magnitude to evolve into sustained unsteady flow. We generated time series data by
time integration, $\bu(t) = \phi^t(\bu(0))$, with the symmetry group $H$ imposed by periodically projecting
with its generators, as discussed in section \ref{sec:preliminary}. 
Since traveling waves are moving steady states, initial guesses for the traveling-wave search were
selected by scanning the time-series data for states with a close balance of energy input and dissipation
rates. Following \cite{toh2003periodic} and
\cite{park2015exact} we define these rates as
\begin{align}
 I &= \frac{1}{2L_{z}}\int_{0}^{L_{z}}\int_{-1}^{1} ( \ptot\utot|_{x=0} \;-\; \ptot\utot|_{x=L_{x}} ) \; dy \, dz, \label{eqn:I} \\
 D &= \frac{1}{2 L_{x} L_{z}} \int_{0}^{L_{z}}\int_{-1}^{1} \int_{0}^{L_{x}} (|\bsym{\nabla} \utot|^{2} + |\bsym{\nabla} v_{\textrm{tot}}|^{2} + |\bsym{\nabla} w_{\textrm{tot}}|^{2})\; dx\, dy\, dz. \label{eqn:D}
\end{align}
We selected initial guesses $\hbu$ from time-series states $\bu(t)$ that satisfied $| D - I | \leq 2 \times 10^{-3}$.
Guesses $\hat{a}_x, \hat{a}_z$ for domain-normalized phase shifts were then obtained by minimizing 
$\| \btau(\hat{a}_x, \hat{a}_z) \phi^T(\hbu) - \hbu\|$ over $\hat{a}_x, \hat{a}_z$ in $[-1/2, 1/2)$.
We fixed the integration times to $T=5$, which slowed the rate of convergence of the Newton-Krylov
iterations, in that it sometimes took more than $100$ iterations for the Newton step to converge. 
However, the large number of iterations was partially counterbalanced by the shorter integration time.
For larger values of $T$, we rarely observed any success. We also tried initial guesses from states
that lingered around a point in the $D$ versus $I$ plane for $10$ time units or more. Our overall success
rate was low: approximately one in ten initial guesses converged to a traveling wave.

As noted in section \ref{sec:preliminary}, a $z$-reflection symmetry pins the phase of a velocity
field to a center in $z$, thus precluding traveling waves from traveling in $z$ and
fixing the spanwise wavespeed to $c_z=0$. 
This property holds for any symmetry subgroup with a generator
that has a factor of $\sz$. Let $\bu$ be a traveling-wave solution of (\ref{eqn:tw}) with symmetry
subgroup $H$ with generator $h = \sy^j \sz \tau(a,b)$ where $j =0$ or 1. Since $h \in \Gppf$ it
commutes with $\phi^t$ (\ref{eqn:invariance}), so 
\begin{align}\label{eqn:commute}
 h \phi^t(\bu) = \phi^t(h \bu) 
\end{align}
for all $t \geq0$. From \ref{eqn:tw} we have $\phi^t(\bu) = \tau(-c_x t, -c_z t)\, \bu$. Substituting 
this and the value of $h$ into (\ref{eqn:commute}) gives
\begin{align*}
   \sy^j\, \sz\, \tau(a,b)\, \tau(-c_x t, -c_z t)\, \bu  &= \tau(-c_x t, -c_z t)\, \sy^j\, \sz\, \tau(a,b)\, \bu \\
                                                         &= \sy^j\, \sz\, \tau(-c_x t, c_z t)\, \tau(a,b)\, \bu
\end{align*}
Multiplying this from the left by $\tau(c_x t, c_z t)\, \tau(-a,-b)\, \sz\, \sy^j$ gives
\begin{align} 
\bu &= \tau(0, 2c_z t)\, \bu \label{eqn:twpinning}
\end{align}
for all $t \geq0$. By this equation, either $c_z = 0$ or $\bu$ is independent of $z$, in which case $c_z$
is arbitrary and can be taken as $c_z = 0.$ Thus, if the symmetry subgroup $H$ has a generator with a factor
of $\sz$, we compute traveling waves by solving \ref{eqn:tw} with $a_z=0$ held fixed. 
              
For relative periodic orbits, the spanwise pinning effect is slightly different. The invariance equation 
(\ref{eqn:invariance}) in this case is
\begin{align}\label{eqn:rpo}
\tau(\ell_x, \ell_z) \, \phi^T(\bu) - \bu = 0
\end{align}
for unknown $\ell_x, \ell_z, T$, and $\bu$, where $T > 0$. Let $\ell_x, \ell_z, T$, and 
$\bu$ be a solution of (\ref{eqn:rpo}), and let the symmetry group of $\bu$ have a generator of form 
$h = \sy^j \sz \tau(a,b)$ where $j=0$ or 1. 
From (\ref{eqn:rpo}) we have $\phi^T(\bu) = \tau(-\ell_x, -\ell_z) \bu$.
Substituting this and the value of $h$ into (\ref{eqn:commute}) at $t=T$ gives
\begin{align*}
   \sy^j\, \sz\, \tau(a,b)\, \tau(-\ell_x, -\ell_z)\, \bu  &= \tau(-\ell_x, -\ell_z)\, \sy^j\, \sz\, \tau(a,b)\, \bu \\
                                                         &= \sy^j\, \sz\, \tau(-\ell_x, \ell_z t)\, \tau(a,b)\, \bu
\end{align*}
Multiplying this from the left by $\tau(\ell_x, \ell_z)\, \tau(-a,-b)\, \sz\, \sy^j$ and simplifying gives
\begin{align}\label{eqn:rpopinning}
\bu &= \tau(0, 2 \ell_z)\, \bu 
\end{align}
If the minimal spanwise periodic length is $L_z>0$, we then have $\ell_z=0$ or $L_z/2$. 
Unlike the traveling-wave case where (\ref{eqn:twpinning}) holds for continuous $t$, here 
we cannot rule out the nonzero value of the phase shift. 

The appearance of a zero spanwise wavespeed in a traveling-wave computation or a spanwise phase shift
of $L_z/2$ for a relative periodic orbit can be used to deduce a $z$ reflection symmetry that was not
enforced in the search but instead arose as a property of the solution. For example, the traveling wave
$\text{P}2$ of \cite{park2015exact} with $\sy$ symmetry has spanwise wavespeed $c_z=0$; thus we conjecture
that it also has a $z$ reflection symmetry of generic form $h = \sy^k \sz \tau(a,b)$, which can be
transformed to a $z$-centered symmetry $h' = \sy^k \sz \tau(a,0)$ by an appropriate phase shift
$\bu' = \tau(0,b/2) \bu$, $h' = \tau(0,-b/2) \, h \, \tau(0,b/2)$ following Theorem 1. The appropriate
value of $b$ can be found numerically, for example, by minimizing $\| \sy^k \sz \tau(a,b) \bu -\bu\|$ over 
$a,b$ and $k=0,1$. Similarly, a solution computed on a given $L_x,L_z$-periodic domain can exhibit a 
$\tau(a,0)$ or $\tau(0,b)$ symmetry for $0 < a < L_x$, $0 < b < L_z$, implying periodicity on a smaller
domain. This occurred in our computations of $\TW_{9}$ and $\TW_{12}$ and is furthered discussed in 
section \ref{subsec:describe_tw}.
%-----------------------------------%
\subsection{Nonlinear traveling waves: structures and bifurcations}\label{subsec:describe_tw}
%-----------------------------------%

\begin{table}
  \begin{center}
\def~{\hphantom{0}}
  \begin{tabular}{l|rccS[table-format=1.5]cS[table-format=1.6]cccc}
      TW  & $L_{x}$   &   $L_{z}$ && $c_{x}$ & & $c_{z}$ & symmetry group \\ [3pt]
      \hline
      $\TW_{1}$ & $\pi$ & $\pi/2$ &&  0.1085 &&-0.00226 & $\langle\sy\rangle$ \\
      $\TW_{2}$ & $\pi$ & $\pi/2$ &&  -0.035 && 0 & $\langle\sy, \sz \tx \rangle$ \\
      $\TW_{3}$  & $\pi$ & $\pi/2$ &&  0.06 && 0.004 & $\langle\sy\rangle$ \\
      $\TW_{4}$  & $\pi$ & $\pi/2$ && 0.1265  && 0.0008 & $\langle\sy\rangle$ \\
      \hline
      $\TW_{5}$  & $2\pi$ & $\pi$ && -0.5559 && 0 & $\langle \sz, \txz \rangle$ \\
      $\TW_{6}$  & $2\pi$ & $\pi$ && -0.5587 && 0 & $\langle \sz, \txz \rangle$ \\
      \hline
      $\TW_{7}$ & $2\pi$ & $\pi$ && -0.53391 && 0.000245 & $\bsym{\langle\sy,\, \txz\rangle}$ \\
      %
      % (630-05) & $2\pi$ & $\pi$ & -0.51454 & 0.000198 & $\langle\txz, \sy\rangle$ \\
      %
      $\TW_{8}$  & $2\pi$ & $\pi$ && -0.51454 && 0.000198 & $\bsym{\langle\sy,\, \txz\rangle}$ \\
      \hline
      $\TW_{9}$ & $\pi$ & $\pi$ && -0.55543 && 0 & $\langle\sy, \sz, \txz\rangle$\\
      $\TW_{10}$ & $\pi$ & $\pi$ && -0.49283 && 0 & $\langle\sy, \sz\rangle$\\
      \hline
      $\TW_{11}$  & $2\pi$ & $\pi$ &&  -0.43136 && 0 & $\langle \sy, \sz\tx \rangle$\\
      $\TW_{12}$  & $2\pi$ & $\pi/2$ && -0.45775 && 0 & $\langle \sy, \sz\tx \rangle$\\
      $\TW_{13}$  & $2\pi$ & $\pi$ && -0.55467 && 0 & $\langle \sy, \sz\tx \rangle$\\
      $\TW_{14}$ & $2\pi$ & $\pi$ && -0.53339 && 0  & $\langle \sy, \sz\tx \rangle$\\
      \hline
      $\TW_{15}$ &  $2\pi$ & $\pi$ && -0.50449 && 0.000365 & $\bsym{\langle \tau(L_x/3, L_z/3)\rangle}$ \\ 
      \hline
      \hline
  \end{tabular}
\caption{{\bf Properties of numerically calculated traveling waves}. 
Fifteen newly-found traveling waves of plane Poiseuille flow are summarized in terms of their periodic domains,
wavespeeds, and symmetry groups. The wavespeeds are reported at $\Rey=2000$ on the solution's lower branch under
continuation in $\Rey$. Boldfaced symmetry groups indicate hitherto unexplored invariant spaces.}
  \label{tbl:tws}
  \end{center}
\end{table}

Table \ref{tbl:tws} summarizes the nonlinear traveling waves we found by the methods of \ref{subsec:num_approach}. 
$\TW_{1}$ through $\TW_{4}$ were computed for $L_x, L_z = \pi, \pi/2$.
All the other solutions were computed for $L_x, L_z = 2\pi,\pi$. $\TW_{9}$ was computed on this domain but
then found to have an unimposed $\tx = \tau(L_x/2,0)$ symmetry, and was then recast to its minimal $\pi \times \pi$
domain. Similarly, $\TW_{12}$ was computed on a $2\pi,\pi$ domain, was found to have $\tz = \tau(0,L_z/2)$
symmetry, and was recast onto its minimal $2\pi, \pi/2$ domain. A striking fact reflected in 
table \ref{tbl:tws} is the dominance of streaky flows. Streamwise speeds ($c_x$) are generally orders of magnitude
larger than the spanwise speeds ($c_z$) owing to the strong driving streamwise bulk velocity and the streaky flows
that result.

\begin{figure} 
  \begin{center}
    (a) \includegraphics[scale = 0.34]{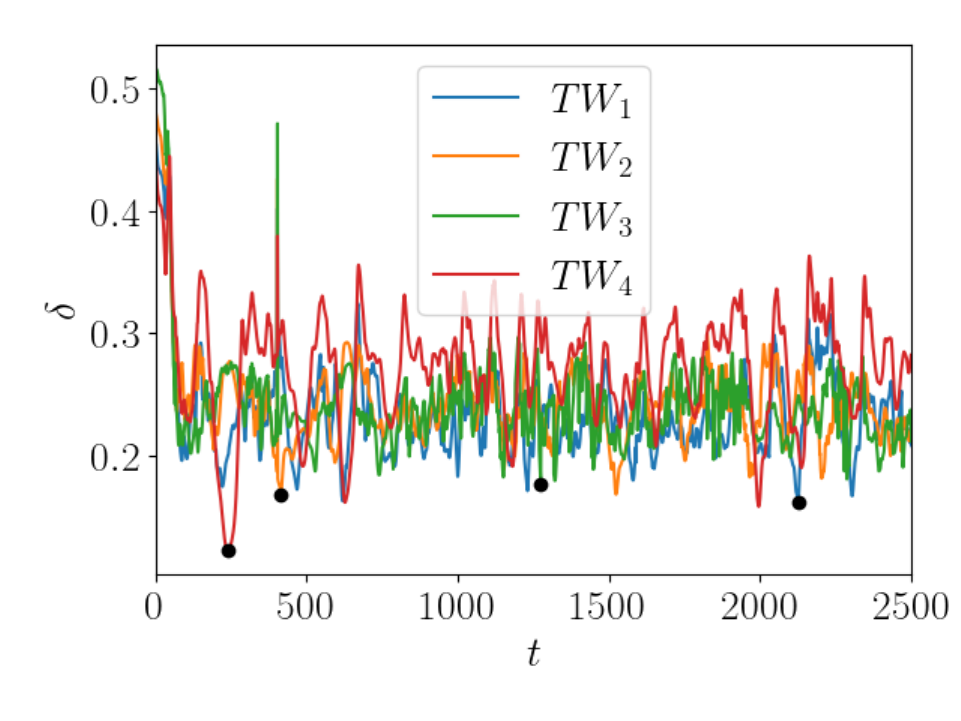}
    (b) \includegraphics[scale = 0.34]{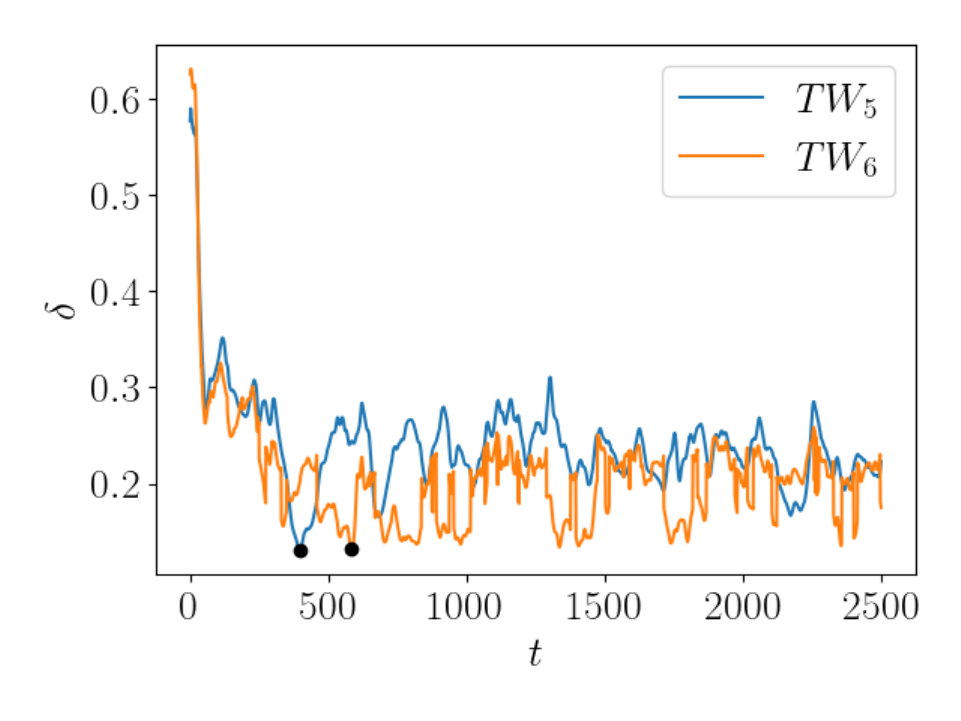} \\
    (c) \includegraphics[scale = 0.34]{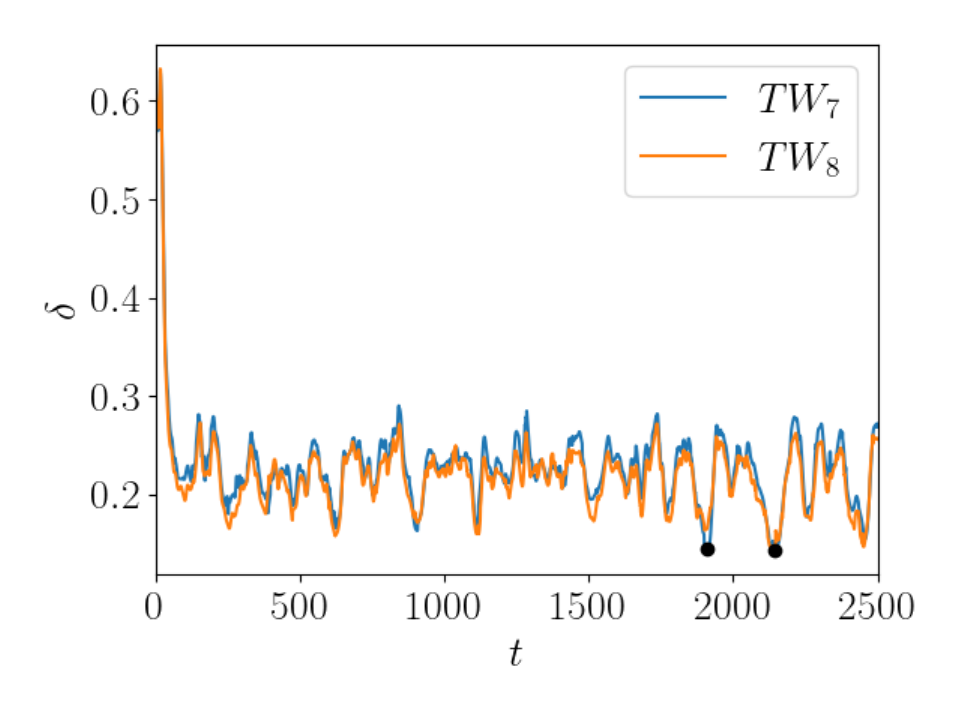}
    (d) \includegraphics[scale = 0.34]{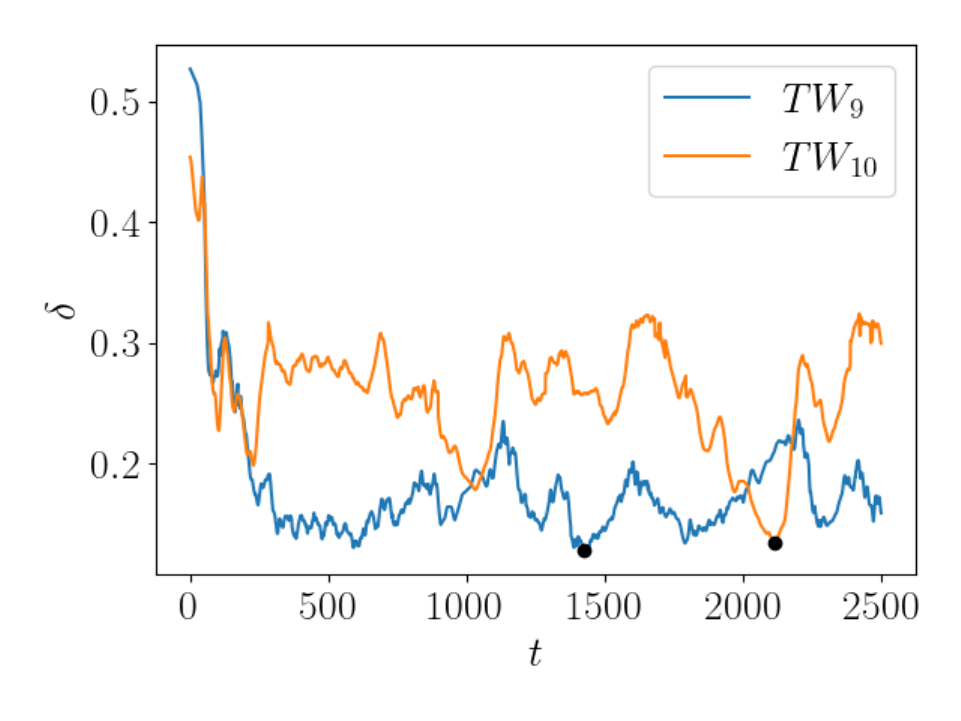} \\
    (e) \includegraphics[scale = 0.34]{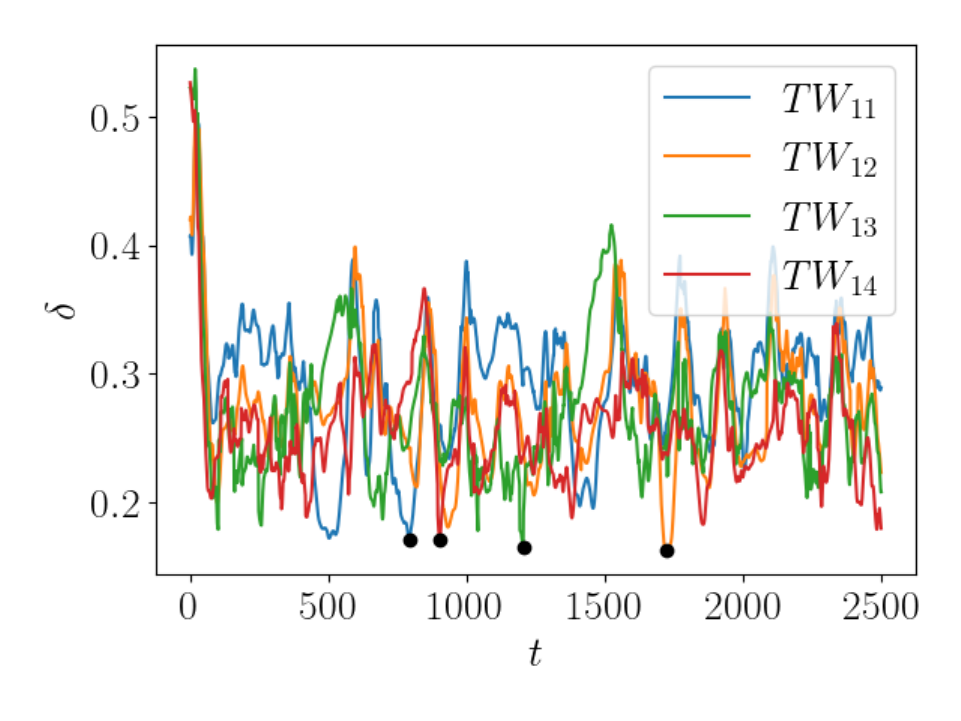} 
    (f) \includegraphics[scale = 0.34]{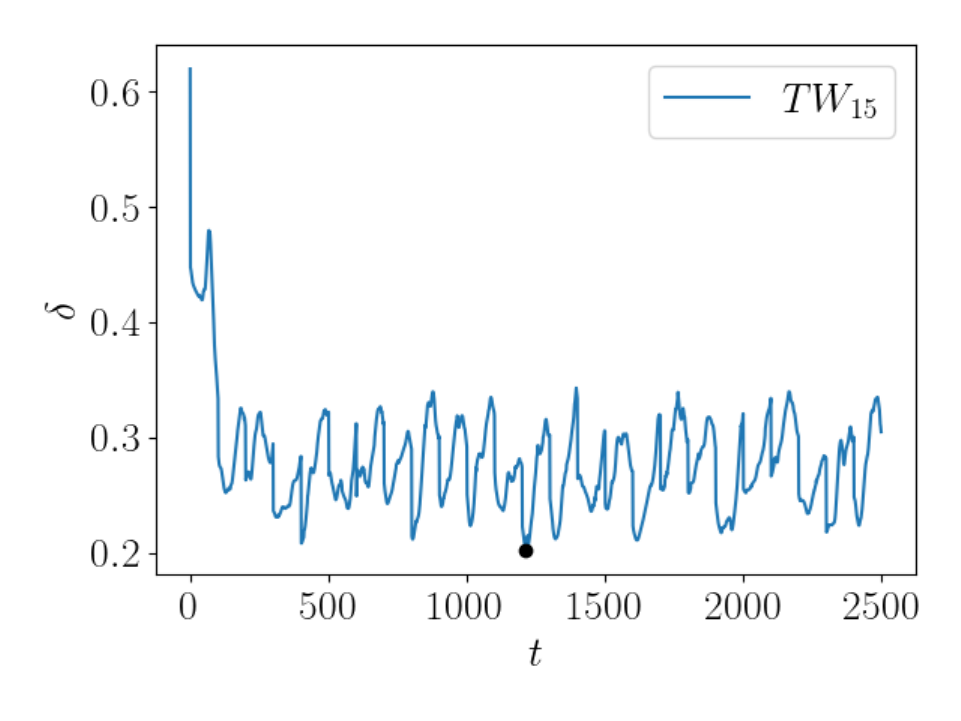} 
  \end{center}
  \caption{{\bf Close passes of turbulent trajectories to traveling waves}
    illustrated by the normalized distance $\delta (t)$ versus $t$ between a turbulent trajectory and the
    given traveling waves in the  (a) $\langle \sy \rangle$, (b) $\langle \txz, \sz \rangle$,
    (c) $\langle \txz, \sy \rangle$, (d) $\langle \sz, \sy \rangle$, (e) $\langle \sz \tx, \sy \rangle$, and 
    (f) $\langle \tau(L_x/3, L_z/3)\rangle$ symmetric subspaces. Black dots denote locations of the closest 
    passes to each traveling wave.}
  \label{fig:l2dist}
\end{figure}

To address the dynamical importance of the traveling waves, we calculate the norm of the difference between turbulent
trajectories in the invariant subspaces and each of the traveling-wave solutions with the same symmetries.
The minimum distance between a field taken from a turbulent trajectory $\bu(t)$ and a traveling wave $\bu_{\TW}$
was calculated by optimizing over phase shifts in both streamwise and spanwise directions. Modifying the definition
 of \cite{park2015exact}, we denote this distance by $\delta(t)$, defined as follows:
\begin{equation}\label{eqn:def_delta}
 \delta(t) = \min_{0 \leq a_x,a_z < 1} \frac{\|\bu(x+a_xL_x, y, z+a_zL_z, t) - \bu_{\TW}(x, y,z, t) \|}{\|\bu(t) \|_{\rms}}
\end{equation}
where $\|\bu(t) \|_{\rms}$ is the root mean square of the norm of velocity fields of the turbulent time series. 
Figure \ref{fig:l2dist} shows the normalized $\delta(t)$ for the computed traveling waves; each subfigure illustrates
close passes to the traveling waves found in the same symmetric subspace. $\TW_{2}$ and $\TW_{9}$ each have an additional
symmetry beyond those enforced in the simulation; they are grouped with other traveling waves in the symmetric subspace
in which they were found. The magnitude of $\delta(t)$ for all close passes was $O(10^{-1})$. The significance of this
result is discussed in section \ref{sec:linearized_dynamics}.

Figures \ref{fig:closepass_one} and \ref{fig:closepass_two} show the streamwise-average velocity fields
of the traveling waves and the close passes to them by turbulent trajectories. The close passes reflect
the general structure of the traveling waves they approach, in terms of the locations and strengths 
of mean spanwise $v,w$ rolls, but there are notable differences in moderate- to fine-scale structure.
The close passes shown in figs.\ \ref{fig:closepass_one} and \ref{fig:closepass_two} correspond to the 
black dots in fig.\ \ref{fig:l2dist}.

\begin{figure}
  \begin{center}
    (a) \includegraphics[scale = 0.3]{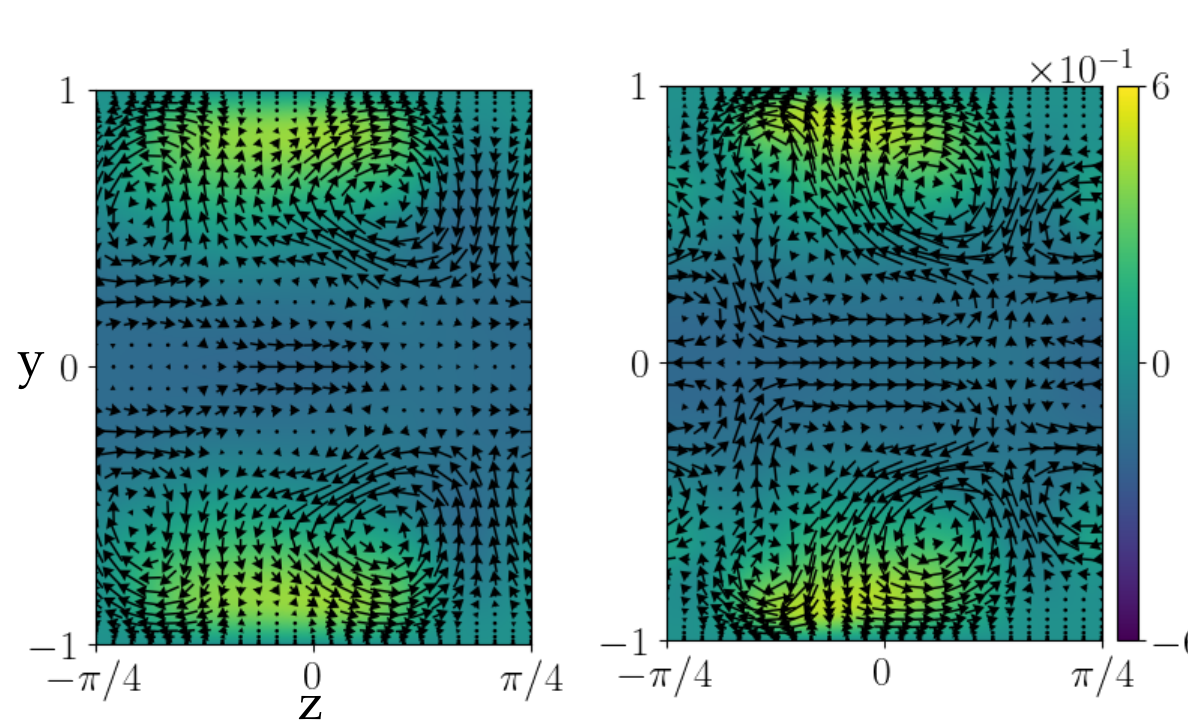}
    (b) \includegraphics[scale = 0.3]{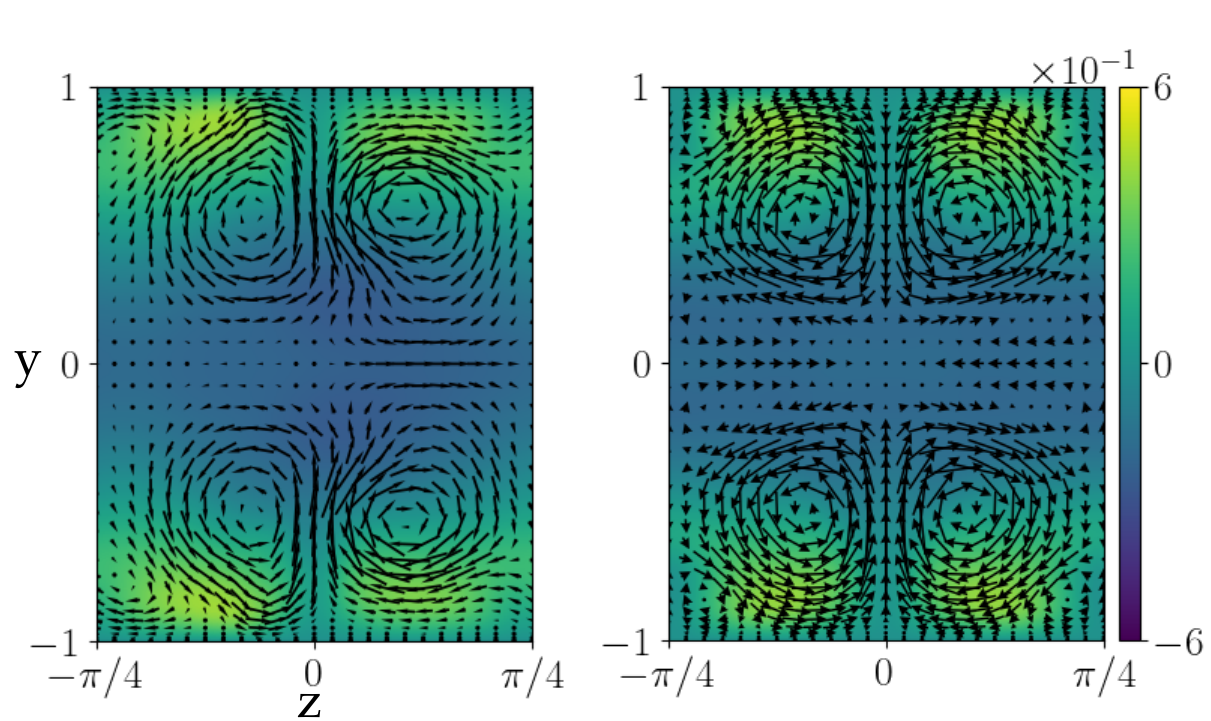} \\
    (c) ~\includegraphics[scale = 0.3]{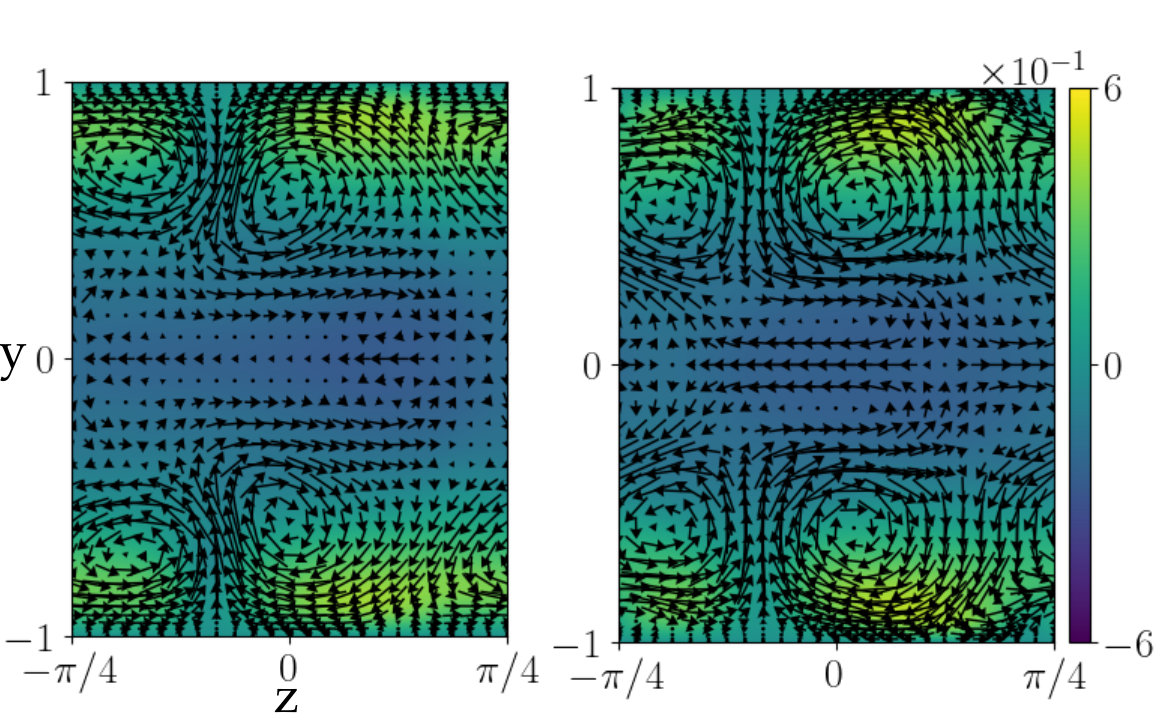} 
    (d) ~\includegraphics[scale = 0.3]{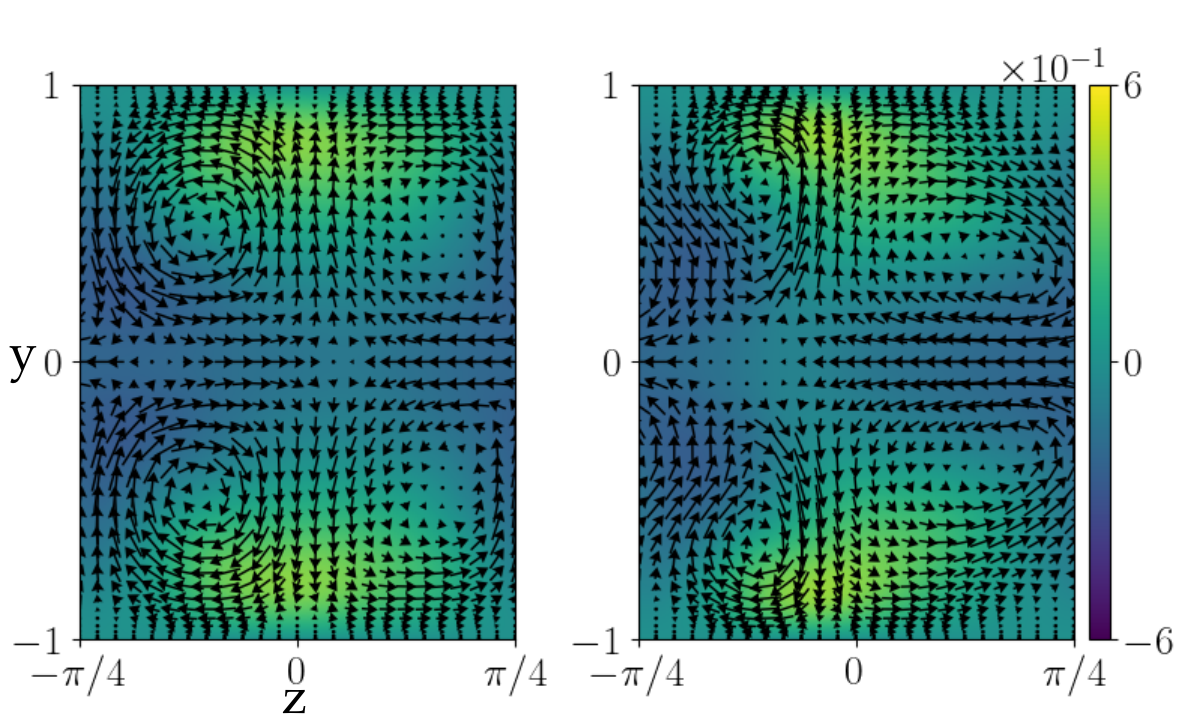}
  \end{center}
  \caption{{\bf Close passes of turbulent trajectories to $\TW_1$--$\TW_4$}
    illustrated by streamwise-averaged $y,z$ cross sections of the velocity fields. 
    (a) $\TW_{1}$, (b) $\TW_{2}$, (c) $\TW_{3}$, (d) $\TW_{4}$.
    In each case the right subplot is the given traveling wave, and the left subplot
    is the the streamwise-averaged cross section of the closest pass of the turbulent
    trajectory, at the point marked with a black dot in figure \ref{fig:l2dist}. Arrows indicate 
    cross-stream $v,w$ velocity, and the colormap indicates the streamwise velocity $u$.
  }
  \label{fig:closepass_one}
\end{figure}

\begin{figure}
  \begin{center}
    (a) \includegraphics[scale = 0.35]{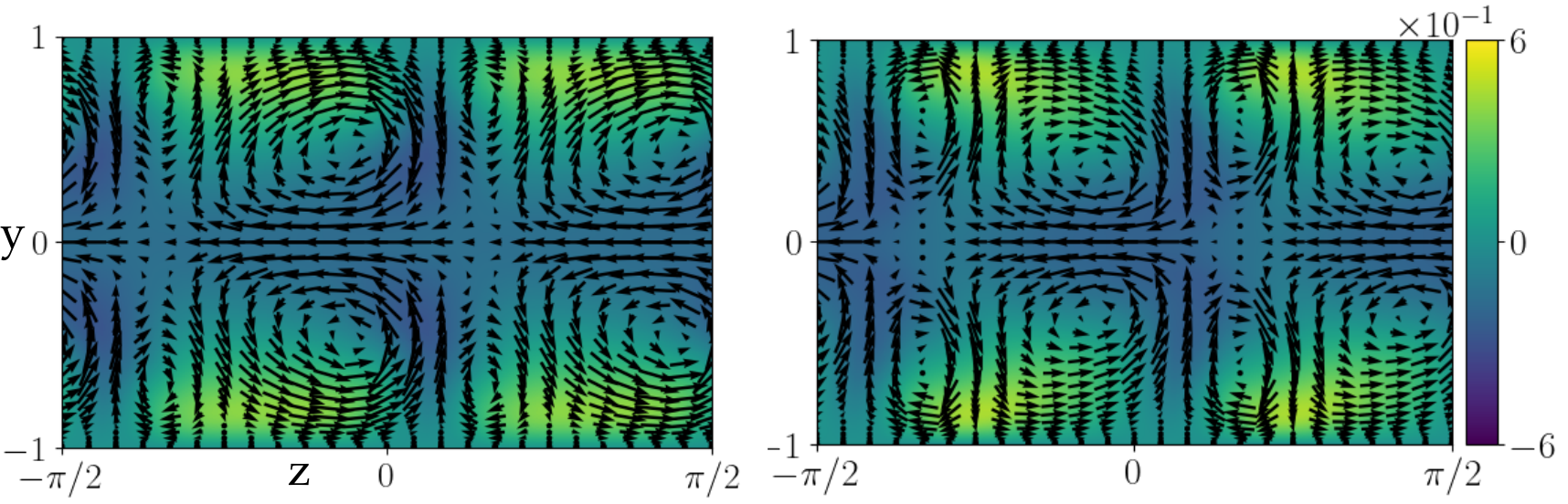}\\
    (b) \includegraphics[scale = 0.35]{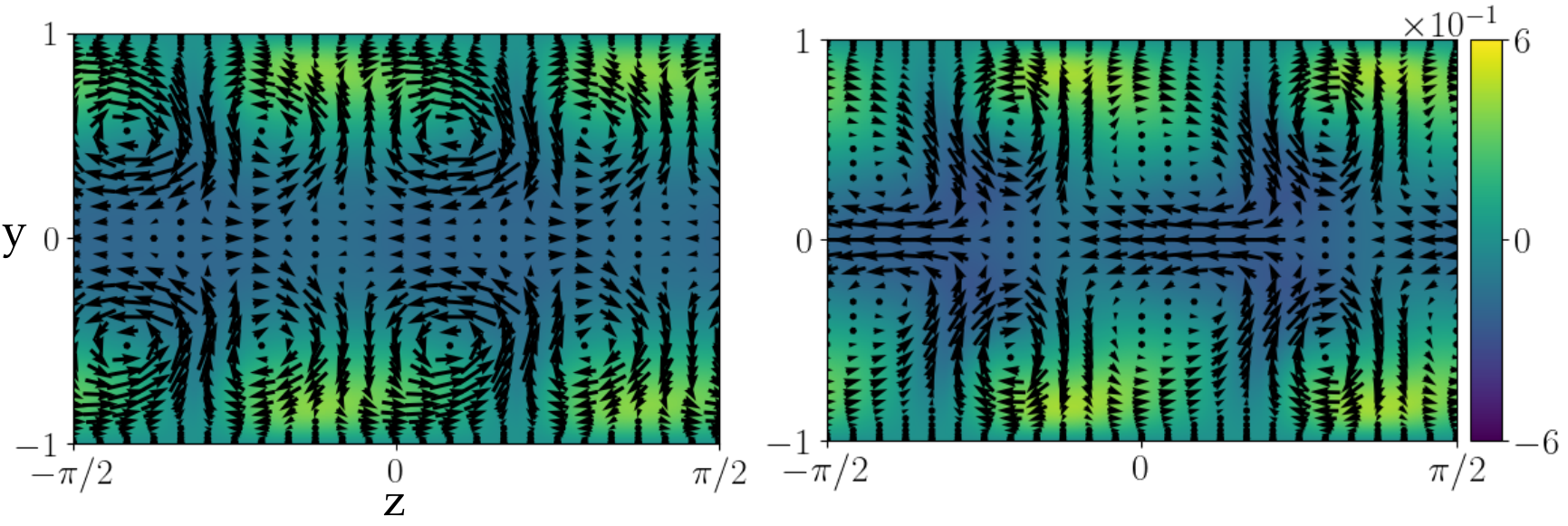}\\
    (c) \includegraphics[scale = 0.35]{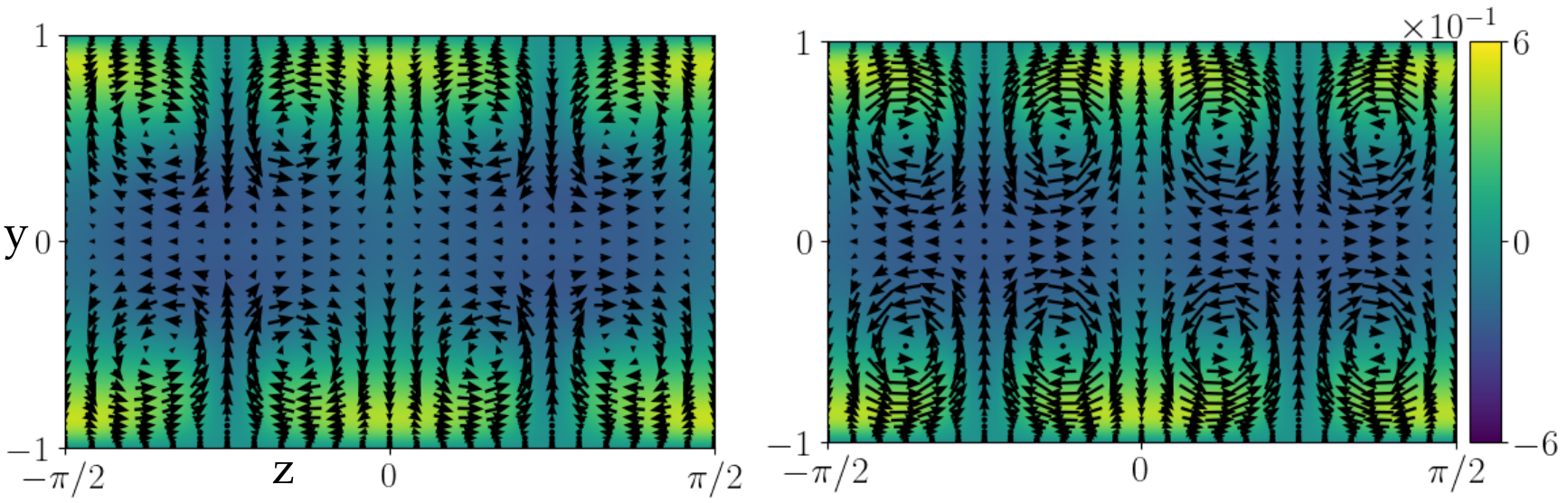} \\
    (d) \includegraphics[scale = 0.35]{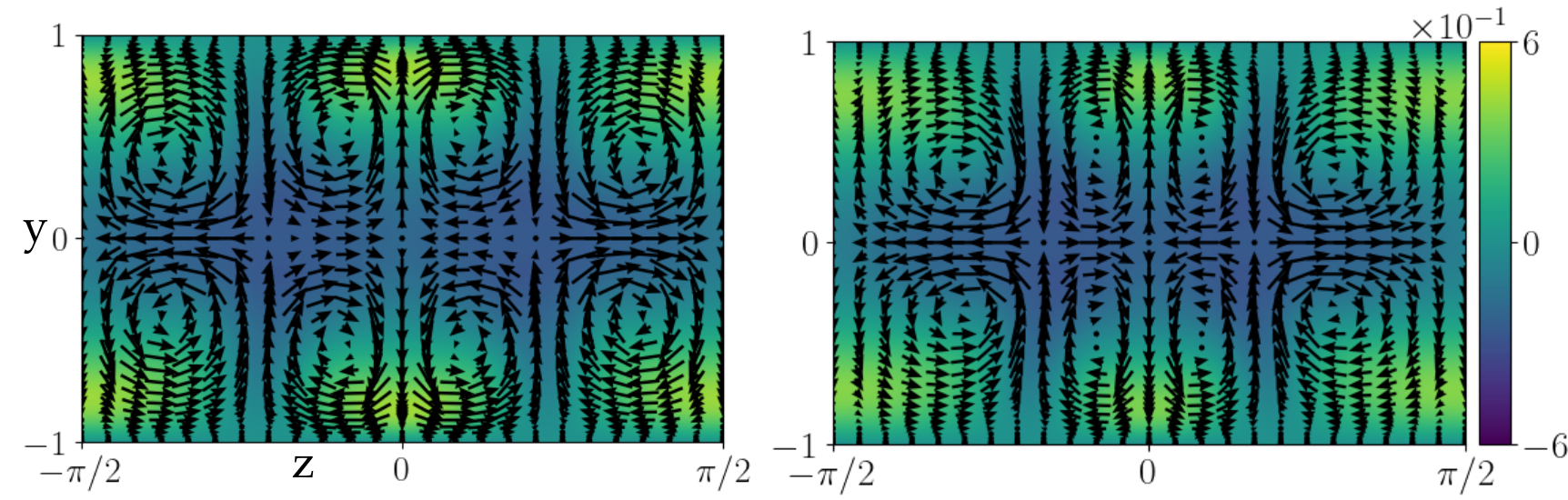} \\
  \end{center}
  \caption{{\bf Close passes of turbulent trajectories to $\TW_7$--$\TW_{10}$}
    illustrated by streamwise-averaged $y,z$ cross sections of the velocity fields. 
    (a) $\TW_{7}$, (b) $\TW_{8}$, (c) $\TW_{9}$, (d) $\TW_{10}$.
    In each case the right subplot is the given traveling wave, and the left subplot
    is the streamwise-averaged cross section of closest pass of the turbulent trajectory,
    at the point marked with a black dot in figure \ref{fig:l2dist}. Arrows indicate 
    cross-stream $v,w$ velocity, and the colormap indicates the streamwise velocity $u$.
    }
  \label{fig:closepass_two}
\end{figure}
\begin{figure}
  \begin{center}
    (a) \includegraphics[scale = 0.35]{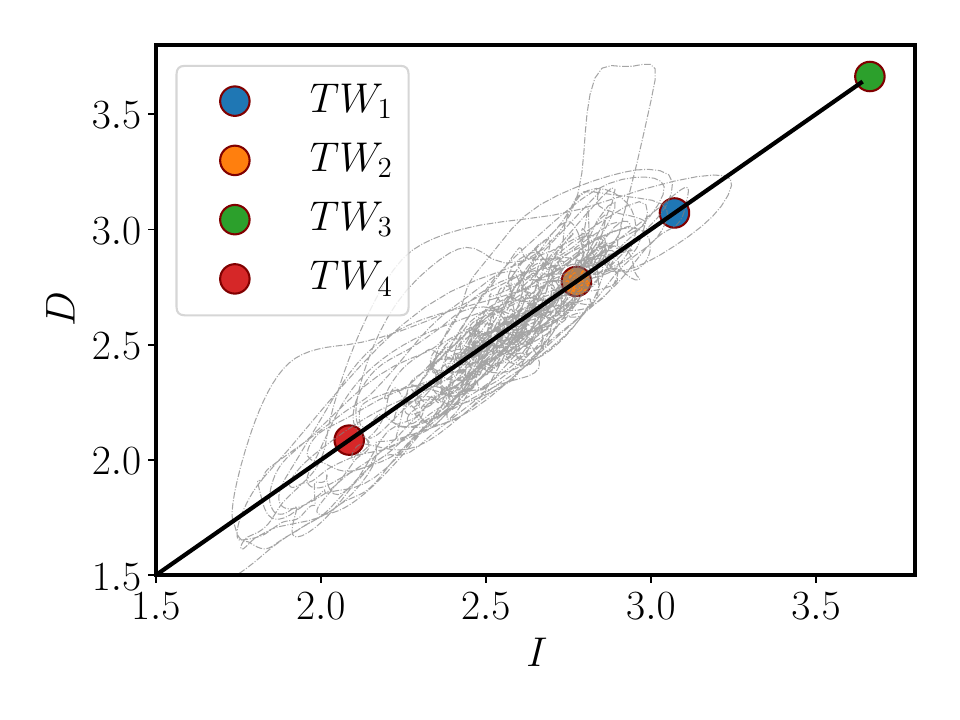}
    (b) \includegraphics[scale = 0.35]{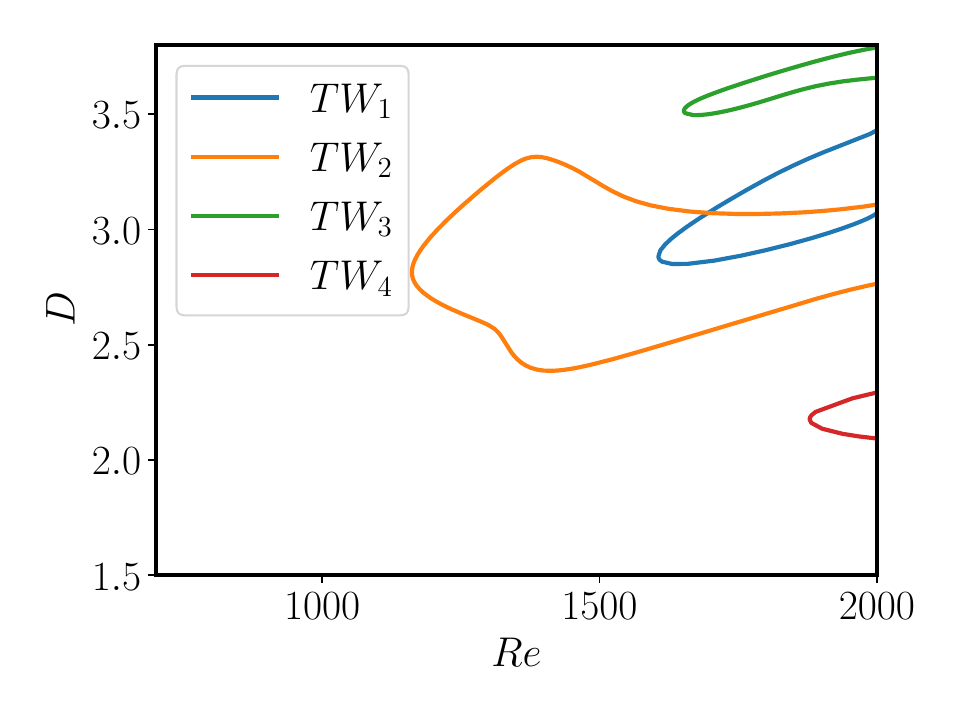} \\
    (c) \includegraphics[scale = 0.35]{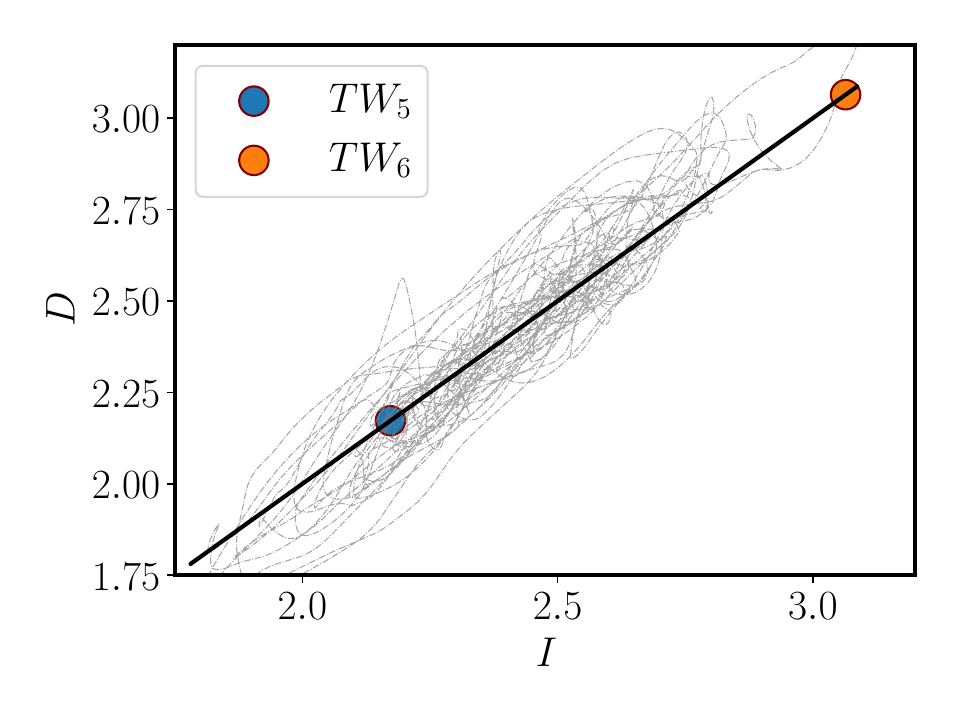} 
    (d) \includegraphics[scale = 0.35]{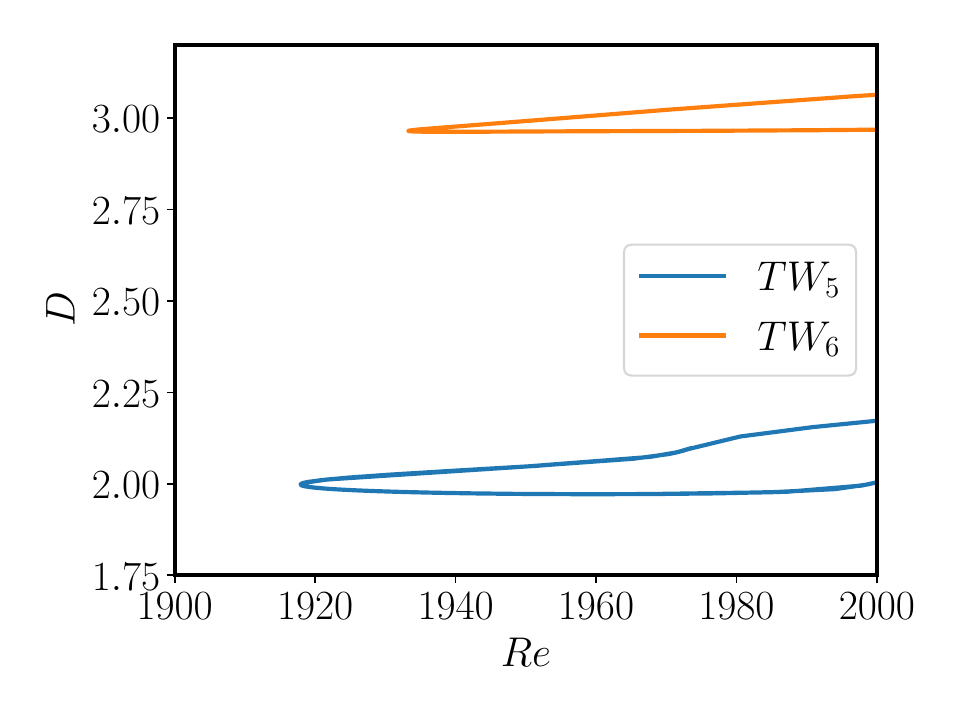}\\ 
    (e) \includegraphics[scale = 0.35]{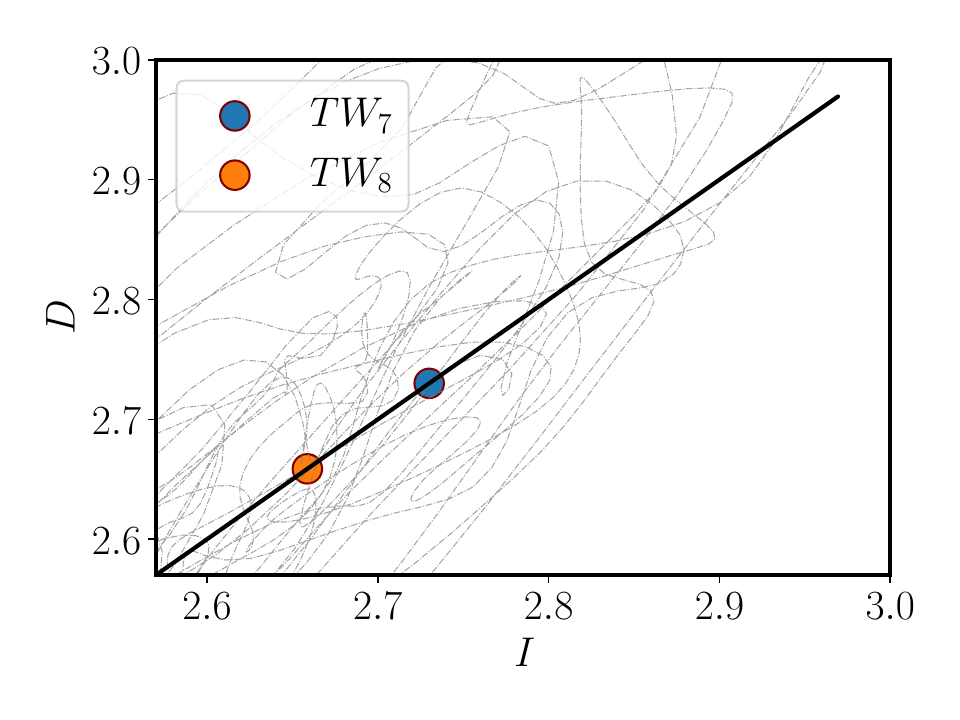} 
    (f) \includegraphics[scale = 0.35]{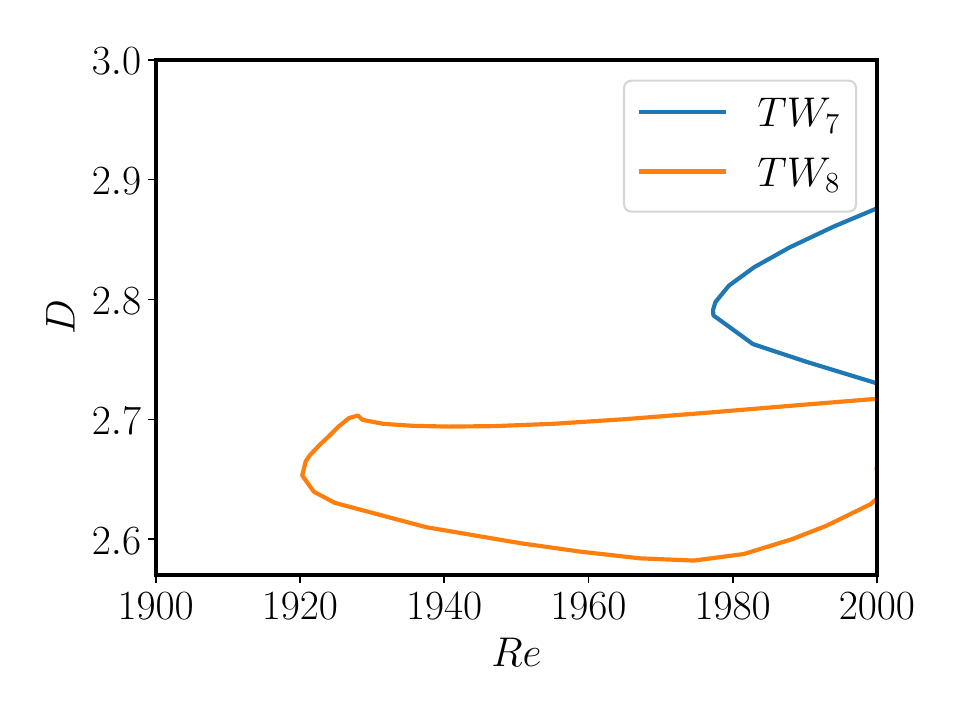}
  \end{center}
  \caption{{\bf Dissipation versus input and continuation curves for $\TW_1$ through $\TW_8$.}
    (left) Dissipation $D$ versus energy input $I$ at $Re = 2000$. Traveling waves are marked with dots.
    The dashed line in each case is a typical turbulent trajectory in the given symmetric subspace. 
    (right) $D$ versus $Re$ continuation curves of the traveling waves. The enforced symmetry groups for 
    each pair of curves are 
    (a,b) $\langle \sy \rangle$,
    (c,d) $\langle \sz,\, \txz \rangle$,
    (e,f) $\langle \sy,\, \txz \rangle$,}
 \label{fig:D-vs-Re-one}
\end{figure}

\begin{figure}
  \begin{center}
    (a) \includegraphics[scale = 0.35]{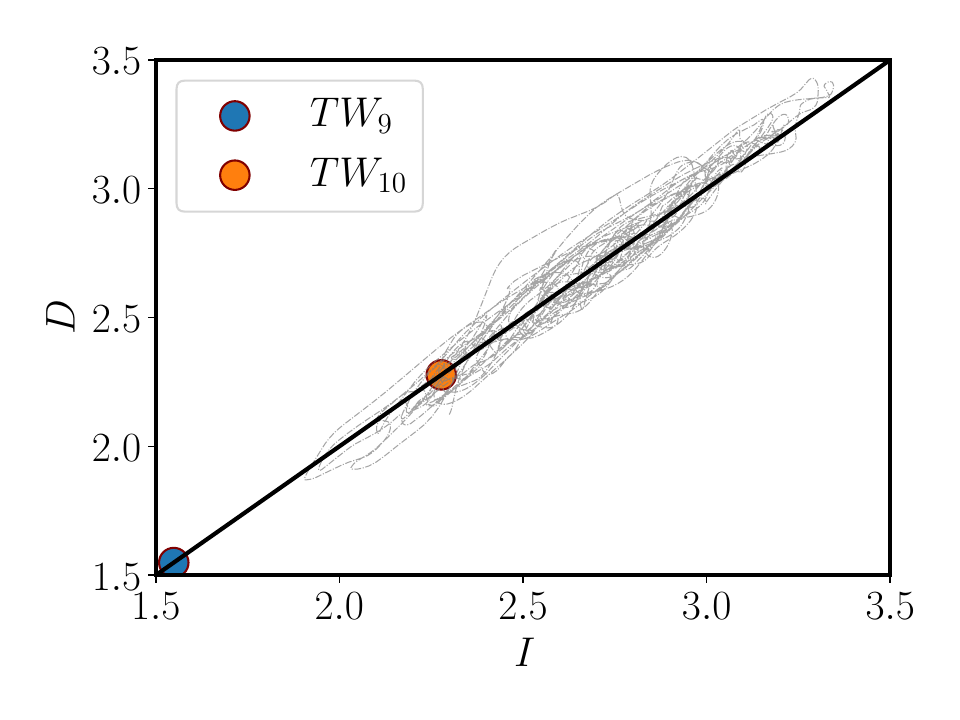}
    (b) \includegraphics[scale = 0.35]{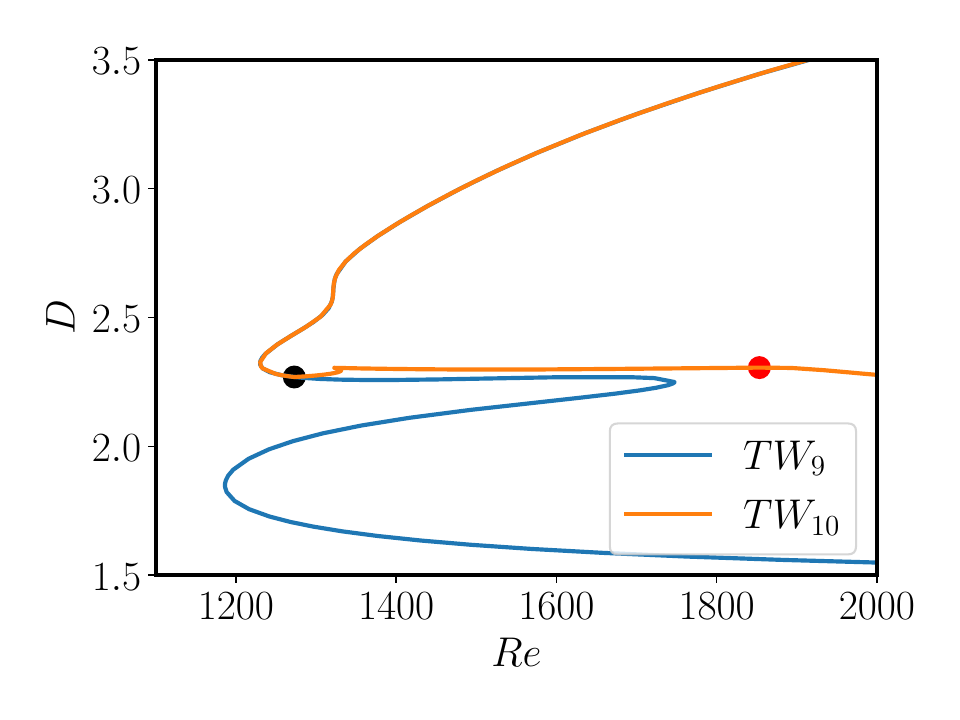} \\
    (c) \includegraphics[scale = 0.35]{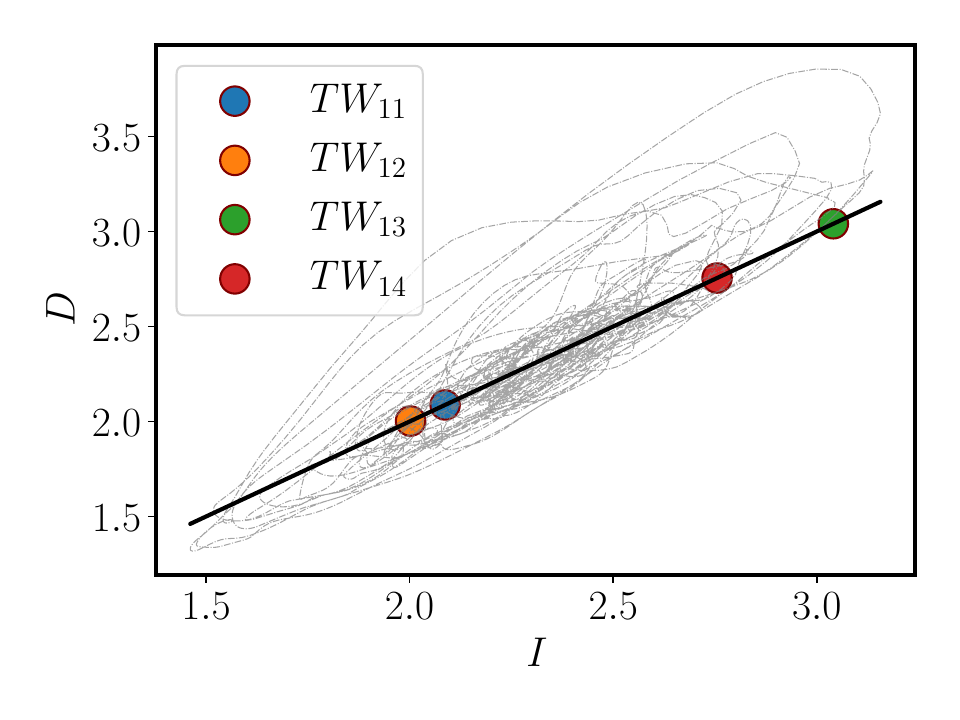} 
    (d) \includegraphics[scale = 0.35]{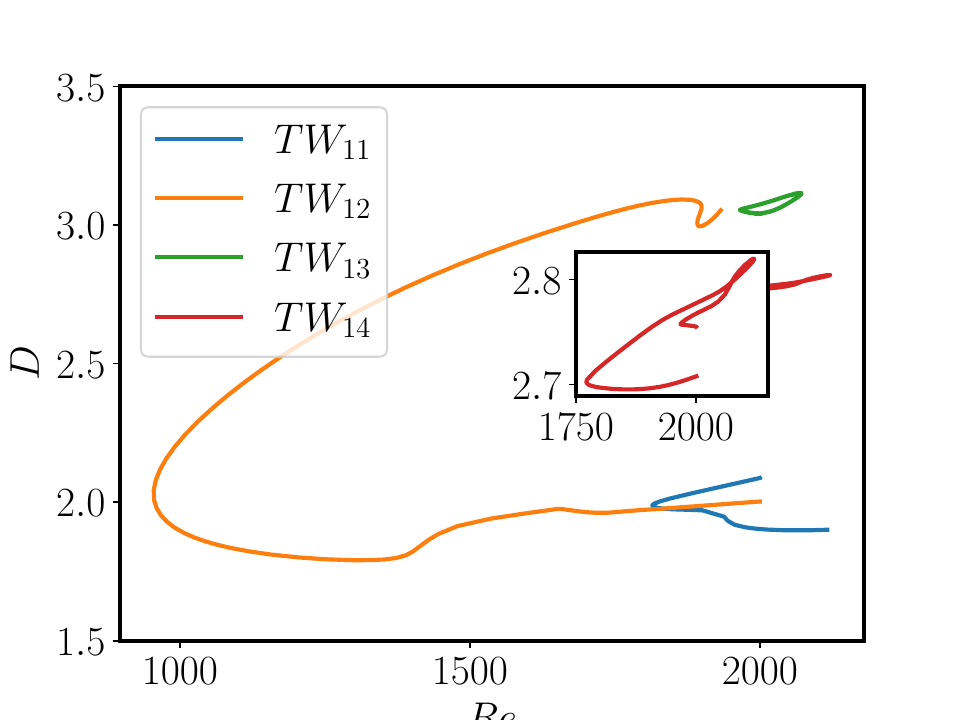}
  \end{center}
  \caption{{\bf Dissipation versus input and continuation curves for $\TW_9$ through $\TW_{14}$.}
    (left) Dissipation $D$ versus energy input $I$ at $Re = 2000$. Traveling waves are marked with dots.
    The dashed line in each case is a typical turbulent trajectory in the given symmetric subspace. 
    (right) $D$ versus $Re$ continuation curves of the traveling waves. The enforced symmetry groups for 
    each pair of curves are 
    (a,b) $\langle \sy,\, \sz \rangle$ 
    (c,d) $\langle \sy,\, \sz\tx \rangle$.
    $\TW_{9}$ has a degeneracy
    as discussed in section \ref{subsec:describe_tw}. At $Re \approx 1273$  (black dot in (b)), $\TW_{10}$ 
    connects to $\TW_{9}$ in a symmetry-breaking bifurcation. This is followed by another symmetry-breaking
    bifurcation at  $Re \approx 1894$ (red dot). See section \ref{subsec:describe_tw} for further details.}
    \label{fig:D-vs-Re-two}
\end{figure}

Figures \ref{fig:D-vs-Re-one} and \ref{fig:D-vs-Re-two} show energy dissipation versus input 
for traveling waves and typical turbulent trajectories (left) and continuation curves for 
traveling waves (right). The dissipation versus versus input curves were computed at $Re=2000$
in the symmetric subspaces and spatial domains in which each traveling wave was found. The 
dots on the $D=I$ lines indicate the dissipation-energy balance of each traveling wave. Each 
of the traveling waves shown is the lower branch of its continuation curve in Reynolds number
shown on the right. It is apparent from the dissipation-input figures that most of these
lower-branch traveling waves are embedded in the chaotic sets of the turbulent trajectories.

Figure \ref{fig:D-vs-Re-two}(b) shows a bifurcation of $\TW_{10}$ from $\TW_{9}$ at $Re \approx 1273$,
marked by a black dot in the figure. We found both these solutions in a search for traveling waves
with $\langle\sy, \sz\rangle$ symmetry on a $L_x, L_z = 2\pi, \pi$ computational domain. After convergence
$\TW_{9}$ was found to have additional symmetries and the full symmetry group $\langle\sy, \sz, \tau(L_x/4, L_z/2)\rangle$.
The $\tau(L_x/4, L_z/2)$ symmetry implies $\tx = (\tau(L_x/4, L_z/2))^2$ symmetry as well, and thus
$\TW_{9}$ is periodic on the half domain $L_x/2, L_z = \pi, \pi$, as reported in table \ref{tbl:tws}.
The $\tau(L_x/4, L_z/2)$ symmetry on the full domain is equivalent to $\txz$-symmetry on the half domain. 
The bifurcation at $Re \approx 1273$ is a symmetry-breaking bifurcation in which $\TW_9$ loses 
it $\txz$-symmetry on the half domain and $\TW_{10}$ is born with $\langle\sy, \sz\rangle$ symmetry
on the half domain and $\langle\sy, \sz, \tx \rangle$ symmetry on the computational domain. 
Another symmetry-breaking bifurcation occurs at $Re \approx 1894$ (red dot in figure \ref{fig:D-vs-Re-two}(b)),
where $\tx$ symmetry is broken and the new solution has symmetry group $\langle\sy, \sz\rangle$. 

$\TW_{12}$ was computed on domain $L_x, L_z = 2\pi, \pi$ with imposed symmetries $\langle \sz \tx, \sy \rangle$.
The converged solution found to have an additional $\tz$ symmetry and symmetry group $\langle \tz, \sz \tx, \sy \rangle$.
After recasting on the half domain $L_x, L_z = 2\pi, \pi/2$ implied by the $\tz$ symmetry, $\TW_{12}$ has 
symmetry group $\langle \sz \tx, \sy \rangle$.
%-------------------------------- 
\subsection{A traveling wave with $\langle \tau(L_x/3, L_z/3)\rangle$ symmetry}\label{sec:tw15}
%-------------------------------- 
Although we primarily focused on finite subgroups of plane Poiseuille flow with second-order elements
and half-box shifts, it is also possible to compute invariant solutions with symmetry groups involving
higher-order elements and non-half-box shifts. One such nontrivial, non-half-box shift symmetry group
is $\langle \tau(L_x/3, L_z/3)\rangle$. We found one traveling wave with this symmetry group, $\TW_{15}$,
listed in table \ref{tbl:tws}.  With no $z$ reflection symmetry, this solution has nonzero streamwise
and spanwise wavespeeds. The streamwise-averaged $y,z$ cross section and continuation in Reynolds
number for $\TW_{15}$ are shown in figure \ref{fig:tw15}.   
% %--------------------------------  
 \begin{figure}
 \begin{center}
 \includegraphics[scale = 0.35]{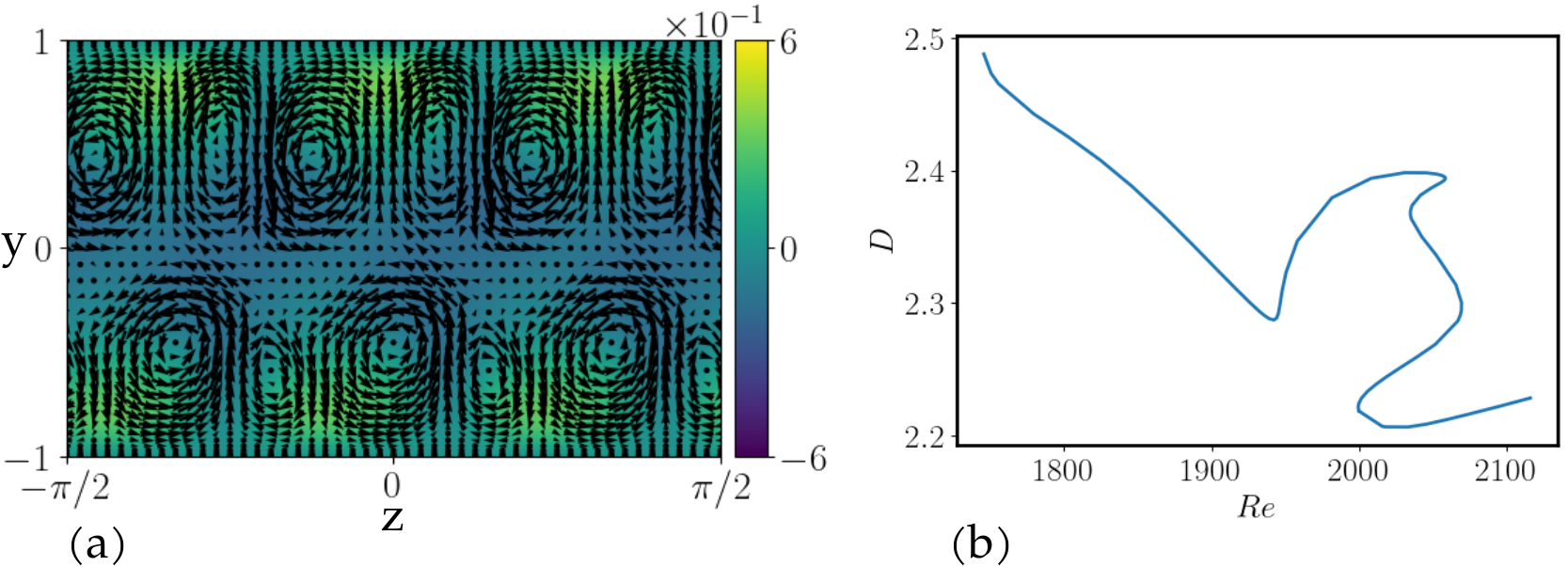}
  \end{center}
\caption{{\bf Streamwise-average velocity and continuation in Reynolds number for 
$\langle \tau(L_x/3, L_z/3) \rangle$-symmetric traveling wave $\TW_{15}$.} 
(a) The streamwise-averaged $y,z$ cross section of $\TW_{15}$ at $Re=2000$. 
Arrows indicate cross-stream $v,w$ velocity. The colormap indicates streamwise
velocity $u$. (b) Continuation $\TW_{15}$ in Reynolds number. 
}
 \label{fig:tw15}
\end{figure}
%--------------------------------    
%----------------------------------%
% \import{sections/}{section-4.tex} 
%----------------------------------%
%-----------------------------------%
\section{Discussion}\label{sec:discussion}
%-----------------------------------%

\subsection{Linearized dynamics and bifurcations} \label{sec:linearized_dynamics}

A key aim of the study of invariant solutions is to model the dynamics of transitionally turbulent flows
in terms of the invariant solutions. One might model transitionally turbulent dynamics as a set of finite-time
hops between the neighborhoods of invariant solutions, governed by the linearization of the solutions' finite-time
dynamics. The results of this study, however, suggest that a much denser set of solutions is required for this 
approach. 
Close passes between turbulent trajectories and traveling waves in the same symmetric subspaces were observed to be $O(10^{-1})$
in root-mean-square normalized magnitude and $O(10^{-2})$ in unnormalized magnitude. This length scale is significantly
larger than the range of accuracy of the linearized finite-time dynamics $\phi^T(\bu + \Delta \bu) \approx \phi^T(\bu) + D\phi^T \Delta \bu$
estimated during the Newton-hookstep search. For each step $\bu \rightarrow \bu+\Delta\bu$, the Newton-hookstep algorithm computes
the Newton step $\Delta\bu_N$ using the linearized dynamics and estimates the trust-region radius $\dtr$ within which the linearization
is accurate. If the Newton step is within the trust region, $\|\Delta\bu_N\| \leq \dtr$, the algorithm takes the Newton step,
$\bu \rightarrow \bu+\Delta\bu_N$. 
If $\|\Delta\bu_N\| > \dtr$, the algorithm computes the optimal hookstep $\Delta\bu_H$ within the trust region and takes the hookstep,
$\bu \rightarrow \bu+\Delta\bu_H$. A typical search takes several slowly-converging hooksteps until the estimated
solution is within the trust region, at which point the algorithm switches to Newton steps and converges quickly. Thus the trust-region
radius $\dtr$ at which the search switches from hooksteps to Newton steps can be taken as an estimate of the radius of validity
of the linearization of the dynamics about the converged solution. The searches for the traveling waves in section \ref{sec:nl_tws}
showed that this radius is typically $\dtr \sim O(10^{-3})$ and sometimes $O(10^{-4})$, significantly smaller than the
$\delta = O(10^{-2})$ unnormalized closest passes of turbulent trajectories to the traveling waves. This suggests that
a much denser set of traveling waves would need to be computed in order to model turbulent trajectories with the linearized
dynamics of the traveling waves they shadow. Alternatively, the radius of accuracy for the linearized dyanmics could
be increased by reducing the time scale significantly below the $T=5$ used in our Newton-hookstep searches, but this would
result in more frequent hops between invariant solutions and, again, a denser set of solutions. 
\subsection{Implications for the study of transitional turbulence}

Our study of the plane Poiseuille symmetry subgroups was motivated by their importance in computing
invariant solutions: for enumerating the different symmetry subgroups in which to seek invariant solutions, 
for specifying which kind of solutions each subgroups allows, and for increasing the efficiency of numerical
searches.  However, knowledge of the symmetry subgroups has application beyond the computation of invariant
solutions. As discussed in section \ref{sec:intro}, the subgroups of $\Gppf$ divide the space of all solutions
of the flow into different invariant symmetric subspaces. The dynamics of transitional or turbulent flow in
each these subspaces can be studied as independent dynamical systems with independent but presumably related
dynamical properties. Small, doubly periodic boxes have been studied as ``minimal flow units'' for plane
Poiseuille and plane Couette flow since \cite{jimenez1991minimal} and \cite{hamilton1995regeneration},
addressing such questions as
(i) What is the minimal domain size and Reynolds number that sustains unsteady flow?
(ii) What dynamical processes sustain unsteady flow? (e.g. Waleffe's self-sustaining process)
(iii) Can the statistics of turbulence in large-aspect-ratio domains be replicated in small periodic boxes?
(iv) What are the minimal perturbations to laminar flow that trigger turbulence?
(v) What are the dynamical pathways by which transitional turbulence collapses to laminar flow?
Each of these questions can be addressed within the context of a specific symmetry subgroup, and the
symmetry-specific answers may well shed light on the dynamics of general unsteady flows, both individually and by comparison.
For example, if we find that the $\langle \sy, \txz \rangle$ subspace requires a much higher Reynolds number
to sustain turbulence than the $\langle \sy, \sz \rangle$ subspace for a fixed domain size, what does this
symmetry-dependent behavior tell us about the the mechanisms by which transitional turbulence is triggered
and kept from relaminarizing? Or, if we find that the statistics of one subspace closely replicate those of
the general flow while another's statistics differ markedly, does it mean that transitional turbulence is
dominated by structures with the former symmetries and not the latter?

The study of invariant solutions is motivated by the hope of addressing such questions precisely from a
dynamical-systems perspective, with well-resolved, time-dependent dynamic processes, rather than time or
ensemble averages and approximations to the equations of motion. From this perspective, we reframe the
above questions in the context of symmetry, as follows. 
{\em For each distinct symmetry subgroup of plane Poiseuille flow, plane Couette flow, or pipe flow}
(vi) What is the minimal Reynolds number that sustains transitional turbulence flow on a given domain?
(vii) What domain supports transitional turbulence at the lowest Reynolds number?
(viii) What is the lowest-Reynolds value of $(L_x, L_z, \Rey)$ for which traveling-waves or equilibria exist?
For periodic or relative periodic orbit solutions? For edge states? 
(ix) What are the physical mechanisms corresponding to the quantitative force balance in the
lowest-Reynolds invariant solutions? I.e., what is the minimal self-sustaining process for each
symmetry subgroup?
(x) Which invariant solutions best replicate the statistics of transitional turbulence with no
symmetries, or of transitional flow with the given symmetries?
(xi) Can the dynamics of transitional turbulence be understood as a chaotic walk among the set of
unstable invariant solutions within the network of their low-dimensional unstable manifolds?

The furthest explorations of these questions have been in the context the $\langle \sz \tx, \sxyz \tz \rangle$-symmetric
subspace of plane Couette flow, which contains the Nagata equilibrium \citep{nagata1990three}, Waleffe's self-sustaining
process \citep{waleffe1997self} and the periodic orbits and state-space visualizations of
\citep{gibson2008visualizing,gibson2009equilibrium}.
In this study, we attempted to find traveling waves for plane Poiseuille flow in a few symmetry subgroups other
than those presented in table \ref{tbl:tws} at similar Reynolds numbers and domains, and found that it was 
difficult to generate sustained transitional turbulence for some subgroups. Perturbations of even very large
magnitudes decayed quickly to laminar flow. The questions above arise from these observations of the strong
influence of symmetry on the basin of attractions of laminar and turbulent flow in small doubly-periodic domains. 

Another line of inquiry stems from the relations between periodic patterns and patches of turbulent
flow within a laminar background. Some doubly-periodic solutions have related spanwise-localized
forms \citep{schneider2010snakes,gibson2014spanwise}. Are there other doubly-periodic solutions in different symmetry
groups that can also be localized? So far there is little success in understanding localization
of periodic solutions through modulation and amplitude equations derived from the Navier-Stokes
equations. Perhaps this problem is more tractable in a symmetry subgroup that has not yet received
much study.

\subsection{Extension to higher-order group elements and tilted domains}
\label{sec:extension_to_tilted_domains}
Our analysis of the subgroups of $\Gppf$ is complete only up to subgroups with second-order elements.
Some examples of subgroups with order-3 and higher elements are presented in section \ref{sec:nonhalfbox_groups},
along with some preliminary principles for their organization. Extending the classification results of this
paper to groups with $n$th-order elements would be a matter of group theory and number theory.
%, and would not require any expertise in fluid dynamics.
Extension to doubly-periodic domains tilted with respect to the
mean pressure gradient or bulk-velocity contraint is also possible. This would especially be applicable in the case of transitional turbulence in large domains, where alternating laminar and turbulent structure align themselves at a particular angle with respect to the streamwise direction (see \cite{tuckerman2011patterns} and \cite{tuckerman2014turbulent}).
In this case $\sz$ is no longer a symmetry
of the equations of motion, so a classification of plane Poiseuille subgroups can be arrived at simply by dropping
subgroups containing $\sz$ as an element or as a factor of an element. In plane Couette flow, one can preserve the
tilt in the $x,z$ plane by inversion in both $x$ and $z$, so the symmetry subgroups of tilted dynamics allow
elements containing the product $\sigma_{xz} = \sx\sz$, but not $\sz$ or $\sx$ individually.

%-----------------------------------%
\section{Conclusion}\label{sec:conclusion}
%-----------------------------------%
This article analyzed subgroups of the symmetry group of plane Poiseuille flow in doubly-periodic domains, focusing mainly
on subgroups generated by spanwise and wall-normal reflections and streamwise and spanwise half-box shifts. We enumerated
the set of subgroups generated by these symmetries and classified them in equivalence classes related by conjugacies.
This procedure is general and can be used to classify symmetry subgroups of other parallel shear flows. We recomputed,
verified, and extended previous analysis of the symmetry subgroups of plane Couette flow (see appendix \ref{sec:appendixB}).
We presented some examples of groups with phase shifts other than half the periodic length and inferred some general
principles of subgroups with higher-order elements.

The interplay between symmetries and dynamics was explained and exploited to find $14$ new traveling waves
in a variety of symmetry subgroups. $\TW_{1-4}$ and $\TW_{7, 8}$ were shown to travel in both streamwise and spanwise
directions, consistent with the imposed symmetries. The spanwise wave-speeds ($c_z$) were found to be smaller by orders
of magnitude than the streamwise wave-speeds ($c_x$). This observation is consistent with the streakiness of 
plane Poseuille flow, in which perturbation velocities are dominated by streaks in the streamwise direction of 
the forcing. To the best of our knowledge, these solutions are first of their kind for plane Poiseuille flow, since all
the previously reported traveling waves for this flow travel only in the streamwise direction. Based on behavior of the
solution and knowledge of symmetries we conjecture that $\text{P}2$ of \cite{park2015exact} has $\sz$ or $\sz \tx$ symmetry
in addition to $\sy$ symmetry. We also computed one traveling wave in a non-half-box symmetry group $\langle \tau(L_x/3, \, L_z/3) \rangle$.

The complete analysis of half-box subgroups and examples of non-half-box symmetry subgroups suggest a rich variety
of physically distinct symmetric subspaces, each with different dynamical behavior for small domains and transitional
Reynolds numbers. We pose a number of questions about the interplay between symmetry and large-amplitude dynamics of
transitional turbulence in plane Poiseuille, plane Couette, and pipe flow, and what kind of minimal self-sustaining
process each symmetry subgroup supports. 

In addition to these directions for future work, we note that the doubling of the convergence rate obtained
by enforcing a second-order symmetry during Krylov iteration will change when the symmetries enforced are third-order
or higher. Understanding the increased convergence rate for higher-order symmetries is nontrivial and will likely be
related to the isotypic decomposition of the space restricted to the symmetry group of interest. Recent work by
\cite{rudge2024crystallographic} explained the usage of crystallographic notation for fluid dynamics and should prove
useful in further analysis of the symmetries of wall--bounded, transitional shear flows.

%-----------------------------------%
\textbf{Declaration of interests:} The authors report no conflict of interest.
%-----------------------------------%
\section*{Acknowledgements}
%-----------------------------------%
This work was partially supported by the National Science Foundation under NSF CBET Award No. 1554149.
Computations were performed on Marvin, a Cray CS500 supercomputer at UNH supported by the NSF MRI program
under grant AGS-1919310.
%-----------------------------------%
%----------------------------------%
\clearpage
\bibliographystyle{jfm}
\bibliography{jfm}
%Use of the above commands will create a bibliography using the .bib file. Shown below is a bibliography built from individual items.
%----------------------------------%
\clearpage
\appendix
%----------------------------------%
% \import{sections/}{supplementary.tex}
%-------------------------------------------
\section{Conjugacy of subgroups of $\Gppfh$.}
\label{sec:appendixA}
%-------------------------------------------

To determine the equivalence classes of subgroups of $\Gppfh = \langle \sy, \sz,\tx,\tz\rangle$ (\ref{def:Gppfh}),
we need to determine if two unequal subgroups $H$ and $H'$ are conjugate, i.e. if $H' = \gamma H \gamma^{-1}$ for
some $\gamma \in \Gppf$. A search for such $\gamma$ over its general form $\gamma = \sy^j \sz^k \tau(a,b)$ would
require a search over $j,k \in \{0,1\}$, $a \in [0, L_x),$ and $b \in [0, L_z)$. However, the following theorem
shows that we need only to check the single value $\gamma = \tau(0, L_z/4)$. 

\begin{theorem}
Let $H, H'$ be unequal but equivalent subgroups of $\Gppfh$. Then $H$ and $H'$ are 
conjugate under $\gamma = \tau(0,L_z/4)$.
\end{theorem}

\begin{proof}
  Let $H', H$ be unequal, equivalent subgroups of $\Gppfh$.
  Then $H' = \gamma H \gamma^{-1}$ for some $\gamma \in \Gppf$, $\gamma \neq 1$. 
  The element $\gamma$ can be expressed as $\gamma = \sy^j \sz^k \tau(a,0) \tau(0,b)$ for 
  some $j,k \in \{0,1\}$ and $a \in [0, L_x), \, b \in [0, L_z)$. 
  Since $H$ is a subgroup of $\Gppfh$, each of its elements can be expressed in form 
  $h = \sy^m \sz^n \tx^p \tz^q$ for some $m, n, p, q$ each either 0 or 1.
  The elements $h'$ of $H'$ are then given by $h' = \gamma h \, \gamma^{-1}$ for $h \in H$.
  Substitution gives
  \begin{align}
    h' &= (\tau(a,0)\, \tau(0,b)\, \sy^j \sz^k)\, (\sy^m\, \sz^n\, \tx^p\, \tz^q)\, (\sz^{-k}\, \sy^{-j} \, \tau(-a,0)\, \tau(0,-b)).
  \end{align}
  Since the elements $\sy, \sz, \tx, \tz,$ and $\tau(\pm a,0)$ commute, we can rearrange these factors
  and simplify with $ \tau(a,0) \tau(-a,0) = \sy^{j} \sy^{-j} = \sz^{k} \sz^{-k} = 1$, yielding
  \begin{align}
    h' &= \tau(0,b)\, \sy^m\, \sz^{n}\, \tx^{p}\, \tz^{q}\, \tau(0,-b). \label{eqn:theorem_two_h_conjugacy}
  \end{align}
  Since $\tau(0,-b)$ commutes with $\sy, \tx,$ and $\tz$, but
  $\tau(0,-b)\, \sz^{n} = \tau(0, -(-1)^{n}b) = \tau(0, (-1)^{n+1}b)$, we have
  \begin{align}
    h' &= \tau(0,(1+(-1)^{n+1})b) \, \sy^{m}\, \sz^{n}\, \tx^{p}\, \tz^{q}.
  \end{align}
  There must be at least one $h$ in $H$ for which $\tau(0,(1+(-1)^{n+1})b)$
  is not the identity, since the contrary assumption leads to 
  $h' = \sy^m\, \sz^n\, \tx^p\, \tz^q\ = h$ for all $h$ in $H$,
  and thus $H' = H$, which contradicts the assumption $H' \neq H$.

  Thus we have $n_j = 1$ for some $h_j$ in $H$, and $h'_j$ has the form
  \begin{align}
    h'_j = \tau(0, 2b)\, \sy^{m_j}\, \sz^{n_j}\, \tx^{p_j}\, \tz^{q_j}, \label{eqn:theorem_two_hj}
  \end{align}
  where $\tau(0,2b) \neq 1$. Since $h'_j$ is an element of a subgroup of $\Gppfh$, it is a product
  of factors $\sy, \sz, \tx,$ and $\tz.$ It follows then from \ref{eqn:theorem_two_hj} that 
  $\tau(0,2b) = \tz = \tau(0, L_z/2)$ and thus $2b = L_z/2$.  The only values of $b$ in the periodic
  domain $[0, L_z)$ that satisfy this equality are $b=L_z/4$ and $b=3L_z/4$.

  If $b=L_z/4$ then substituting this value and $h = \sy^m\, \sz^n\, \tx^p\, \tz^q$
  into \ref{eqn:theorem_two_h_conjugacy} gives $h' = \tau(0, L_z/4)\, h\, \tau(0, -L_z/4)$ for
  all $h$ in $H$, and thus $H'$ is conjugate to $H$ under $\gamma = \tau(0,L_z/4)$.
  If $b=3L_z/4$, then we write $\tau(0, \pm 3L_z/4) = \tau(0, \pm L_z/4) \tau(0, \pm L_z/2) = \tau(0, \pm L_z/4) \tz$
  and substitute into \ref{eqn:theorem_two_h_conjugacy}, giving
  \begin{align}
    h' &= \tau(0,L_z/4)\, \tz\, \sy^m\, \sz^n\, \tx^p\, \tz^q\, \tau(0,-L_z/4)\, \tz.
  \end{align}
  Since $\tz$ commutes with $\sy, \sz, \tx, \tz,$ and $\tau(0, \pm L_z/4)$, this can be reduced to
  $h' = \tau(0,L_z/4)\, \sy^m\, \sz^n\, \tx^p\, \tz^q\, \tau(0,-L_z/4) = 
  \tau(0, L_z/4)\, h\, \tau(0,-L_z/4)$, for all $h$ in $H$. Thus $H'$ is conjugate to $H$ under 
  $\gamma = \tau(0,L_z/4)$. 
\end{proof}

%-------------------------------------------
\section{Symmetry subgroups and equivalence classes for plane Couette flow}
\label{sec:appendixB}
%-------------------------------------------
To verify the symbolic code, we calculated the equivalence classes of symmetry subgroups for
plane Couette flow and compared them with \cite{gibson2009equilibrium}. The results,
presented in table \ref{tab:pcf_subgroups}, are for symmetries acting on the fluctuating velocity
$\bu(x,t)$ of plane Couette flow after a decomposition of total velocity and pressure into a
laminar base flow and fluctuation $\bu(x,t)$ following \cite{gibson2009equilibrium},
namely $\utot(\bx,t) = y \be_x + \bu(\bx, t)$. This results in Dirichlet boundary conditions on 
$\bu$ at the walls $y = \pm 1$ and renders the space of $\bu$ a vector space. The plane Couette
symmetry group $\Gpcf$ corresponding to $\Gppf$ (\ref{def:Gppf}) for plane Poiseuille flow is 
\begin{align} \label{def:Gpcf}
  \Gpcf = \langle \sxy,\, \sz,\, \{\tau(\lx, \lz) \; : \; a,b \in \bbR \} \rangle,
\end{align}
and the 16th-order abelian subgroup of $\Gpcf$ limited to $z$-centered, half-box phase shifts
corresponding to $\Gppfh$ (\ref{def:Gppfh}) is
\begin{align} \label{def:Gpcfh}
  \Gpcfh = \langle \sxy,\, \sz,\, \tx, \tz \rangle.
\end{align}

%-----------------------------------%
\begin{table}
  %\begin{tabular}{p{0.49\textwidth} p{0.49\textwidth} }
  \begin{tabular}{ll}
     order $1$ : & 1 subgroup in 1 equivalence class\\
     & $ \langle 1 \rangle$ \\ 
     \\ order $2$ : & 12 subgroups in 6 equivalence classes\\
     & $ \langle \txz \rangle$ \\
     & $\langle \sz \rangle \sim \langle \sz \tz \rangle$ \\
     & $\langle \sxy \rangle \sim \langle \sxy \tx \rangle$ \\
     & $\langle \sz \tx \rangle \sim \langle \sz \txz \rangle$ \\
     & $\langle \sxy \tz \rangle \sim \langle \sxy \txz \rangle$ \\
     & $\langle \sxyz \rangle \sim \langle \sxyz \tx \rangle \sim \langle \sxyz \tz \rangle$ \\    %
     \\ order $4$: & 23 subgroups in 7 equivalence classes \\
     %
%    & $\langle \txz, \sz \rangle \sim \langle \txz, \sz \tz \rangle$  \\
     & $\langle \sz,\, \txz \rangle \sim \langle \sz \tz, \txz \rangle$  \\
%    & $\langle \txz, \sxy \rangle \sim \langle \txz, \sxy \tz \rangle$ \\
     & $\langle \sxy,\, \txz \rangle \sim \langle \sxy \tx,\, \txz \rangle$ \\
%    & $\langle \txz, \sxyz \rangle \sim \langle \txz , \sxyz \tx \rangle \sim \langle \txz , \sxyz \tz \rangle$  \\
     & $\langle \sxyz,\, \txz \rangle \sim \langle \sxyz \tx,\, \txz \rangle \sim \langle \sxyz \tz,\, \txz  \rangle$  \\
%    & $\langle \sz, \sxy \rangle \sim \langle \sz, \sxy \tx \rangle \sim \langle \sz \tz, \sxy \rangle \sim \langle \sz\tz, \sxy \tx \rangle $ \\
     & $\langle \sxy,\, \sz \rangle \sim \langle \sxy \tx,\, \sz \rangle \sim \langle \sxy,\, \sz \tz \rangle \sim \langle \sxy \tx,\, \sz\tz \rangle $ \\
%    & $\langle \sz , \sxy \tz \rangle \sim \langle \sz, \sxy \txz \rangle \sim \langle \sz \tz, \sxy \tz \rangle \sim \langle \sz \tz, \sxy \txz \rangle$ \\
     & $\langle \sxy \tz,\, \sz \rangle \sim \langle \sxy \txz,\, \sz \rangle \sim \langle \sxy \tz,\, \sz \tz \rangle \sim \langle \sxy \txz,\, \sz \tz \rangle$ \\
%    & $\langle \sz\tx, \sxy \rangle \sim \langle \sz\tx, \sxy\tx \rangle \sim \langle \sz\txz,  \sxy \rangle \sim \langle \sz\txz,  \sxy\tx \rangle$ \\
     & $\langle \sxy,\, \sz\tx \rangle \sim \langle \sxy \tx,\, \sz\tx \rangle \sim \langle \sxy,\, \sz\txz \rangle \sim \langle \sxy\tx,\, \sz\txz \rangle$ \\
%    & $\langle \sz\tx,  \sxy\tz \rangle \sim  \langle \sz\tx,  \sxy\txz \rangle \sim \langle \sz\txz,  \sxy\tz \rangle \sim \langle \sz\txz,  \sxy\txz \rangle$\\
     & $\langle \sxy\tz,\, \sz\tx \rangle \sim  \langle \sxy\txz,\, \sz\tx \rangle \sim \langle \sxy\tz,\, \sz\txz \rangle \sim \langle \sxy\txz,\, \sz\txz \rangle$\\
     \\ order $8$: & 4 subgroups in 1 equivalence class\\
%    & $\langle \txz, \sz, \sxy \rangle    \sim \langle \txz, \sz, \sxy \tx \rangle   \sim \langle \txz, \sz \tz, \sxy \rangle   \sim \langle \txz, \sz \tz, \sxy \tx \rangle$  \\
     & $\langle \sxy,\, \sz,\ \txz \rangle \sim \langle \sxy\tx,\ \sz,\, \txz \rangle \sim \langle \sxy,\, \sz\tz,\, \txz \rangle \sim \langle \sxy\tx,\, \sz\tz,\, \txz \rangle$  \\
   \end{tabular}
  \caption{\label{tab:pcf_subgroups} 
    {\bf Symmetry subgroups and their equivalence classes for plane Couette flow in minimal doubly-periodic domains.}
    The symbol $\sim$ denotes equivalence between subgroups under conjugation. 
    A set of subgroups related by $\sim$ symbols forms an equivalence class.
    A subgroup listed in isolation forms an equivalence class with that subgroup as its sole element.
    The first subgroup listed in each equivalence class is chosen as the representative for that equivalence class. 
    The subgroups listed are the complete set of $z$-centered, minimal-domain subgroups with elements of order 2 or less
    for plane Couette flow in doubly-periodic domains. 
    See sections \ref{sec:halfbox_groups} for related discussion.}
 \end{table}

Table \ref{tab:pcf_subgroups} lists the subgroups of  $\Gpcfh$ and their equivalence classes. 
The inversion of the $x$ coordinate by $\sxy$ in $\Gpcf$ and $\Gpcfh$, compared to $\sy$ in $\Gppf$ and
$\Gppfh$, means that plane Couette symmetries have the noncommuting product $\sxy \tau(a,b) = \tau(-a,b) \sxy$ 
in addition to $\sz \tau(a,b) = \tau(a,-b) \sz$. This additional instance of non-abelian behavior expands the
set of elements of $\gamma \in \Gpcf$ for which conjugacies $H' = \gamma H \gamma^{-1}$ are nontrivial and
thereby reduces the total number of equivalence classes from $24$ for plane Poiseuille flow to $14$ for 
plane Couette flow.  The results agree with \cite{gibson2009equilibrium} and expand upon them, since
\cite{gibson2009equilibrium} excluded subgroups with isolated $\txz$ symmetry. Note that \cite{gibson2009equilibrium}
used notation ``$\sx$'' for $\sxy$. \cite{daly2014secondary} showed that rotating plane Couette flow, with
a uniform rotation about the spanwise direction, obeys the same symmetries. 
%----------------------------------%

\end{document}